\PassOptionsToPackage{unicode}{hyperref}
\PassOptionsToPackage{hyphens}{url}
\PassOptionsToPackage{dvipsnames,svgnames,x11names}{xcolor}
\documentclass[
  11pt,
  a4paper,
]{article}
\usepackage{xcolor}
\usepackage[a4paper,textheight=24cm,textwidth=15.5cm]{geometry}
\usepackage{amsmath,amssymb}
\usepackage{iftex}
\ifPDFTeX
  \usepackage[T1]{fontenc}
  \usepackage[utf8]{inputenc}
  \usepackage{textcomp} 
\else 
  \usepackage{unicode-math} 
  \defaultfontfeatures{Scale=MatchLowercase}
  \defaultfontfeatures[\rmfamily]{Ligatures=TeX,Scale=1}
\fi
\usepackage[]{libertinus}
\ifPDFTeX\else
\fi
\IfFileExists{upquote.sty}{\usepackage{upquote}}{}
\IfFileExists{microtype.sty}{
  \usepackage[]{microtype}
  \UseMicrotypeSet[protrusion]{basicmath} 
}{}
\makeatletter
\@ifundefined{KOMAClassName}{
  \IfFileExists{parskip.sty}{%
    \usepackage{parskip}
  }{
    \setlength{\parindent}{0pt}
    \setlength{\parskip}{6pt plus 2pt minus 1pt}}
}{
  \KOMAoptions{parskip=half}}
\makeatother
\makeatletter
\ifx\paragraph\undefined\else
  \let\oldparagraph\paragraph
  \renewcommand{\paragraph}{
    \@ifstar
      \xxxParagraphStar
      \xxxParagraphNoStar
  }
  \newcommand{\xxxParagraphStar}[1]{\oldparagraph*{#1}\mbox{}}
  \newcommand{\xxxParagraphNoStar}[1]{\oldparagraph{#1}\mbox{}}
\fi
\ifx\subparagraph\undefined\else
  \let\oldsubparagraph\subparagraph
  \renewcommand{\subparagraph}{
    \@ifstar
      \xxxSubParagraphStar
      \xxxSubParagraphNoStar
  }
  \newcommand{\xxxSubParagraphStar}[1]{\oldsubparagraph*{#1}\mbox{}}
  \newcommand{\xxxSubParagraphNoStar}[1]{\oldsubparagraph{#1}\mbox{}}
\fi
\makeatother

\usepackage{longtable,booktabs,array}
\usepackage{calc} 
\usepackage{etoolbox}
\makeatletter
\patchcmd\longtable{\par}{\if@noskipsec\mbox{}\fi\par}{}{}
\makeatother
\IfFileExists{footnotehyper.sty}{\usepackage{footnotehyper}}{\usepackage{footnote}}
\makesavenoteenv{longtable}
\usepackage{graphicx}
\makeatletter
\newsavebox\pandoc@box
\newcommand*\pandocbounded[1]{
  \sbox\pandoc@box{#1}%
  \Gscale@div\@tempa{\textheight}{\dimexpr\ht\pandoc@box+\dp\pandoc@box\relax}%
  \Gscale@div\@tempb{\linewidth}{\wd\pandoc@box}%
  \ifdim\@tempb\p@<\@tempa\p@\let\@tempa\@tempb\fi
  \ifdim\@tempa\p@<\p@\scalebox{\@tempa}{\usebox\pandoc@box}%
  \else\usebox{\pandoc@box}%
  \fi%
}
\def\fps@figure{htbp}
\makeatother

\NewDocumentCommand\citeproctext{}{}

\makeatletter
 \let\@cite@ofmt\@firstofone
 \def\@biblabel#1{}
 \def\@cite#1#2{{#1\if@tempswa , #2\fi}}
\makeatother
\newlength{\cslhangindent}
\newlength{\csllabelwidth}
\newenvironment{CSLReferences}[2] 
 {\begin{list}{}{%
  \setlength{\itemindent}{0pt}
  \setlength{\leftmargin}{0pt}
  \setlength{\parsep}{0pt}
  \ifodd #1
   \setlength{\leftmargin}{\cslhangindent}
   \setlength{\itemindent}{-1\cslhangindent}
  \fi
  \setlength{\itemsep}{#2\baselineskip}}}
 {\end{list}}
\usepackage{calc}

\providecommand{\tightlist}{%
  \setlength{\itemsep}{0pt}\setlength{\parskip}{0pt}}

\usepackage{amsmath}
\usepackage{tikz}
\usepackage{xcolor}
\usepackage{tabularx}
\usepackage{orcidlink}
\usepackage{lineno}
\usepackage{unicode-math}
\usepackage{draftwatermark}
\usepackage{adjustbox}
\usepackage{float}
\makeatletter
\@ifpackageloaded{caption}{}{\usepackage{caption}}
\AtBeginDocument{%
\ifdefined\contentsname
  \renewcommand*\contentsname{Table of contents}
\else
  \newcommand\contentsname{Table of contents}
\fi
\ifdefined\listfigurename
  \renewcommand*\listfigurename{List of Figures}
\else
  \newcommand\listfigurename{List of Figures}
\fi
\ifdefined\listtablename
  \renewcommand*\listtablename{List of Tables}
\else
  \newcommand\listtablename{List of Tables}
\fi
\ifdefined\figurename
  \renewcommand*\figurename{Figure}
\else
  \newcommand\figurename{Figure}
\fi
\ifdefined\tablename
  \renewcommand*\tablename{Table}
\else
  \newcommand\tablename{Table}
\fi
}
\@ifpackageloaded{float}{}{\usepackage{float}}
\floatstyle{ruled}
\@ifundefined{c@chapter}{\newfloat{codelisting}{h}{lop}}{\newfloat{codelisting}{h}{lop}[chapter]}
\floatname{codelisting}{Listing}

\usepackage{amsthm}
\theoremstyle{plain}
\newtheorem{proposition}{Proposition}[section]
\theoremstyle{remark}
\AtBeginDocument{}

\makeatother
\makeatletter
\@ifpackageloaded{caption}{}{\usepackage{caption}}
\@ifpackageloaded{subcaption}{}{\usepackage{subcaption}}
\makeatother
\makeatletter
\@ifpackageloaded{algorithm}{}{\usepackage{algorithm}}
\makeatother
\makeatletter
\@ifpackageloaded{algpseudocode}{}{\usepackage{algpseudocode}}
\makeatother
\makeatletter
\@ifpackageloaded{caption}{}{\usepackage{caption}}
\makeatother
    \makeatletter
    \newcommand\fs@nocaption{
      \def\@fs@cfont{\bfseries}
      \let\@fs@capt\floatc@plain
      \def\@fs@pre{}%
      \def\@fs@post{\kern2pt\hrule}%
      \def\@fs@mid{\hrule\kern2pt}%
      \let\@fs@iftopcapt\iftrue}
    \makeatother
  
\usepackage{bookmark}
\IfFileExists{xurl.sty}{\usepackage{xurl}}{} 
\makeatletter
\@ifundefined{xmpquote}{\newcommand{\xmpquote}[1]{#1}}{}
\makeatother
\hypersetup{
  pdftitle={How to Train Your Intercept},
  pdfauthor={Johan Larsson; Frederik Fabricius Bjerre},
  pdfkeywords={\xmpquote{intercepts}, \xmpquote{coordinate
descent}, \xmpquote{optimization}, \xmpquote{regularization}},
  colorlinks=true,
  linkcolor={blue},
  filecolor={Maroon},
  citecolor={Blue},
  urlcolor={Blue},
  pdfcreator={LaTeX via pandoc}}

\ifXeTeX

\fi
\ifPDFTeX
\pdftrailerid{}
\pdftrailer{/ID [ <d461e0a7c4fe00a992bf5524f70fd92d> <d461e0a7c4fe00a992bf5524f70fd92d> ]}
\fi
\ifLuaTeX
\pdfvariable trailerid {[ <d461e0a7c4fe00a992bf5524f70fd92d> <d461e0a7c4fe00a992bf5524f70fd92d> ]}
\fi

\title{How to Train Your Intercept}
\usepackage{etoolbox}
\makeatletter
\providecommand{\subtitle}[1]{
  \apptocmd{\@title}{\par {\large #1 \par}}{}{}
}
\makeatother
\subtitle{A Comparison of Intercept Update Strategies in Coordinate
Descent}
\author{Johan Larsson \and Frederik Fabricius Bjerre}
\date{2026-10-13}
\begin{document}
\definecolor{computo-blue}{HTML}{034E79}

\begin{center}
  {\Large\bfseries How to Train Your Intercept}
  \\[.3cm]{\large A Comparison of Intercept Update Strategies in
Coordinate Descent}
\end{center}

\bigskip
\begin{center}
      \textbf{Johan
Larsson}\textsuperscript{1}~\orcidlink{0000-0002-4029-5945}{\renewcommand{\thefootnote}{\fnsymbol{footnote}}\footnote{Corresponding author: \href{mailto:johan@jolars.co}{johan@jolars.co}}}\quad
      \textbf{Frederik Fabricius
Bjerre}\textsuperscript{1}~\orcidlink{0009-0004-1079-1616}\quad
  
  \smallskip
  {\small
      \textsuperscript{1}Department of Mathematical Sciences, University
of Copenhagen\quad
    }

  \bigskip

  Date: 2026-10-13 \quad Last modified: 2026-10-13
  \end{center}

\bigskip
\begin{abstract}
Coordinate descent solvers for regularized generalized linear models
(GLMs) differ in how they update the intercept. This seemingly routine
choice can substantially affect convergence when the response is
imbalanced. For direct coordinate descent on the original GLM loss, we
compare a conservative step based on worst-case curvature, a safeguarded
Newton step based on current curvature, and repeated Newton steps that
solve the intercept subproblem. The conservative update can leave much
of the interaction between the intercept and the coefficients
unresolved, producing a slowdown when that interaction controls
convergence. A single Newton step avoids this attenuation and, in our
experiments, follows nearly the same outer trajectory as exact
minimization while requiring fewer inner steps. Simulated and real-data
experiments, including comparisons across six production solvers and a
controlled intervention within skglm, support this explanation. We also
show how the same mechanism appears in quadratic-approximation methods
and extends to multinomial models with rare classes. We recommend a
safeguarded Newton update for direct coordinate descent and full
intercept optimization within frozen local quadratic surrogates.
\end{abstract}

\noindent%
\vfill

\DraftwatermarkOptions{stamp=false}

\floatname{algorithm}{Algorithm}

\renewcommand*\contentsname{Contents}
{
\setcounter{tocdepth}{3}
\tableofcontents
}

\section{Introduction}\label{introduction}

Popular solvers for regularized generalized linear models~(GLMs) agree
on coordinate descent~(CD) but disagree about one coordinate: the
intercept. GLMs usually include this coordinate to represent the
baseline response and to absorb constant shifts in the predictors
without changing the fitted slopes. It is unpenalized and never drops
out of the fitted model, yet solvers update it in different ways.
Research on coordinate descent has examined acceleration,
parallelization, and convergence rates, but has largely treated the
intercept update as a routine implementation choice. That choice can
dominate convergence speed when the response is imbalanced. For direct
coordinate descent, a single safeguarded Newton step per pass removes
this slowdown in our experiments.

The disagreement follows a deeper split among implementations. Some
solvers, including \textbf{skglm}\footnote{Among the surveyed solvers,
  only \textbf{skglm}'s intercept update is version-dependent: version
  0.5 and earlier take a conservative gradient step, while a later patch
  replaces it with a Newton step (Section~\ref{sec-production-solvers}).
  We pin the version wherever this behavior matters.}~(Bertrand et al.
2022), run coordinate descent directly on the original GLM loss. Others
repeatedly replace that loss with a quadratic approximation and solve
the approximation by coordinate descent. A local quadratic already
contains the current intercept curvature, so the update strategies
coincide within it. A quadratic built from a worst-case curvature bound
retains the same problem as the conservative direct update. The six
production solvers in Section~\ref{sec-production-solvers} use these
designs. We explain in Section~\ref{sec-theory} and
Section~\ref{sec-irls} when the distinction matters.

We study the optimization problem
\begin{equation}\protect\phantomsection\label{eq-primal-problem}{
  \mathop{\mathrm{minimize}}_{\beta_0 \in \mathbb{R},\beta \in \mathbb{R}^p} \Big( P(\beta_0,\beta) = F(\beta_0, \beta) + G_{\lambda}(\beta) \Big)
}\end{equation} where \(P\) is the primal objective, \(F\) is the loss,
and \(G_\lambda\) penalizes large coefficients while leaving the
intercept unpenalized. Throughout the paper, we assume that
\(G_\lambda\) is a proper, closed, convex, separable penalty and that
\(F\) is the loss averaged over observations, \[
  F(\beta_0, \beta) = \frac{1}{n}\sum_{i=1}^n f(x_i^\top \beta + \beta_0; y_i),
\] where \(f\) is smooth, twice differentiable, and convex. Averaging by
\(n\) keeps the global Lipschitz constant for the intercept,
\(L_0 = \sup f''\), independent of \(n\). It also matches the objective
scaling of the production solvers that we compare in
Section~\ref{sec-production-solvers}:
\((1/n) \cdot \text{loss} + \lambda \|\beta\|_1\). For brevity, we write
\(x_i\) for the \(i\)th row of the design matrix \(X\) and \(x_j\) for
its \(j\)th column, and we let \((\hat{\beta}_0, \hat{\beta})\) denote a
solution to Equation~\ref{eq-primal-problem}. Under these assumptions,
\(P\) is a convex composite objective. The assumptions cover every
penalty and model in our analysis and experiments. We do not consider
nonconvex penalties.

We compare three strategies for updating the intercept once per pass
over the coefficient coordinates:

\begin{description}
\tightlist
\item[Gradient Strategy]
Takes a single gradient step based on a global curvature bound.
\item[Newton Strategy]
Takes a single Newton step based on the curvature at the current
iterate, with an Armijo backtracking line search to guard against
overshooting.
\item[Exact Strategy]
Minimizes the loss over the intercept while holding the coefficients
fixed.
\end{description}

In the analysis, \emph{exact} means setting the intercept to its
conditional minimizer. The experiments compute this minimizer
numerically, with tolerances documented in
\hyperref[sec-s-methodology]{Supplement S3}.

These strategies can converge at substantially different rates when the
response is imbalanced---for example, when positive cases are rare in
logistic regression. The three paths in Figure~\ref{fig-parametric}
trace the strategies across the level sets of
Equation~\ref{eq-primal-problem} for a toy problem with one feature and
an intercept. The gradient strategy approaches the solution through many
small steps, whereas the Newton and exact strategies reach it in only a
few.

\begin{figure}

\centering{

\pandocbounded{\includegraphics[keepaspectratio]{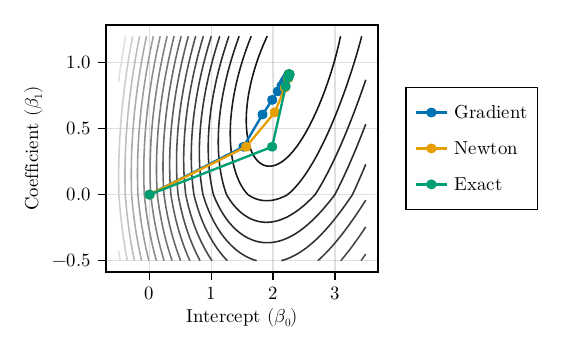}}

}

\caption{\label{fig-parametric}Intercept-update trajectories on a toy
logistic problem. The paths show the iterates of the gradient, Newton,
and exact strategies over the objective level sets for a one-feature
model with an intercept.}

\end{figure}%

In Figure~\ref{fig-first-example}, we compare how quickly the three
strategies converge when fitting \(\ell_1\)-regularized logistic
regression to the w1a dataset~(Platt 1998). The Newton and exact
strategies bring the relative suboptimality bound below \(10^{-8}\)
after roughly 45 passes, while the gradient strategy takes more than
seven times as many.

\begin{figure}

\centering{

\pandocbounded{\includegraphics[keepaspectratio]{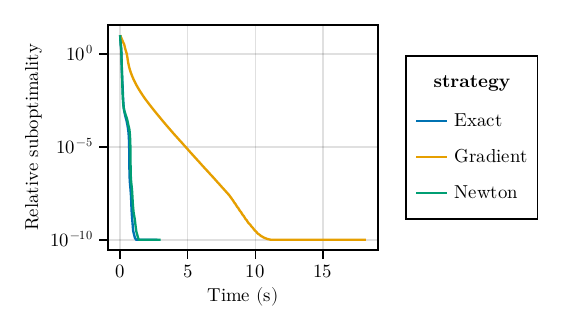}}

}

\caption{\label{fig-first-example}Overall convergence on w1a. The plot
shows a relative suboptimality bound versus time for the three intercept
strategies in \(\ell_1\)-regularized logistic regression.}

\end{figure}%

We make four contributions:

\begin{itemize}
\item
  We formalize three strategies for updating an unpenalized intercept:
  gradient, Newton, and exact minimization.
\item
  We use a Schur complement to describe how an intercept update changes
  the coefficient gradients. A frozen quadratic model then explains why,
  under strong coupling, the gradient-to-Newton pass-count ratio
  approaches the ratio of global to local intercept curvature.
\item
  We classify CD-based GLM solvers by their treatment of curvature and
  the intercept, and we evaluate six production solvers using controlled
  comparisons where possible.
\item
  We extend the analysis to multinomial logistic regression, where rare
  classes create low-curvature directions. Experiments on synthetic and
  real data show that Newton updates closely match exact minimization
  and can be much faster than gradient updates.
\end{itemize}

\section{Related work}\label{sec-related-work}

The choice of intercept update touches three strands of the literature,
none of which has analyzed it as a separate problem. The first concerns
how CD-based GLM solvers incorporate the intercept. Existing work
generally treats this as a routine implementation choice. The
\textbf{glmnet} literature~(Friedman et al. 2010; Tay et al. 2023), for
example, describes two schemes for weighting the inner logistic solver
in iteratively reweighted least squares (IRLS): local weights and a
global 1/4 upper bound. The software exposes them through
\texttt{type.logistic\ =\ "Newton"} and \texttt{"modified.Newton"},
which the documentation presents as a speed trade-off rather than as a
choice of intercept strategy.

\textbf{LIBLINEAR}'s newGLMNET solver~(Fan et al. 2008; Yuan et al.
2012) treats the intercept as the unregularized last coordinate of each
prox-Newton subproblem. Coordinate descent then optimizes the intercept
along with the coefficients. By the end of the inner solve, the
intercept has reached its conditional minimum for the quadratic. The
solver has no separate intercept loop on the original loss.\footnote{Source:
  \url{https://github.com/cjlin1/liblinear/blob/491c9f1/linear.cpp\#L1972-L1977}.
  \texttt{solve\_l1r\_lr} (\texttt{linear.cpp:1796}) appends the
  intercept as a constant feature and excludes it from regularization
  (\texttt{linear.cpp:1789-1790}). Inside the inner CD solve of the
  quadratic program, the intercept update is the pure Newton step
  \texttt{z\ =\ -G/H} (\texttt{linear.cpp:1972-1977}). The outer Newton
  loop ends without any separate intercept loop on the original loss.}

\textbf{BlitzL1}~(Johnson and Guestrin 2015) does the same inside the
prox-Newton inner CD solve. After each subproblem, it also runs a capped
one-dimensional Newton loop on the original loss to bring the intercept
gradient within tolerance.\footnote{Source:
  \url{https://github.com/tbjohns/BlitzL1/blob/aef8d02/src/solver.cpp\#L27-L55}.
  The loop is called from \texttt{run\_prox\_newton\_iteration}
  (\texttt{solver.cpp:262}) after the working-set CD pass and step-size
  backtracking.} \textbf{celer}~(Massias et al. 2020) does not support
intercepts for \(\ell_1\)-logistic regression. Implementations therefore
make varied choices, but previous work has not directly analyzed the
response distributions for which those choices matter.

A second thread is the Frisch--Waugh--Lovell tradition of partialling
out. In linear regression, the intercept can be handled separately:
centering the features decouples it from the slopes and, when the
intercept is left unpenalized, gives \(\hat\beta_0 = \bar y\)~(Lovell
1963). \textbf{glmnet} leaves the intercept unpenalized on
translation-invariance grounds---penalizing \(\beta_0\) would make the
path depend on the arbitrary location of \(y\)---and this decoupling
makes that choice essentially free under Gaussian loss. After solving
for the slopes, the intercept is recovered as
\(\hat\beta_0 = \bar y - \bar x^\top \hat\beta\).

For non-Gaussian GLMs, the decoupling fails because curvature varies
across observations. Ordinary feature centering no longer removes the
interaction between the intercept and the slopes, and the closed form
for \(\hat\beta_0\) disappears. \textbf{adelie}~(Yang and Hastie 2024)
handles this dependence through its proximal quasi-Newton approximation.
At each linearization, it builds a weighted Gaussian surrogate and
absorbs the intercept into that surrogate rather than treating it as a
separate CD coordinate.\footnote{Source:
  \url{https://github.com/JamesYang007/adelie/blob/d6a1229/adelie/src/include/adelie_core/solver/solver_gaussian_pin_naive.hpp}.
  The inner Gaussian solver uses IRLS-weighted column means in its
  gradient update (\texttt{solver\_gaussian\_pin\_naive.hpp:87,\ 120})
  and recovers the intercept at the end of the inner solve as the
  weighted mean of the working response plus the residual mean
  (\texttt{solver\_gaussian\_pin\_naive.hpp:392}). The outer GLM solver
  computes the column means using the IRLS weights
  (\texttt{solver\_glm\_naive.hpp:366}).} For Gaussian data, the
framework in Section~\ref{sec-theory} reduces to the familiar result:
the loss has constant curvature, so the gradient and Newton strategies
coincide on the intercept.

The third thread is majorization--minimization~(MM). The binary logistic
curvature bound \(L_0 = 1/4\) is the binary specialization of the
quadratic Hessian majorant used by Böhning (1992) to build an MM
algorithm for multinomial logistic regression. Hunter and Lange (2004)
give the canonical exposition of MM. Krishnapuram et al. (2005) apply
Böhning's bound to L1-penalized multinomial logistic regression, a
direct precedent for our multinomial extension.

In regularized logistic regression, the same construction appears in the
\texttt{modified.Newton} mode of \textbf{glmnet} and the \texttt{MM}
mode of \textbf{biglasso}~(Zeng and Breheny 2021)
(Table~\ref{tbl-classification}). Because the global majorant does not
adapt to local curvature, it can produce a much smaller intercept
correction than a Newton step when the response is imbalanced. In
multinomial regression, the same mismatch appears along directions
associated with rare classes. The MM literature treats this as loose
majorization. We quantify the resulting slowdown for imbalanced binary
and rare-class multinomial problems (Section~\ref{sec-irls},
Figure~\ref{fig-multinomial-sweep}).

\section{Theory}\label{sec-theory}

The three strategies differ in how they respond to intercept curvature.
The gradient strategy uses a global upper bound, Newton uses the
curvature at the current iterate, and the exact strategy minimizes the
intercept subproblem. On imbalanced data, the local curvature can be
much smaller than the global bound. When the intercept---or a direction
coupling it with a coefficient---controls convergence, this mismatch can
make the gradient strategy slower by a factor proportional to the ratio
of global to local curvature. We first analyze the effect of one
intercept update, then follow its effect through repeated passes on a
fixed quadratic model and compare the resulting local rates with
experiments.

Throughout this section, we work with the canonical-link GLM loss
averaged over observations, \[
  F(\beta_0, \beta) = \frac{1}{n}\sum_{i=1}^n f(x_i^\top \beta + \beta_0; y_i).
\] We first derive the scalar-intercept results for the Gaussian,
binomial, and Poisson models. We later extend them to multinomial
logistic regression, whose intercept is a vector. For observation \(i\),
let \(\eta_i = x_i^\top \beta + \beta_0\) denote the linear predictor,
and let \(g^{-1}\) denote the inverse link. The conditional mean is then
\(\mathop{\mathrm{E{}}}(y_i \mid \eta_i) = g^{-1}(\eta_i)\). For the
canonical-link models considered here, the derivative of the negative
log-likelihood simplifies to the \emph{generalized residual}, \[
  f'(\eta_i; y_i) = g^{-1}(\eta_i) - y_i = r_i.
\] All gradients below follow from this canonical-link identity, and all
Hessians follow from its derivative. The identity does not hold for
arbitrary noncanonical links. The loss, link, and inverse link for the
four models we study---Gaussian, binomial, Poisson, and
multinomial---appear in Table~\ref{tbl-glm}.

\begin{longtable}[]{@{}
  >{\raggedright\arraybackslash}p{(\linewidth - 6\tabcolsep) * \real{0.1731}}
  >{\raggedright\arraybackslash}p{(\linewidth - 6\tabcolsep) * \real{0.4551}}
  >{\raggedright\arraybackslash}p{(\linewidth - 6\tabcolsep) * \real{0.2051}}
  >{\raggedright\arraybackslash}p{(\linewidth - 6\tabcolsep) * \real{0.1667}}@{}}
\caption{Loss, link, and inverse-link functions for the GLM families
studied. The multinomial row uses the reference-class parameterization
with class \(m\) as the reference (\(\eta_m = 0\),
\(\mu_m = 1 - \sum_{j<m}\mu_j\)), so \(\eta\) and \(\mu\) in that row
are \((m-1)\)-vectors, and \(\log\) and \(\exp\) are applied
elementwise.}\label{tbl-glm}\tabularnewline
\toprule\noalign{}
\begin{minipage}[b]{\linewidth}\raggedright
Model
\end{minipage} & \begin{minipage}[b]{\linewidth}\raggedright
\(f(\eta; y)\)
\end{minipage} & \begin{minipage}[b]{\linewidth}\raggedright
\(g(\mu)\)
\end{minipage} & \begin{minipage}[b]{\linewidth}\raggedright
\(g^{-1}(\eta)\)
\end{minipage} \\
\midrule\noalign{}
\endfirsthead
\toprule\noalign{}
\begin{minipage}[b]{\linewidth}\raggedright
Model
\end{minipage} & \begin{minipage}[b]{\linewidth}\raggedright
\(f(\eta; y)\)
\end{minipage} & \begin{minipage}[b]{\linewidth}\raggedright
\(g(\mu)\)
\end{minipage} & \begin{minipage}[b]{\linewidth}\raggedright
\(g^{-1}(\eta)\)
\end{minipage} \\
\midrule\noalign{}
\endhead
\bottomrule\noalign{}
\endlastfoot
Gaussian & \(\frac{1}{2}(y - \eta)^2\) & \(\mu\) & \(\eta\) \\
Binomial & \(\log(1 + e^\eta) - \eta y\) &
\(\log\left( \frac{\mu}{1 - \mu} \right)\) &
\(\frac{e^\eta}{1 + e^\eta}\) \\
Poisson & \(e^\eta - \eta y\) & \(\log(\mu)\) & \(e^\eta\) \\
Multinomial &
\(\log\left( 1 + \sum_{j=1}^{m-1} e^{\eta_j} \right) - \sum_{k=1}^{m-1} y_k \eta_k\)
& \(\log\left( \frac{\mu}{1 - \sum_{j=1}^{m-1} \mu_j} \right)\) &
\(\frac{\exp(\eta)}{1 + \sum_{j=1}^{m-1} e^{\eta_j}}\) \\
\end{longtable}

\subsection{Coordinate descent}\label{coordinate-descent}

Coordinate descent updates one coordinate at a time while holding the
others fixed. For coefficient \(\beta_j\), the chain rule gives \[
  \frac{\partial}{\partial \beta_j} F(\beta_0, \beta)
  = \frac{1}{n}\sum_{i=1}^n x_{ij} f'(\eta_i; y_i).
\]

The smooth part of a coefficient or intercept update takes a step in the
negative partial-gradient direction, \[
  \beta_j^+ = \beta_j
              - \gamma_j \frac{\partial}{\partial \beta_j} F(\beta_0, \beta),
\] where \(\gamma_j\) is the step size. The usual CD choice is a Newton
step, based on the second derivative \[
  H_{jj} = \frac{\partial^2}{\partial \beta_j^2} F(\beta_0, \beta)
         = \frac{1}{n}\sum_{i=1}^n x_{ij}^2 f''(\eta_i; y_i).
\]

Letting \(w_i = f''(\eta_i; y_i)\) be the \emph{weights}, we can rewrite
the Hessian diagonal as \[
  H_{jj} = \frac{1}{n}\sum_{i=1}^n x_{ij}^2 w_i.
\]

This yields the step size \[
  \gamma_j = H_{jj}^{-1}
           = \left( \frac{1}{n}\sum_{i=1}^n x_{ij}^2 w_i \right)^{-1}.
\]

For the intercept, the local curvature is \[
  H_{00} = \frac{\partial^2}{\partial \beta_0^2} F(\beta_0, \beta)
         = \frac{1}{n}\sum_{i=1}^n f''(\eta_i; y_i).
\] The corresponding Newton step size is \(\gamma_0 = H_{00}^{-1}\).
Production solvers, however, use three different intercept updates:

\begin{enumerate}
\def\labelenumi{\arabic{enumi}.}
\tightlist
\item
  A Newton step, setting \(\gamma_0\) as above.
\item
  A gradient descent step, setting \(\gamma_0\) to the reciprocal of the
  global Lipschitz constant with respect to \(\beta_0\) or choosing it
  by a backtracking line search.
\item
  Repeated safeguarded Newton steps on the intercept subproblem until
  the scalar convergence criterion described below is met---the
  \emph{exact strategy}.
\end{enumerate}

For a slope coefficient, this smooth step must also account for the
penalty. By separability, write
\(G_\lambda(\beta) = \sum_{j=1}^p g_{\lambda,j}(\beta_j)\) and define \[
  \operatorname{prox}_{\gamma g}(v)
  = \operatorname*{arg\,min}_{u \in \mathbb{R}} \left\{ g(u) + \frac{1}{2\gamma}(u - v)^2 \right\}.
\] The coordinate update is therefore the proximal-gradient step \[
  \beta_j^+
  = \operatorname{prox}_{\gamma_j g_{\lambda,j}} \left( \beta_j - \gamma_j \frac{\partial}{\partial\beta_j}F(\beta_0,\beta) \right).
\] For the lasso, \(g_{\lambda,j}(u) = \lambda|u|\), this is
soft-thresholding at \(\gamma_j\lambda\). The intercept has no
corresponding proximal step because it is unpenalized. The resulting
coordinate-descent algorithm appears in Algorithm~\ref{algo-cd}.

\floatstyle{ruled}
\restylefloat{algorithm}
\begin{algorithm}[ht]
    \caption{\label{algo-cd}Cyclic coordinate descent.}
    \begin{algorithmic}[1]
        \State \textit{Input:} Data $(X, y)$, $\lambda > 0$, initial intercept and coefficients $(\beta_0, \beta)$
        \While{not converged}
            \For{$j$ in $\{1, 2, ..., p\}$}
                \State $\beta_j \gets \operatorname{prox}_{\gamma_j g_{\lambda,j}}\left( \beta_j - \gamma_j \frac{\partial}{\partial \beta_j} F(\beta_0, \beta) \right)$
            \EndFor

            \If{intercept strategy is "Newton"}
                \State $\delta \gets -\left( \frac{\partial^2}{\partial \beta_0^2} F(\beta_0, \beta) \right)^{-1} \frac{\partial}{\partial \beta_0} F(\beta_0, \beta)$
                \State $\alpha \gets 1$
                \While{$F(\beta_0 + \alpha \delta, \beta) > F(\beta_0, \beta) + c \alpha \delta \frac{\partial}{\partial \beta_0} F(\beta_0, \beta)$}
                    \State $\alpha \gets \alpha / 2$
                \EndWhile
                \State $\beta_0 \gets \beta_0 + \alpha \delta$
            \ElsIf{intercept strategy is "Gradient"}
                \State $\beta_0 \gets \beta_0 - \frac{1}{L_0} \frac{\partial}{\partial \beta_0} F(\beta_0, \beta)$
            \ElsIf{intercept strategy is "Exact"}
                \State $k \gets 0$
                \While{$|\partial_0 F| > 10^{-10}$ and $k < 50$}
                    \State $\delta \gets -\left( \frac{\partial^2}{\partial \beta_0^2} F(\beta_0, \beta) \right)^{-1} \frac{\partial}{\partial \beta_0} F(\beta_0, \beta)$
                    \State choose $\alpha$ by Armijo backtracking as in the Newton branch
                    \State $\beta_0 \gets \beta_0 + \alpha \delta$
                    \State $k \gets k + 1$
                \EndWhile
                \If{$k = 50$ and $|\partial_0 F| > 5 \times 10^{-10}$}
                    \State raise a convergence error
                \EndIf
            \EndIf
        \EndWhile
        \State \textit{Output:} $(\beta_0, \beta)$
    \end{algorithmic}
\end{algorithm}
\floatstyle{plain}

In the implementation, every coefficient update in
Algorithm~\ref{algo-cd} uses a Tseng--Yun Armijo backtracking line
search based on the proximal model decrease, so each pass satisfies the
descent condition by construction. The pseudocode omits this guard for
readability. The Newton branch likewise backtracks only when the full
step overshoots, and the exact branch enforces the residual criterion
described above. \hyperref[sec-s-methodology]{Supplement S3} gives the
constants, iteration caps, and numerical tolerances.

\subsection{What sets the intercept
apart}\label{what-sets-the-intercept-apart}

Four properties distinguish the intercept \(\beta_0\) from a slope
coefficient \(\beta_j\):

\begin{enumerate}
\def\labelenumi{\arabic{enumi}.}
\item
  The intercept's design column is \(\mathbf{1}_n\). Treating
  \(\beta_0\) as a coordinate is equivalent to augmenting \(X\) with a
  constant column of ones. Consequently, \(H_{00} = (1/n)\sum_i w_i\) is
  a mean of weights rather than a weighted mean of squared features. The
  global Lipschitz constant
  \(L_0 = \|\mathbf{1}\|^2 f''_{\max}/n = f''_{\max}\) does not shrink
  with the data; for binary logistic regression, \(L_0 = 1/4\). Finally,
  the cross-Hessian \(H_{0j} = (1/n)\sum_i w_i x_{ij}\) is the
  \(w\)-weighted column mean. The intercept therefore couples to every
  slope through a single scalar rather than through pairwise feature
  correlations.
\item
  The intercept is unpenalized. The regularizer \(G_\lambda\) acts only
  on \(\beta\), so \(\beta_0\) is the only coordinate whose update is
  smooth---no proximal operator, no soft-thresholding, no
  nondifferentiable optimality condition. Hence, the strategies in
  Algorithm~\ref{algo-cd} can update \(\beta_0\) directly.
\item
  The intercept is never inactive. State-of-the-art CD solvers use
  screening rules and working sets to limit the number of coordinates
  visited each pass. In high-dimensional settings, most \(\beta_j\) are
  visited in only a fraction of passes, and their average per-pass cost
  is correspondingly small. The intercept is in the active set of every
  pass by construction, so its per-pass treatment dominates the overall
  intercept-update cost, unlike that of any single slope.
\item
  In IRLS-CD, optimizing the intercept re-centers the quadratic
  surrogate by setting the weighted mean of its residuals to zero. No
  individual slope performs this role. \textbf{adelie} can therefore
  absorb the intercept into a parameterization of the inner Gaussian
  problem centered by weighted means rather than update it as a separate
  coordinate. Section~\ref{sec-irls} derives this identity.
\end{enumerate}

The ratio \(L_0/H_{00}\) compares the Newton step size \(1/H_{00}\) with
the fixed gradient step size \(1/L_0\). The first property explains why
this ratio is unbounded under response imbalance: \(L_0\) does not adapt
to the curvature at the current iterate, while \(H_{00}\) collapses with
\((1/n)\sum_i w_i\) as the linear predictor saturates. The second and
third explain why solvers can choose among several strategies and why
the cost of that choice persists on every pass. The fourth explains why
an IRLS solver can absorb the intercept into its quadratic surrogate.

\subsection{Cross-influence of intercept and
coefficients}\label{cross-influence-of-intercept-and-coefficients}

An intercept update can change the gradients used for the next
coefficient updates. If the intercept lags behind its conditional
optimum, those gradients reflect the remaining intercept error, which
can slow progress across passes. The update strategies differ in how
much of this error they remove. To see how the error reaches the
coefficients, we examine the curvature shared by the intercept and each
coefficient. Let \(w_i = f''(\eta_i; y_i)\). The Hessian blocks are \[
  H_{jk}
  = \frac{\partial^2}{\partial \beta_j \partial \beta_k} F(\beta_0, \beta)
  = \frac{1}{n}\sum_{i=1}^n x_{ij} x_{ik} w_i
\] for the coefficient--coefficient block and \[
  H_{0j}
  = \frac{\partial^2}{\partial \beta_0 \partial \beta_j} F(\beta_0, \beta)
  = \frac{1}{n}\sum_{i=1}^n x_{ij} w_i
\] for the intercept--coefficient block. The coefficients couple to each
other \emph{pairwise} through the weighted Gram matrix of the features.
Correlated features produce large \(|H_{jk}|\), which is the classical
reason coordinate descent struggles on such designs. The intercept
instead couples to each coefficient \emph{globally} through the weighted
mean of that feature, without an explicit pairwise-correlation term.

For squared loss, \(w_i = 1\) and \(H_{0j} = (1/n)\sum_i x_{ij}\), so
centering the features makes \(H_{0j} = 0\) and the intercept decouples
from the coefficients entirely. For other losses, the weights depend on
the linear predictor and the response, so \(H_{0j}\) is a
\(w_i\)-weighted mean of \(x_{ij}\), which is generally nonzero even
after centering. That is why the intercept and coefficients can be
strongly coupled in GLMs, and why the choice of intercept-update
strategy can affect convergence.

\subsubsection{A normalized measure of intercept--coefficient
coupling}\label{a-normalized-measure-of-interceptcoefficient-coupling}

The cross-curvature \(H_{0j}\) measures how much a small intercept
change affects the gradient for coefficient \(j\). Its magnitude depends
on the feature's units, so we normalize it by the curvature in each
coordinate. For \(H_{00} > 0\) and \(H_{jj} > 0\), define
\begin{equation}\protect\phantomsection\label{eq-rho}{
  \rho_{0j} = \frac{H_{0j}}{\sqrt{H_{00} H_{jj}}} \in [-1, 1],
}\end{equation} where the bound follows from Cauchy--Schwarz. If
\(\rho_{0j} = 0\), changing the intercept leaves the coefficient's
gradient unchanged to first order. In the local quadratic model,
\(|\rho_{0j}|\) close to one means that adjusting the intercept along
with the coefficient can greatly reduce the curvature compared with
holding the intercept fixed.

For example, suppose that \(x_{ij} = 2\) for every observation.
Increasing \(\beta_j\) by \(0.1\) adds \(0.2\) to every linear
predictor. Decreasing \(\beta_0\) by \(0.2\) cancels this change
exactly, so the loss stays unchanged. Here \(\rho_{0j} = 1\): the
intercept and coefficient have interchangeable effects on the loss. A
similar local effect occurs when the feature is nearly constant at a
nonzero value among observations with substantial curvature weights
\(w_i\), even if it varies across the full data set.

This coupling makes it important for the intercept to keep up with
coefficient changes. A coefficient update can move the conditional
optimum of the intercept, and if the next intercept update leaves much
of that error unresolved, the error carries into subsequent coefficient
gradients and risks hampering progress of the whole solver.

Positive rescaling of the feature or loss leaves \(\rho_{0j}\)
unchanged, while reversing the feature's sign reverses \(\rho_{0j}\). We
use \(\rho_{0j}\) only when both diagonal curvatures are positive---if
\(H_{00} = 0\), a local Newton step is also undefined.

\subsubsection{The profiled objective}\label{the-profiled-objective}

As a reference for comparing the strategies, suppose that the intercept
always reaches its conditional optimum. The coefficients then determine
the \emph{profiled} objective, obtained by minimizing over the
intercept: \[
  \tilde F(\beta) = \min_{\beta_0} F(\beta_0, \beta).
\]

Evaluate the Hessian blocks at \((\beta_0^\star(\beta), \beta)\). When
\(H_{00} > 0\), the Hessian of \(\tilde F\) is the Schur complement
\begin{equation}\protect\phantomsection\label{eq-schur}{
  \tilde H_{jk} = H_{jk} - \frac{H_{0j} H_{0k}}{H_{00}},
}\end{equation} and in particular
\begin{equation}\protect\phantomsection\label{eq-schur-diag}{
  \tilde H_{jj} = H_{jj} - \frac{H_{0j}^2}{H_{00}} = H_{jj}(1 - \rho_{0j}^2).
}\end{equation}

When \(|\rho_{0j}|\) is close to \(1\), adjusting the intercept offsets
much of a change in coefficient \(j\). The resulting profiled curvature
\(\tilde H_{jj}\) is therefore much smaller than the fixed-intercept
curvature \(H_{jj}\): \(\rho_{0j}^2\) is the fraction removed in the
local quadratic model (Equation~\ref{eq-schur-diag}). This difference
affects coordinate descent in two ways. Immediately after optimizing the
intercept, the coefficient gradient equals that of \(\tilde F\), but the
step still uses \(H_{jj}\) and can be conservative relative to the
profiled curvature. After a coefficient update, the intercept may become
stale. Until it is updated again, its error enters subsequent
coefficient gradients through the cross-curvature \(H_{0j}\) and can
slow progress.

\subsubsection{Strategies as Schur
approximations}\label{strategies-as-schur-approximations}

An intercept error changes the gradients used to update the
coefficients. We compare the three strategies by asking how much of this
gradient error each removes. Exact minimization supplies the profiled
gradient, one Newton step approximates it locally, and a gradient step
can leave much of the error behind. We first illustrate this comparison
with logistic regression, then establish the local bounds and consider
whether extra intercept iterations justify their cost.

Write the intercept as \(a\) and one coefficient as \(b\). Consider 30
observations, with ten at each feature value \(x = 0, 1, 2\) and,
respectively, one, one, and two positive responses. Their average
logistic loss is
\begin{equation}\protect\phantomsection\label{eq-logistic-example}{
  F(a, b)
  = \frac{1}{30}\sum_{i=1}^{30} \left[ \log(1 + e^{a+bx_i}) - y_i(a + bx_i) \right].
}\end{equation} For fixed \(b\), the optimal intercept \(a^\star(b)\)
makes the average fitted probability equal the observed proportion of
positive responses, \(4/30\). Unlike a quadratic intercept subproblem,
this problem generally requires more than one Newton step to solve. The
coefficient gradient is \[
  \partial_b F(a, b) = \frac{1}{30}\sum_{i=1}^{30} x_i[p_i(a,b) - y_i],
  \qquad
  p_i(a,b) = \frac{1}{1 + e^{-(a+bx_i)}}.
\] Changing the intercept changes the fitted probabilities \(p_i\) and
therefore the coefficient gradient, even though \(b\) stays fixed.

Start at \((a, b) = (-2, 1)\) and hold \(b\) fixed. The gradient
strategy uses the global intercept curvature bound \(L_0 = 1/4\),
whereas Newton uses the current curvature \(H_{00}\). Exact minimization
solves \(\partial_a F = 0\) numerically. All three updates decrease the
loss at this starting point, and the full Newton step satisfies the
Armijo condition used by our solver. Their effects on the coefficient
gradient differ:

\begin{longtable}[]{@{}
  >{\raggedright\arraybackslash}p{(\linewidth - 4\tabcolsep) * \real{0.2462}}
  >{\raggedleft\arraybackslash}p{(\linewidth - 4\tabcolsep) * \real{0.2923}}
  >{\raggedleft\arraybackslash}p{(\linewidth - 4\tabcolsep) * \real{0.4615}}@{}}
\caption{Intercept and coefficient gradient for
Equation~\ref{eq-logistic-example} at \(b = 1\). Values are rounded to
three decimal places. The exact row gives the profiled gradient
\(\tilde F'(1)\).}\label{tbl-logistic-updates}\tabularnewline
\toprule\noalign{}
\begin{minipage}[b]{\linewidth}\raggedright
Intercept update
\end{minipage} & \begin{minipage}[b]{\linewidth}\raggedleft
Intercept afterward
\end{minipage} & \begin{minipage}[b]{\linewidth}\raggedleft
Coefficient gradient afterward
\end{minipage} \\
\midrule\noalign{}
\endfirsthead
\toprule\noalign{}
\begin{minipage}[b]{\linewidth}\raggedright
Intercept update
\end{minipage} & \begin{minipage}[b]{\linewidth}\raggedleft
Intercept afterward
\end{minipage} & \begin{minipage}[b]{\linewidth}\raggedleft
Coefficient gradient afterward
\end{minipage} \\
\midrule\noalign{}
\endhead
\bottomrule\noalign{}
\endlastfoot
None & \(-2.000\) & \(0.256\) \\
Gradient & \(-2.651\) & \(0.116\) \\
Newton & \(-2.885\) & \(0.072\) \\
Exact & \(-3.106\) & \(0.035\) \\
\end{longtable}

Here Newton brings the coefficient gradient closer to the profiled
gradient than the gradient strategy does, but it does not reach it. This
example illustrates the distinction; the local bounds below state when
the remaining error becomes negligible. For a quadratic loss, one full
Newton step instead solves the intercept subproblem exactly, so Newton
and exact minimization coincide.

Figure~\ref{fig-schur-example} shows both the intercept updates and what
happens after an exact solve. The optimal intercept depends on \(b\):
changing the coefficient while holding the intercept fixed can move the
iterate away from the curve \(a = a^\star(b)\). Thus, even exact
intercept optimization need not give later coefficient updates in the
same pass the profiled gradient.

\begin{figure}

\centering{

\pandocbounded{\includegraphics[keepaspectratio]{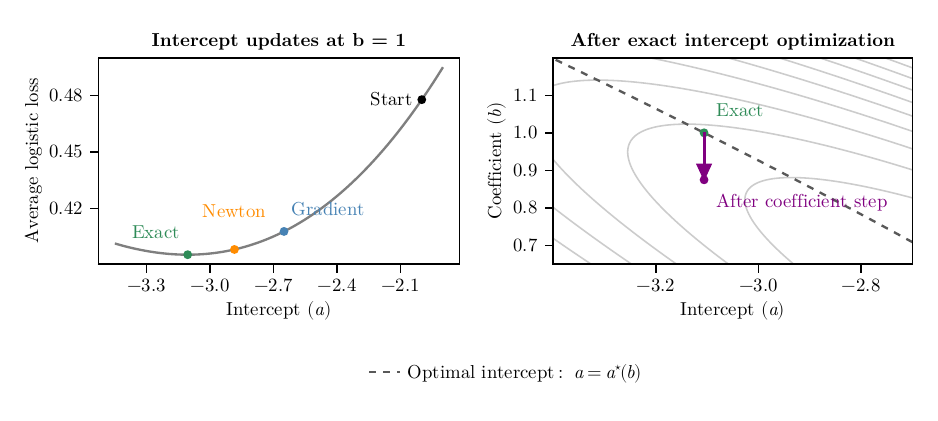}}

}

\caption{\label{fig-schur-example}Intercept updates for the logistic
example in Equation~\ref{eq-logistic-example}. Left: the loss as a
function of the intercept with \(b = 1\), showing the starting point and
the three updated intercepts. Right: after exact intercept optimization,
a coefficient Newton step with the intercept held fixed moves \(b\) from
\(1\) to about \(0.875\) and decreases the loss, but leaves the curve of
conditionally optimal intercepts. Gray curves are loss level sets. The
right panel enlarges the neighborhood of the exact update.}

\end{figure}%

To explain these differences for a general smooth loss, we expand the
coefficient gradient after an intercept step. If we take the step
\(\beta_0 \mapsto \beta_0 - \alpha \partial_0 F\), twice continuous
differentiability gives
\begin{equation}\protect\phantomsection\label{eq-gradient-response}{
  \partial_j F(\beta_0 - \alpha\partial_0 F, \beta)
  = \partial_j F(\beta_0, \beta)
    - \alpha H_{j0} \partial_0 F(\beta_0, \beta) + r_j(\alpha),
}\end{equation} where \(r_j(\alpha) = o(|d_0|)\) as the intercept
displacement \(d_0 = -\alpha \partial_0 F\) tends to zero. If \(H_{j0}\)
is locally Lipschitz, then \(r_j(\alpha) = O(d_0^2)\). For a quadratic
objective, the gradient is affine, so the remainder is zero.

When \(H_{00} > 0\), the local Schur correction
\(\partial_j F - (H_{j0}/H_{00})\partial_0 F\) is the gradient that a
Newton intercept step predicts for the profiled objective at the current
\(\beta\). It equals the profiled gradient in a frozen quadratic and
approximates it locally for a nonlinear loss. Conditional minimization
gives the true profiled gradient. The three strategies differ in how
closely the gradient response in Equation~\ref{eq-gradient-response}
approaches that target.

\paragraph{What exact intercept optimization
achieves}\label{what-exact-intercept-optimization-achieves}

Exact minimization provides the reference for comparing the strategies.
It gives the next coefficient update the profiled gradient, but the
update still uses the curvature with the intercept held fixed. The
following proposition makes this distinction precise.

\begin{proposition}[]\protect\hypertarget{prp-profile-equiv}{}\label{prp-profile-equiv}

(First-update profile equivalence.) Let \(F\) be twice continuously
differentiable and strictly convex in \(\beta_0\) for every \(\beta\),
and let \(\tilde F(\beta) = \min_{\beta_0} F(\beta_0, \beta)\). Suppose
the exact strategy has driven \(\beta_0\) to the conditional minimizer
\(\beta_0^{(t)} = \arg\min_{\beta_0} F(\beta_0, \beta^{(t)})\), and
suppose \(H_{00} > 0\) and \(\tilde H_{jj} > 0\) there. Then the
\emph{first} coefficient update of the next pass uses the gradient
\(\partial_j F(\beta_0^{(t)}, \beta^{(t)}) = \partial_j \tilde F(\beta^{(t)})\),
so its update direction agrees with the corresponding coordinate
direction for \(\tilde F\) at \(\beta^{(t)}\). Its diagonal step size
remains \(1/H_{jj}\). Because \(\tilde H_{jj} \leq H_{jj}\) by
Equation~\ref{eq-schur-diag}, this step size is no larger than
\(1/\tilde H_{jj}\), the step size for \(\tilde F\)'s local quadratic
model along \(e_j\).

\end{proposition}

\begin{proof}
Let \(\beta_0^\star(\beta)\) be the conditional minimizer near
\(\beta^{(t)}\). It exists at \(\beta^{(t)}\) by assumption and is
unique by strict convexity. Since \(H_{00} > 0\) there, the implicit
function theorem applied to \(\partial_0 F = 0\) makes \(\beta_0^\star\)
differentiable locally. Thus
\(\tilde F(\beta) = F(\beta_0^\star(\beta), \beta)\) near
\(\beta^{(t)}\). Differentiating through \(\beta_0^\star\) and using the
stationarity condition
\(\partial_0 F(\beta_0^\star(\beta), \beta) = 0\), \[
  \partial_j \tilde F(\beta)
  = \partial_j F(\beta_0^\star(\beta), \beta)
    + \partial_0 F\big(\beta_0^\star(\beta), \beta\big)\partial_j \beta_0^\star(\beta)
  = \partial_j F\big(\beta_0^\star(\beta), \beta\big).
\] At the conditional minimizer, \(\partial_0 F = 0\), so the
intercept's dependence on \(\beta\) contributes nothing to the
derivative. Evaluating at \(\beta_0^{(t)} = \beta_0^\star(\beta^{(t)})\)
gives the gradient equality. For the step size,
\(\tilde H_{jj} = H_{jj}(1 - \rho_{0j}^2) \leq H_{jj}\), so
\(1/H_{jj} \leq 1/\tilde H_{jj}\)~(Equation~\ref{eq-schur-diag}). The
update therefore uses a step size no larger than the profiled local
model's \(1/\tilde H_{jj}\). This local comparison alone does not
guarantee descent of \(\tilde F\) after a finite update.
\end{proof}

The proposition establishes a gradient identity at the current
coefficients. It does not turn a whole coordinate sweep into
optimization of \(\tilde F\): the coefficient steps still use
\(H_{jj}\), and later updates can move the intercept away from its
conditional optimum. Before examining that drift, we ask how closely one
Newton step reproduces the benefit of an exact solve.

\paragraph{How closely one Newton step matches an exact
solve}\label{how-closely-one-newton-step-matches-an-exact-solve}

A full Newton step leaves an intercept error that is quadratic in the
current error, provided the curvature is positive and varies smoothly
nearby. Write \(g(t) = \partial_0 F(t, \beta)\), let
\(\beta_0^* = \beta_0^*(\beta)\) satisfy \(g(\beta_0^*) = 0\), and set
\(H_{00} = g'(\beta_0)\). Taylor-expanding \(g(\beta_0^*)\) around the
current iterate \(\beta_0\) gives
\begin{equation}\protect\phantomsection\label{eq-intercept-taylor}{
  \partial_0 F(\beta_0, \beta)
  = H_{00}(\beta_0 - \beta_0^*)
    - \tfrac{1}{2}\partial_{000} F(\xi)(\beta_0 - \beta_0^*)^2,
}\end{equation} for some \(\xi\) between \(\beta_0\) and
\(\beta_0^*(\beta)\). Dividing by \(H_{00}\) and recognizing
\(\beta_0^{\text{Newton}} = \beta_0 - \partial_0 F/H_{00}\) on the left
rewrites Equation~\ref{eq-intercept-taylor} as the standard Newton
identity for the iterate error:
\begin{equation}\protect\phantomsection\label{eq-newton-iterate}{
  \bigl|\beta_0^{\text{Newton}} - \beta_0^*(\beta)\bigr|
  = \frac{|\partial_{000} F(\xi)|}{2 |H_{00}|}(\beta_0 - \beta_0^*)^2.
}\end{equation} For a local bound, suppose that \(g'(t) \ge m > 0\) and
\(|g''(t)| \le M\) between \(\beta_0\) and \(\beta_0^*\). The mean-value
theorem gives \(m|\beta_0 - \beta_0^*| \le |g(\beta_0)|\), so
Equation~\ref{eq-newton-iterate} yields
\begin{equation}\protect\phantomsection\label{eq-newton-error}{
  \bigl|\beta_0^{\text{Newton}} - \beta_0^*(\beta)\bigr|
  \le \frac{M}{2m^3} \bigl|\partial_0 F(\beta_0, \beta)\bigr|^2.
}\end{equation} Here we may take \(M\) as a bound on
\(|\partial_{000}F| = |(1/n)\sum_i f'''(\eta_i)|\), and hence as no
larger than \(\sup |f'''|\). For squared loss, \(f''' \equiv 0\) and the
Newton step is exact. For logistic regression, \(|f'''|\) is bounded by
a small constant. Thus, wherever the intercept curvature stays bounded
away from zero, the residual shrinks quadratically as the current
intercept gradient vanishes.

A small intercept error matters only through its effect on the
coefficient update. We therefore carry this bound into the coefficient
gradient. Let \(r_0 = \beta_0^{\text{Newton}} - \beta_0^*(\beta)\) be
the Newton residual, with \(|r_0|\) bounded by
Equation~\ref{eq-newton-error}. Taylor-expanding \(\partial_j F\) in
\(\beta_0\) around the conditional minimizer \(\beta_0^*(\beta)\), using
\(\partial_j F(\beta_0^*, \beta) = \partial_j \tilde F(\beta)\) and
\(H_{j0}^* = \partial_{0j} F(\beta_0^*, \beta)\), yields
\begin{equation}\protect\phantomsection\label{eq-newton-coupling}{
  \partial_j F\bigl(\beta_0^{\text{Newton}}, \beta\bigr)
  = \partial_j \tilde F(\beta) + H_{j0}^*r_0 + O(r_0^2),
}\end{equation} The coefficient gradient differs from the profiled
gradient by \(b_j = H_{j0}^*r_0 + O(r_0^2)\). For a penalized problem,
we also need to know how this difference changes the actual coordinate
step. We isolate its effect by comparing both steps at the same positive
curvature \(h_j\). Write
\(G_\lambda(\beta) = \sum_j G_{\lambda,j}(\beta_j)\) and define the
profiled step \[
  d_j^{\mathrm{prof}}
  = \operatorname{prox}_{G_{\lambda,j}/h_j} \left( \beta_j - \frac{\partial_j\tilde F(\beta)}{h_j} \right)
    - \beta_j.
\] The step \(d_j^{\mathrm{N}}\) after the Newton intercept update
replaces \(\partial_j\tilde F\) by \(\partial_j\tilde F + b_j\).
Nonexpansiveness of the proximal map gives
\begin{equation}\protect\phantomsection\label{eq-newton-good-enough}{
  |d_j^{\mathrm{N}} - d_j^{\mathrm{prof}}| \leq \frac{|b_j|}{h_j},
  \qquad
  |b_j| \ll h_j |d_j^{\mathrm{prof}}|
}\end{equation} as an absolute perturbation bound and a sufficient
criterion for a small relative step error, respectively. The relevant
comparison is with the composite coordinate residual
\(h_j d_j^{\mathrm{prof}}\). At an active lasso optimum,
\(\partial_j\tilde F = -\lambda\operatorname{sign}(\beta_j)\), so the
smooth gradient itself generally does not vanish.

Under the local conditions of Equation~\ref{eq-newton-error}, \(b_j\) is
\(O(|\partial_0 F|^2)\) when the cross derivatives remain bounded. If
\(h_j\) stays bounded away from zero, the profiled coordinate step
remains first order in the iterate error, and \(\partial_0 F\) vanishes
at least linearly, the Newton perturbation is second order and becomes
negligible relative to that step. This is a conditional local
comparison: when the profiled step is zero or vanishes faster, the
absolute bound in Equation~\ref{eq-newton-good-enough} still applies,
but it supplies no relative guarantee.

Far from the optimum, the full-step approximation can fail when the
third derivative is large, the intercept curvature approaches zero, or
the intercept gradient is large. Poisson regression, where
\(f''' = e^\eta\), can exhibit this cold-start behavior. Armijo
backtracking safeguards descent, but a damped step has a different error
bound from a full Newton step. With \(e_0 = \beta_0 - \beta_0^*(\beta)\)
and accepted step fraction \(\alpha \in (0,1]\), the damped intercept
error is
\begin{equation}\protect\phantomsection\label{eq-damped-newton-error}{
  \beta_0^+ - \beta_0^*(\beta) = (1 - \alpha)e_0 + \alpha r_0.
}\end{equation} The quadratic error bound applies when the full step is
accepted. For \(\alpha < 1\), sufficient decrease alone does not remove
the first-order term \((1 - \alpha)e_0\).

The error bound and coupling correction describe a single Newton step at
the current iterate
(Equation~\ref{eq-newton-error}--Equation~\ref{eq-newton-coupling}). The
experiments in Section~\ref{sec-results} show how those gains accumulate
across an outer solve.

\paragraph{What a gradient step leaves
behind}\label{what-a-gradient-step-leaves-behind}

The gradient strategy uses the curvature bound \(L_0\) in place of the
current intercept curvature \(H_{00}\). Its step can therefore be much
smaller than the Newton step. When \(H_{00} > 0\), substituting
\(\alpha = 1/L_0\) into Equation~\ref{eq-gradient-response} gives
\begin{equation}\protect\phantomsection\label{eq-grad-residual}{
  \partial_j F^{\text{grad}}
  = \partial_j F - \frac{H_{00}}{L_0}\frac{H_{j0}}{H_{00}}\partial_0 F
    + r_j(1/L_0),
}\end{equation} so, to first order, the gradient strategy removes a
fraction \(H_{00}/L_0 \in (0, 1]\) of the coupling component
\((H_{j0}/H_{00})\partial_0 F\) that Newton targets at leading order.
The exact strategy instead reaches the true profiled gradient by
conditional minimization. In a frozen quadratic, \(r_j(1/L_0) = 0\) and
the correction fraction is exact. For a nonlinear loss, the
leading-order residual relative to the local Schur-corrected target is
\((1 - H_{00}/L_0)(H_{j0}/H_{00})\partial_0 F + r_j(1/L_0)\). When
\(L_0 = \infty\), the intercept step is zero, so no coupling is removed.

The leading-order correction is fully effective only when
\(H_{00} \approx L_0\), that is, when the Lipschitz bound is tight at
the current iterate. Whenever \(H_{00}\) is appreciably below
\(L_0\)---as happens for logistic regression once predictions become
confident, or whenever the weights \(w_i\) are concentrated on a small
subset of observations---the gradient strategy leaves a fraction
\(1 - H_{00}/L_0\) of the leading-order coupling correction unapplied.
The remaining intercept error continues to affect the coefficient
gradients and can slow subsequent passes.

For Poisson loss, this strategy has no finite global curvature bound to
use: \(L_0 = \infty\) gives a zero intercept step at every finite
iterate. The coefficient gradients then retain the full effect of the
intercept error. This limitation concerns the strategy based on a global
bound; a gradient method with an adaptive step size or line search need
not leave the intercept fixed.

\paragraph{Why extra intercept iterations may not pay
off}\label{why-extra-intercept-iterations-may-not-pay-off}

The comparisons above hold the coefficients fixed. During a coordinate
sweep, however, each coefficient update can move the intercept's
conditional optimum. Even an exact intercept solve therefore guarantees
the profiled gradient only for the first coefficient update of the next
pass.

To measure this drift, expand \(\partial_0 F\) in the \(\beta_j\)
direction at the conditional minimizer. With
\(\partial_{0j} F = H_{0j}\) and
\(\partial_0 F(\beta_0^{(t)}, \beta^{(t)}) = 0\), a single coefficient
update \(\beta_j^{(t)} \mapsto \beta_j^{(t)} + \delta_j\) gives, when
the relevant Hessian entries are locally Lipschitz,
\begin{equation}\protect\phantomsection\label{eq-intercept-drift}{
  \partial_0 F \big(\beta_0^{(t)}, \beta^{(t)} + \delta_j e_j \big)
  = H_{0j}\delta_j + O(\delta_j^2).
}\end{equation} When this change is nonzero, the intercept is no longer
conditionally optimal, and Proposition~\ref{prp-profile-equiv} no longer
guarantees the profiled gradient for the next coefficient update. The
drift depends on both the coupling and the coefficient step; it need not
occur after every update.

The exact strategy spends extra work optimizing an intercept that the
next coefficient update may immediately move away from its optimum. Let
\(k_0\) count the inner work units needed at the end of a pass to reach
\(|\partial_0 F| < \varepsilon\), including one residual check when the
intercept already meets this tolerance. Across the imbalance levels in
our experiments, \(k_0\) ranged from one to five, with each unit costing
\(O(n)\) work~(Figure~\ref{fig-per-pass-cost}). With the coefficients
held fixed, Newton's method converges rapidly near the intercept's
conditional minimizer: under its usual local convergence conditions,
reaching tolerance \(\varepsilon\) requires
\(O(\log\log(1/\varepsilon))\) further iterations. This rate does not
account for the work needed to enter that neighborhood, nor does it
guarantee that the intercept starts there after each coefficient pass.

Together, the Newton error bound and the drift calculation explain why
one Newton step can be enough: locally, it leaves only a second-order
intercept error, while the next coefficient update can introduce fresh
error even after an exact solve. These calculations do not establish how
much the extra inner iterations affect convergence over a full solve.
The experiments in \hyperref[sec-per-pass-cost]{Supplement S4.3} supply
that comparison: on our test problems, the additional work produces
little improvement in outer convergence.

The tradeoff can differ in other settings. Along a warm-started
\(\lambda\) path, a small initial intercept gradient can reduce \(k_0\)
to one or two work units, leaving little overhead for the exact
strategy. Strong coupling makes any remaining intercept error more
consequential for the coefficient gradients
(Equation~\ref{eq-newton-coupling}), but it also increases the drift
caused by coefficient updates (Equation~\ref{eq-intercept-drift}).
Coupling alone therefore does not determine whether extra intercept
iterations are worthwhile.

IRLS and prox-Newton methods also require a separate comparison because
they update the coefficients through a quadratic surrogate. With the
coefficients fixed, one full Newton step solves that surrogate's
intercept subproblem exactly. Updating the coefficients can still move
its optimal intercept, and rebuilding the surrogate does not by itself
eliminate this drift. \textbf{LIBLINEAR} resolves the intercept within
its frozen quadratic, whereas \textbf{BlitzL1} applies an exact-strategy
loop to the original loss after the prox-Newton step
(Section~\ref{sec-irls}). The cost and benefit depend on where these
solves occur in the outer algorithm. Our recommendation to avoid extra
intercept iterations concerns the cold-start, direct-CD experiments
studied here; it is not a general claim about these other settings.

\subsubsection{Local-rate
interpretation}\label{local-rate-interpretation}

To follow the correction in Equation~\ref{eq-grad-residual} through
successive passes, we study a local quadratic model. For the lasso,
freeze the Hessian at the optimum and restrict the coefficients to the
active set \(A\) with their signs fixed. Assume this restricted Hessian
is positive definite and define the collective coupling
\begin{equation}\protect\phantomsection\label{eq-collective-coupling}{
  \kappa = \frac{H_{0A}H_{AA}^{-1}H_{A0}}{H_{00}} \in [0, 1).
}\end{equation} If each pass first minimizes the active coefficient
block and then updates the intercept, Newton contracts the intercept
error by \(\kappa\), whereas the gradient strategy contracts it by
\(1 - (H_{00}/L_0)(1 - \kappa)\). As \(\kappa \to 1\), the ratio of the
required pass counts approaches \(L_0/H_{00}\). Thus, coupling
determines when the intercept mode controls convergence, and the
curvature ratio determines how much slower the gradient strategy becomes
once it does.

This is a local explanation, not a global iteration-complexity theorem.
A coordinate sweep retains coefficient error modes that the exact block
solve eliminates. \hyperref[sec-s-rate]{Supplement S1.1} gives the block
calculation and validates the full frozen coordinate-sweep model against
nonlinear convergence rates. The centering experiment there varies
coupling while holding the intercept curvature fixed, exposing the
transition to the \(L_0/H_{00}\) plateau.

\subsubsection{A safe Newton variant}\label{a-safe-newton-variant}

The Newton-coupling bias (Equation~\ref{eq-newton-coupling}) identifies
the cold-start regime---large \(|\partial_0 F|\) combined with rapidly
varying \(f''\)---as the setting where an unguarded Newton step on the
intercept perturbs the next coefficient pass enough to matter. The fix
is to wrap the Newton direction \(-\partial_0 F / H_{00}\) in an Armijo
backtracking line search: accept the full step (\(\alpha = 1\)) if it
decreases \(F\) enough. Otherwise, halve \(\alpha\) until it does. Our
Newton strategy includes this safeguard. When it accepts the full step,
the safeguard adds constant work; when the linear model is inadequate,
the number of backtracking trials grows logarithmically. To isolate the
effect of the guard, we compare the guarded variant directly with a bare
Newton step in Figure~\ref{fig-cold-start} and
Figure~\ref{fig-cold-start-poisson} and show when the guard changes the
result.

\subsubsection{Extension to vector
intercepts}\label{sec-multinomial-main}

The same mechanism extends to multinomial logistic regression. Under a
reference-class parameterization, the intercept becomes a
\((K - 1)\)-vector, and its Hessian is \[
  H_{00}
  = \frac{1}{n}\sum_{i=1}^n \left\{ \operatorname{diag}(p_{i,1:K-1}) - p_{i,1:K-1}p_{i,1:K-1}^{\top} \right\}.
\] Rare free classes create low-curvature coordinate directions. A rare
reference class creates a low-curvature common-shift direction. A global
Böhning majorant therefore attenuates the intercept correction in
precisely the modes associated with rare classes. A single
Armijo-guarded Newton block step adapts to those modes, while the exact
strategy adds inner iterations without a visible improvement in the
outer trajectory. \hyperref[sec-vector-intercepts]{Supplement S2.1}
derives the block update, identifies the relevant Rayleigh quotients,
and tests this prediction on a grid of synthetic problems.

The pattern appears on two real datasets with contrasting dimensions
(Figure~\ref{fig-multinomial-summary}). Yeoh2002 has \(n = 248\),
\(p = 12\,625\), and six classes, while StatLog Shuttle has
\(n = 58\,000\), \(p = 9\), and seven classes. We set
\(\lambda = 0.05\lambda_{\max}\) in both problems, where
\(\lambda_{\max}\) is the smallest penalty for which all fitted slopes
are zero.

On Yeoh2002, the gradient strategy needs roughly four times as many
passes as Newton to bring the relative suboptimality bound below
\(10^{-4}\). On Shuttle, whose rarest classes contain only ten and
thirteen observations, it fails to reach \(10^{-4}\) within 1,000
passes; Newton and exact reach that threshold in \(20\)--\(21\) passes.

\begin{figure}

\centering{

\pandocbounded{\includegraphics[keepaspectratio]{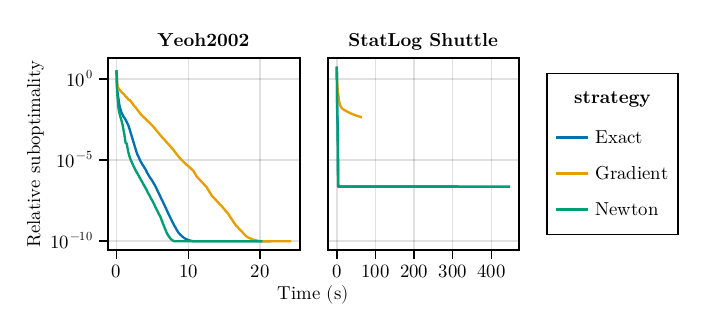}}

}

\caption{\label{fig-multinomial-summary}Multinomial logistic regression
on two real datasets. The panels show a relative suboptimality bound
versus time for Yeoh2002 (\(p \gg n\)) and StatLog Shuttle (\(n \gg p\))
at \(\lambda = 0.05\lambda_{\max}\) under cyclic coordinate descent.}

\end{figure}%

\subsection{Connection to IRLS}\label{sec-irls}

So far, we have considered \emph{direct} coordinate descent on \(F\),
which recomputes the weights \(w_i = f''(\eta_i; y_i)\) after every
coefficient update. A widely used alternative---and the one used by
\textbf{glmnet}~(Friedman et al. 2010)---is iteratively reweighted least
squares (IRLS)~(McCullagh and Nelder 1989). At the current iterate
\((\beta_0^{(t)}, \beta^{(t)})\), IRLS forms the weighted least squares
quadratic surrogate \[
  \tilde F^{(t)}(\beta_0, \beta)
  = \frac{1}{2n} \sum_{i=1}^n w_i^{(t)} \bigl( z_i^{(t)} - \beta_0 - x_i^\top \beta \bigr)^2,
\] with working response
\(z_i^{(t)} = \eta_i^{(t)} - r_i^{(t)} / w_i^{(t)}\). It then solves
\(\tilde F^{(t)}\) by inner CD and re-linearizes. In practice, solvers
use two kinds of weights: \emph{local} weights
\(w_i^{(t)} = f''(\eta_i^{(t)}; y_i)\) (\textbf{glmnet}'s
\texttt{type.logistic\ =\ "Newton"}), and \emph{MM-majorized} weights
replacing \(w_i^{(t)}\) by a global upper bound on \(f''\) (e.g.,
\(w_i \equiv 1/4\) for logistic regression, as in \textbf{glmnet}'s
\texttt{type.logistic\ =\ "modified.Newton"}, following the MM tradition
of Hunter and Lange (2004)).

Within an inner solve with local weights, \(\tilde F^{(t)}\) is
quadratic in \(\beta_0\) with curvature
\((1/n)\sum_i w_i^{(t)} = H_{00}\), so the inner Lipschitz constant
satisfies \(L_0^{\text{inner}} = H_{00}\). The fraction in
Equation~\ref{eq-grad-residual} is therefore
\(H_{00} / L_0^{\text{inner}} = 1\). Because a Newton step on a
quadratic is exact, the gradient, Newton, and exact strategies all
coincide on the inner subproblem. Thus, the distinction that we analyze
is specific to direct-CD solvers, such as \textbf{skglm}, and does not
arise within a local-weight IRLS inner solve.

The strategies no longer coincide under MM majorization. Replacing
\(w_i^{(t)}\) by the constant majorant \(1/4\) in the logistic case, the
closed-form inner intercept update on \(\tilde F^{(t)}\) is \[
  \beta_0 \gets \beta_0 + \frac{4}{n} \sum_{i=1}^n (y_i - p_i)
                = \beta_0 - \frac{1}{L_0} \partial_0 F, \qquad L_0 = 1/4,
\] which is the gradient strategy with a global Lipschitz step. The
identity holds at the linearization point, where \(p_i\) is evaluated,
so it describes the first inner intercept update after each
re-linearization; once inner CD has moved \(\beta\), subsequent inner
updates drive the intercept to the conditional minimum of the frozen
majorant quadratic (Table~\ref{tbl-classification}) rather than taking
further \(1/L_0\) steps on the original loss.

The correction transmitted at each re-linearization is nevertheless
governed by the same attenuation. The Schur-coupling residual
(Equation~\ref{eq-grad-residual}) thus applies directly to
\textbf{glmnet}'s \texttt{modified.Newton} mode: under response
imbalance, the leading-order correction fraction \(H_{00}/L_0\) shrinks
and the intercept update slows, approaching a true stall only as fitted
probabilities reach the boundary.

MM majorization was introduced to stabilize the working response \(z_i\)
when \(w_i \to 0\)---the same regime in which the resulting update
converges slowly. For Poisson loss, no global upper bound on \(f''\)
exists, so this MM construction is unavailable. \textbf{glmnet} instead
uses local-weight IRLS, where the three strategies coincide as described
above.

Changing from local weights to the constant majorant also changes every
coefficient's curvature and cross-curvature in the surrogate. A
within-package mode comparison therefore isolates the IRLS weight
scheme, not the intercept update alone. The intercept correction is the
component of that scheme analyzed by Equation~\ref{eq-grad-residual};
Figure~\ref{fig-skglm-controlled} compares changes to intercept handling
within a direct-CD solver.

Across IRLS linearizations, the weights \(w_i^{(t)}\) change with
\(\eta^{(t)}\), so even a fully resolved inner intercept becomes stale
once the weights are recomputed. This outer-level drift is analogous to
but distinct from the within-pass drift of
Equation~\ref{eq-intercept-drift}, which is specific to direct CD.

The appropriate intercept strategy therefore depends on the solver
family. The production solvers fall into the classes in
Table~\ref{tbl-classification}. Of the CD-based packages that we
surveyed, only \textbf{skglm} 0.5 and earlier apply a conservative
gradient update directly to the original loss. The remaining
solvers---\textbf{glmnet}, \textbf{biglasso}, \textbf{adelie},
\textbf{LIBLINEAR}, \textbf{BlitzL1}, \textbf{Lasso.jl}, and
\textbf{ncvreg}~(Breheny and Huang 2011)---all linearize before running
CD, using either IRLS or proximal-Newton~(Lee et al. 2014) on a
quadratic upper bound.\footnote{Sources for the two solvers not in
  Table~\ref{tbl-classification}: \textbf{ncvreg} recomputes local
  weights \texttt{w{[}i{]}\ =\ fmax2(mu\ *\ (1\ -\ mu),\ 0.0001)} at
  each linearization and updates the intercept as the
  weighted-least-squares conditional minimum of the surrogate
  (\texttt{b0{[}l{]}\ =\ xwr\ /\ xwx\ +\ a0},
  \url{https://github.com/pbreheny/ncvreg/blob/77a3d83/src/glm.c\#L198-L237});
  \textbf{Lasso.jl} fits GLMs by \textbf{glmnet}-style IRLS, an outer
  re-linearization loop with working weights and residuals around an
  inner CD solve
  (\url{https://github.com/JuliaStats/Lasso.jl/blob/cffedae/src/coordinate_descent.jl\#L696-L704}).}
The mode switches within \textbf{glmnet} (\texttt{Newton} versus
\texttt{modified.Newton}) and \textbf{biglasso} (\texttt{Newton} versus
\texttt{MM}) isolate the effect of the weight scheme analyzed in
Equation~\ref{eq-grad-residual}.

We exclude two related implementations from the table.
\textbf{Scikit-learn}'s L1-logistic path either delegates to
\textbf{LIBLINEAR} or uses SAGA, a variance-reduced stochastic-gradient
method outside the CD framework analyzed here. \textbf{celer}'s
L1-logistic estimator does not support an intercept in version 0.7.4,
and its documentation directs users to \textbf{skglm} for logistic
regression.

\begin{longtable}[]{@{}
  >{\raggedright\arraybackslash}p{(\linewidth - 4\tabcolsep) * \real{0.3750}}
  >{\raggedright\arraybackslash}p{(\linewidth - 4\tabcolsep) * \real{0.1750}}
  >{\raggedright\arraybackslash}p{(\linewidth - 4\tabcolsep) * \real{0.4500}}@{}}
\caption{Classification of CD-based production solvers by algorithmic
family and intercept treatment. IRLS and prox-Newton rows conditionally
minimize a frozen quadratic surrogate, not the original GLM loss;
\textbf{BlitzL1} alone additionally applies a post-step one-dimensional
Newton loop on the original
loss.}\label{tbl-classification}\tabularnewline
\toprule\noalign{}
\begin{minipage}[b]{\linewidth}\raggedright
Solver
\end{minipage} & \begin{minipage}[b]{\linewidth}\raggedright
Family
\end{minipage} & \begin{minipage}[b]{\linewidth}\raggedright
Intercept treatment
\end{minipage} \\
\midrule\noalign{}
\endfirsthead
\toprule\noalign{}
\begin{minipage}[b]{\linewidth}\raggedright
Solver
\end{minipage} & \begin{minipage}[b]{\linewidth}\raggedright
Family
\end{minipage} & \begin{minipage}[b]{\linewidth}\raggedright
Intercept treatment
\end{minipage} \\
\midrule\noalign{}
\endhead
\bottomrule\noalign{}
\endlastfoot
\textbf{skglm} 0.5 & Constant-Lipschitz CD & Gradient on original loss
(step size \(1/4\)) \\
\textbf{glmnet} (\texttt{type.logistic\ =\ "modified.Newton"}) & MM-IRLS
& Conditional minimum of global-majorant quadratic \\
\textbf{biglasso} (\texttt{alg.logistic\ =\ "MM"}) & MM-IRLS &
Conditional minimum of global-majorant quadratic \\
\textbf{glmnet} (\texttt{type.logistic\ =\ "Newton"}) & Local-IRLS &
Conditional minimum of frozen local quadratic \\
\textbf{biglasso} (\texttt{alg.logistic\ =\ "Newton"}, default) &
Local-IRLS & Conditional minimum of frozen local quadratic \\
\textbf{adelie} & Local-IRLS & Conditional minimum of weighted, centered
quadratic \\
\textbf{LIBLINEAR} (newGLMNET, L1-LR) & Prox-Newton & Conditional
minimum of frozen quadratic; unregularized \\
\textbf{BlitzL1} & Prox-Newton & Post-step Newton convergence on
original loss \\
\end{longtable}

\textbf{LIBLINEAR} (newGLMNET) and \textbf{BlitzL1} treat the intercept
differently at the boundary of each prox-Newton step. \textbf{LIBLINEAR}
leaves the intercept unregularized inside the inner CD solve. As the
quadratic program converges, the intercept reaches the conditional
minimum of the frozen quadratic, which generally differs from that of
the original loss. \textbf{BlitzL1} additionally runs an explicit
one-dimensional Newton loop to convergence on the original loss after
each subproblem and therefore applies the exact strategy at that
boundary.

The within-pass drift of Equation~\ref{eq-intercept-drift} penalizes the
exact strategy within direct CD because the inner Hessian drifts between
iterations. That critique does not apply at the prox-Newton boundary,
where the next subproblem's Hessian is recomputed anyway.

For direct CD, we therefore recommend a Newton update of the intercept.
Within each IRLS or prox-Newton quadratic surrogate, we recommend
optimizing the intercept fully, with an optional correction on the
original loss at the linearization boundary. The surveyed production
solvers follow the first two recommendations; \textbf{BlitzL1} also
follows the third.

In Figure~\ref{fig-irls-comparison}, we compare direct CD and IRLS on a
synthetic logistic problem. Throughout the synthetic experiments,
\(\mu_0\) denotes the baseline mean response at \(x = 0\)---a
positive-class probability for binomial models and a rate for Poisson
models---and \(s\) is the number of nonzero slopes in each generated
coefficient vector. Once the slopes enter the linear predictor,
\(\mu_0\) need not equal the empirical response mean.

\begin{figure}

\centering{

\pandocbounded{\includegraphics[keepaspectratio]{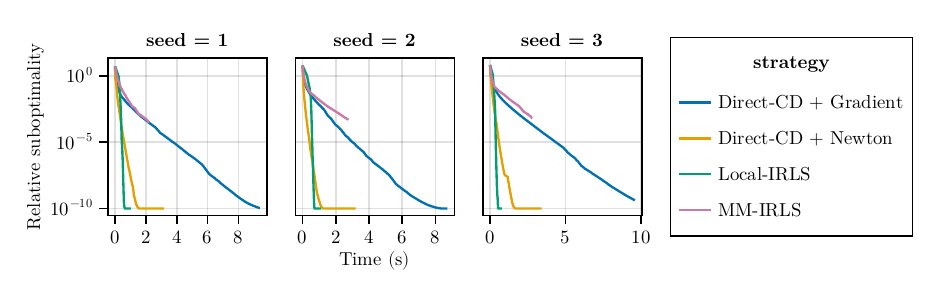}}

}

\caption{\label{fig-irls-comparison}Direct CD and IRLS on an imbalanced
logistic design. The panels show a relative suboptimality bound versus
time for \(\mu_0 = 0.99\), \(n = 500\), \(p = 1000\), \(s = 10\),
\(\lambda = 0.05\lambda_{\max}\), and three random seeds. The Local-IRLS
curve is shown once because the three intercept strategies coincide on
the inner quadratic.}

\end{figure}%

Local-IRLS converges fastest under \(\mu_0 = 0.99\), crossing the
\(10^{-8}\) bound in about 25 passes across the three seeds
(Figure~\ref{fig-irls-comparison}). Direct-CD with Newton takes roughly
twice as many passes, while Direct-CD with the gradient strategy takes
several hundred. MM-IRLS remains above \(10^{-4}\) after 200 passes.
Within Local-IRLS, the three strategies coincide exactly: once
\(L_0^{\text{inner}} = H_{00}\), they share the single inner update
\(-\partial_0 \tilde F^{(t)} / H_{00}\) and produce identical iterates.
The figure therefore plots only one curve. Unit tests verify exact
agreement of the intercept and coefficient paths.

Local-IRLS also outpaces Direct-CD with the Newton strategy on this
problem. This difference concerns the two algorithmic families rather
than their intercept strategies: IRLS-CD uses a constant inner Hessian
and has locally quadratic outer convergence, which favors it on
canonical-link GLMs where the linearization is tight. Direct-CD has the
edge when \(f'''\) is large (e.g.\\
Poisson with extreme rates, where each linearization is loose far from
the current iterate), when saturation makes the working response
unstable, or when screening and working-set methods (standard in modern
\(\ell_1\) solvers) work poorly with re-linearization~(Bertrand et al.
2022). This family-level difference does not change the intercept
recommendation: within either family, use local Newton curvature rather
than a global Lipschitz bound.

\section{Results}\label{sec-results}

\subsection{Methodology and reproducibility}\label{sec-methodology}

We run the internal experiments using our Julia implementations of the
solvers and report outer passes as the primary measure of work. Unless
otherwise noted, synthetic designs use cyclic ordering, fixed seeds,
standardized features, and an AR(1) correlation of \(\rho = 0.6\). We
evaluate comparable trajectories against the strongest feasible dual
value found for their shared problem instance, so the displayed
suboptimality bound does not depend on which strategy produced the
primal iterate. We use pinned datasets in the real-data experiments and
follow the same shared-reference principle in the production comparison.
We load cached outputs for computationally expensive experiments. The
toy examples in Figure~\ref{fig-parametric} and
Figure~\ref{fig-schur-example} run when the notebook renders. We pin the
Julia, R, and Python environments in the repository and document the
problem construction, tolerances, budgets, data preparation, hardware,
and reproduction protocol in \hyperref[sec-s-methodology]{Supplement
S3}.

\subsection{Simulated data}\label{simulated-data}

We begin with the simplest consequence of Section~\ref{sec-theory}. The
gradient strategy's intercept step is the local Newton direction scaled
by \(H_{00}/L_0\). As imbalance pushes \(\mu_0\) toward one, \(H_{00}\)
shrinks while \(L_0 = 1/4\) remains fixed. We therefore expect the
intercept update to slow.

\subsubsection{Synthetic trajectories and
safeguards}\label{synthetic-trajectories-and-safeguards}

We compare the strategies on a standardized binary-logistic design with
\(n = 500\), \(p = 1000\), \(s = 10\), AR(1) correlation \(\rho = 0.6\),
cyclic CD, and \(\lambda = 0.05\lambda_{\max}\).

\begin{figure}

\centering{

\pandocbounded{\includegraphics[keepaspectratio]{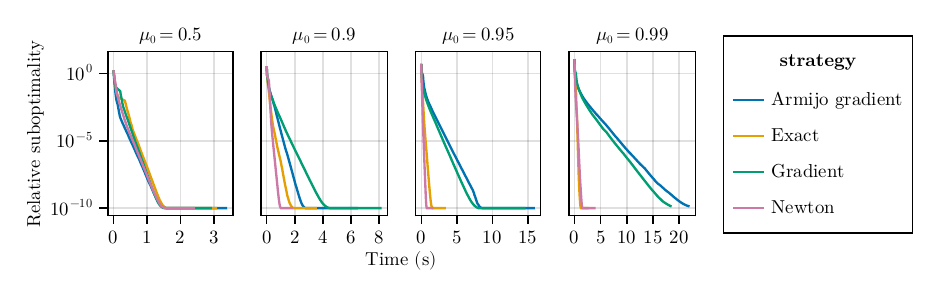}}

}

\caption{\label{fig-mu-extreme}Imbalance sweep on the standardized
logistic design. The panels show a relative suboptimality bound versus
time for the intercept-update variants in the legend as \(\mu_0\)
varies, with \(n = 500\), \(p = 1000\), \(s = 10\),
\(\lambda = 0.05\lambda_{\max}\), and cyclic CD.}

\end{figure}%

The Schur-coupling residual (Equation~\ref{eq-grad-residual}) predicts
that the gradient strategy will slow as \(H_{00}/L_0\) approaches zero
when \(\mu_0\) approaches one. At \(\mu_0 = 0.99\), the gradient
strategy needs about fifteen times as many passes as Newton to reach
\(10^{-8}\) in our experiment, whereas the Newton and exact strategies
remain close across the sweep (Figure~\ref{fig-mu-extreme}). Their
agreement supports the prediction that exact profiling adds little
progress in the outer solve beyond one Newton step
(Equation~\ref{eq-intercept-drift});
\hyperref[sec-per-pass-cost]{Supplement S4.3} compares their inner work.

\begin{figure}

\centering{

\pandocbounded{\includegraphics[keepaspectratio]{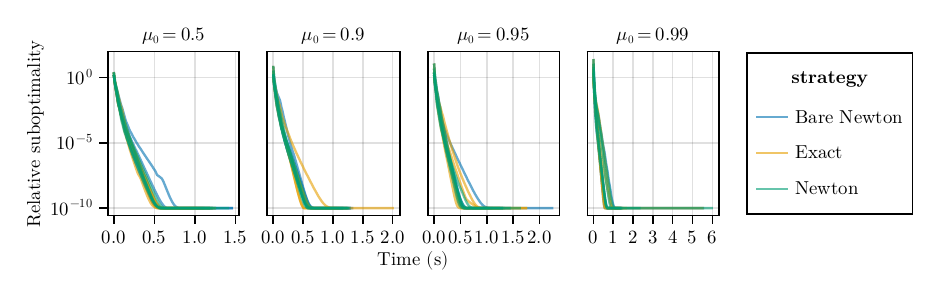}}

}

\caption{\label{fig-cold-start}Logistic cold-start comparison. The
panels show a relative suboptimality bound versus time for bare Newton,
Armijo-guarded Newton, and the exact strategy across the \(\mu_0\)
levels of Figure~\ref{fig-mu-extreme}. Each panel shows five simulated
replicates per strategy.}

\end{figure}%

The Newton-coupling bias (Equation~\ref{eq-newton-coupling}) leads to a
second prediction: at a cold start, a large \(|\partial_0 F|\) can make
an unguarded Newton step on the intercept unstable. To test this
prediction directly, we run bare Newton (without Armijo backtracking)
alongside the guarded variant and the exact strategy in
Figure~\ref{fig-cold-start}.

For logistic regression, the three trajectories are indistinguishable,
even in the most imbalanced case at \(\mu_0 = 0.99\). The Armijo guard
rarely activates: most guarded runs are identical bit for bit to their
bare Newton counterparts, and the remaining pairs begin to differ only
late in the solve, never at cold start. Every guarded and bare Newton
pair crosses the shared \(10^{-8}\) bound on the same pass. The exact
strategy differs from them by only a handful of passes.

Thus, the Armijo safeguard does not improve the logistic cold starts
tested here: the full Newton step is accepted at the initial iterate,
and the few later differences do not systematically favor either
variant. Logistic curvature is bounded, \(f'' = \mu(1 - \mu) \le 1/4\),
but bounded curvature alone does not guarantee that an undamped Newton
step decreases the objective at every iterate. These experiments
establish the narrower empirical result for the zero-initialized problem
family above.

Losses with unbounded curvature behave differently. For Poisson
regression with the canonical log link, \(f''(\eta) = e^\eta\) and
\(f'''(\eta) = e^\eta\) both grow without bound, so the Newton
linearization is loose far from the optimum and the unguarded step can
overshoot by an arbitrary factor. We repeat the cold-start experiment on
synthetic Poisson data with \(n = 500\), \(p = 1000\), \(s = 10\),
\(\lambda = 0.1\lambda_{\max}\), and baseline rates
\(\mu_0 \in \{10, 30, 100, 300\}\).

\begin{figure}

\centering{

\pandocbounded{\includegraphics[keepaspectratio]{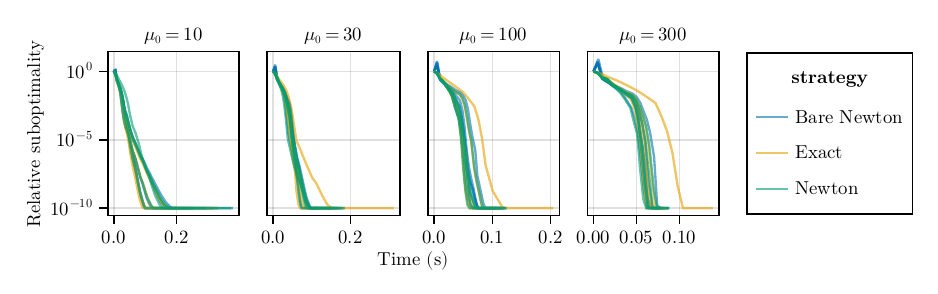}}

}

\caption{\label{fig-cold-start-poisson}Poisson cold-start comparison.
The panels show a relative suboptimality bound versus time for bare
Newton, Armijo-guarded Newton, and the exact strategy at baseline rates
\(\mu_0 \in \{10, 30, 100, 300\}\), with \(n = 500\), \(p = 1000\),
\(s = 10\), and \(\lambda = 0.1\lambda_{\max}\). Each panel shows five
simulated replicates per strategy.}

\end{figure}%

At \(\mu_0 = 300\), the first bare Newton step moves the cold-start
intercept from zero to \(\bar{y} - 1 \approx 300\), although the
intercept-only optimum is only \(\log(\bar{y}) \approx 6\). Because the
log link exponentiates the intercept, this step implies a fitted rate
near \(e^{300}\) and makes the loss astronomical at that intermediate
iterate. The safeguard prevents this overshoot
(Figure~\ref{fig-cold-start-poisson}): the unguarded bound jumps
several-fold after the first pass, while the guarded bound falls
immediately.

The overshoot is nonetheless recoverable. The outer CD loop's own
per-coordinate Armijo guard absorbs it over the next few passes, and the
final primal values agree to numerical precision. The three strategies
return to nearly the same trajectory within a few passes.

The guard therefore keeps the iterates numerically reasonable but
provides no visible improvement in the convergence rate here. The exact
strategy is competitive in both experiments but incurs the cost of
within-pass drift identified in Equation~\ref{eq-intercept-drift}. It
applies the same Armijo safeguard to every inner Newton direction. Every
returned scalar update satisfies \(|\partial_0 F| \le 10^{-10}\).
Reaching the iteration cap would stop the experiment before the result
could be cached.

\subsection{Real data}\label{real-data}

We next examine whether the same scaling appears on real problems. We
run the three strategies on six LIBSVM benchmark problems at
\(\lambda = 0.05\lambda_{\max}\). Table~\ref{tbl-real-logreg-problems}
lists the problems in order of the reciprocal intercept-only curvature
fraction \([4 \pi_+ (1 - \pi_+)]^{-1}\). The table lists the curvature
fraction itself, \(4 \pi_+ (1 - \pi_+)\). The six problems span a
ninefold range in this quantity and include sparse and dense designs in
both dimensional regimes.

The two most imbalanced problems, w1a and news20-3pct, are especially
useful. They have nearly the same \(\pi_+\) but very different sparse
designs: w1a has 300 web-indicator features, while news20-3pct is drawn
from a text corpus with an original vocabulary of roughly \(1.3\)
million terms and retains 4,675 features after filtering. Agreement
between them would suggest that the slowdown is not peculiar to w1a.

\begin{longtable}[]{@{}
  >{\raggedright\arraybackslash}p{(\linewidth - 10\tabcolsep) * \real{0.2321}}
  >{\raggedleft\arraybackslash}p{(\linewidth - 10\tabcolsep) * \real{0.0893}}
  >{\raggedleft\arraybackslash}p{(\linewidth - 10\tabcolsep) * \real{0.0893}}
  >{\raggedright\arraybackslash}p{(\linewidth - 10\tabcolsep) * \real{0.1250}}
  >{\raggedleft\arraybackslash}p{(\linewidth - 10\tabcolsep) * \real{0.1250}}
  >{\raggedleft\arraybackslash}p{(\linewidth - 10\tabcolsep) * \real{0.3393}}@{}}
\caption{LIBSVM problems used in Figure~\ref{fig-real-logreg}. \(\pi_+\)
is the positive-class marginal, and \(4\pi_+(1 - \pi_+)\) is the
intercept-only curvature fraction, not the fitted \(H_{00}/L_0\).
news20-3pct is a deterministic 60/1,940 subsample of
\protect\texttt{news20.binary}; after subsampling, features with fewer
than 20 nonzeros are
dropped.}\label{tbl-real-logreg-problems}\tabularnewline
\toprule\noalign{}
\begin{minipage}[b]{\linewidth}\raggedright
dataset
\end{minipage} & \begin{minipage}[b]{\linewidth}\raggedleft
\(n\)
\end{minipage} & \begin{minipage}[b]{\linewidth}\raggedleft
\(p\)
\end{minipage} & \begin{minipage}[b]{\linewidth}\raggedright
density
\end{minipage} & \begin{minipage}[b]{\linewidth}\raggedleft
\(\pi_+\)
\end{minipage} & \begin{minipage}[b]{\linewidth}\raggedleft
\(4\pi_+(1 - \pi_+)\)
\end{minipage} \\
\midrule\noalign{}
\endfirsthead
\toprule\noalign{}
\begin{minipage}[b]{\linewidth}\raggedright
dataset
\end{minipage} & \begin{minipage}[b]{\linewidth}\raggedleft
\(n\)
\end{minipage} & \begin{minipage}[b]{\linewidth}\raggedleft
\(p\)
\end{minipage} & \begin{minipage}[b]{\linewidth}\raggedright
density
\end{minipage} & \begin{minipage}[b]{\linewidth}\raggedleft
\(\pi_+\)
\end{minipage} & \begin{minipage}[b]{\linewidth}\raggedleft
\(4\pi_+(1 - \pi_+)\)
\end{minipage} \\
\midrule\noalign{}
\endhead
\bottomrule\noalign{}
\endlastfoot
w1a & 2,477 & 300 & sparse & 0.029 & 0.113 \\
news20-3pct & 2,000 & 4,675 & sparse & 0.030 & 0.116 \\
a4a & 4,781 & 123 & sparse & 0.248 & 0.747 \\
leukemia & 38 & 7,129 & dense & 0.711 & 0.823 \\
breast-cancer & 683 & 10 & dense & 0.350 & 0.910 \\
gisette & 6,000 & 5,000 & dense & 0.500 & 1.000 \\
\end{longtable}

\begin{figure}

\centering{

\pandocbounded{\includegraphics[keepaspectratio]{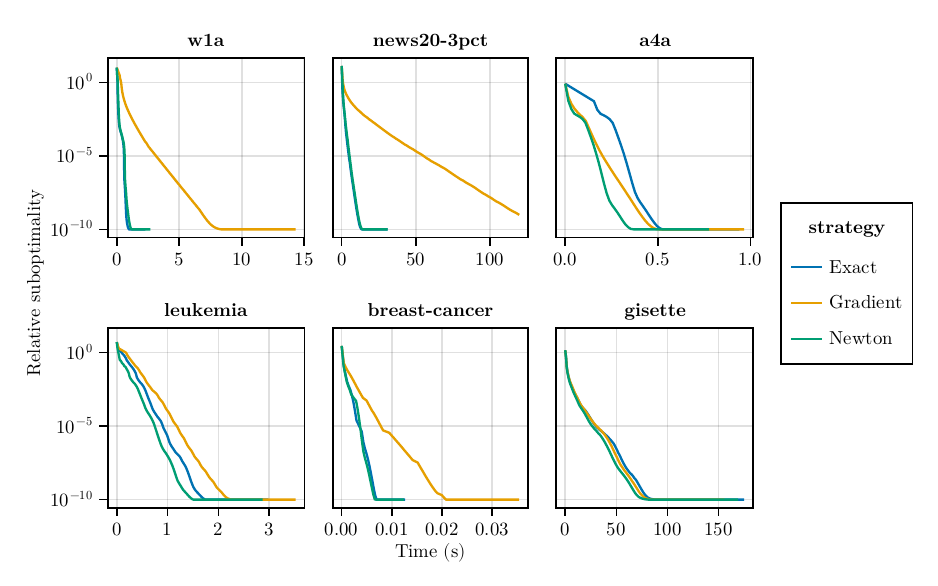}}

}

\caption{\label{fig-real-logreg}Real-data logistic convergence on six
LIBSVM problems. The panels show a relative suboptimality bound versus
time for the three intercept strategies at
\(\lambda = 0.05\lambda_{\max}\), in the problem order of
Table~\ref{tbl-real-logreg-problems}.}

\end{figure}%

The two severely imbalanced datasets reproduce both predicted patterns
on very different sparse designs. The gradient strategy takes about six
times as many passes as Newton to reach \(10^{-6}\) on w1a and about ten
times as many to reach \(10^{-4}\) on news20-3pct, before the latter
reaches its wall-time cap. Both datasets have small intercept-only
curvature fractions. Their large slowdowns across different designs
support the response-imbalance mechanism, although the fitted curvature
and collective coupling determine the local rate.

For the moderately imbalanced a4a and breast-cancer data, the gap
between the gradient and Newton strategies is modest---less than
threefold---consistent with their larger intercept-only curvature
fractions. Leukemia (\(n \ll p\), penalty-dominated) and gisette (dense
and balanced) serve as negative controls: leukemia shows little
separation, and the gisette pass counts are identical. The exact
strategy tracks Newton on every panel, consistent with
Equation~\ref{eq-intercept-drift} at this tolerance.

Poisson data offer a sharper version of the same test. The Poisson loss
has no global upper bound on \(f''(\eta) = e^\eta\), so \(L_0 = \infty\)
and the gradient-strategy update degenerates: the per-pass change is
\(-\sum_i (\hat\mu_i - y_i)/(L_0 n) = 0\), so the intercept remains
unchanged. We test this on the congress109 phrase-count corpus of Taddy
(2015).

The corpus contains \(n = 529\) speakers from the 107th--109th US
Congress and \(p = 999\) political phrases (the 1,000-phrase vocabulary
with the response phrase removed). The response \(y_i\) is speaker
\(i\)'s count of the phrase \texttt{american.people}
(\(\bar y = 11.83\), \(y_{\max} = 396\)), the highest-frequency phrase
in the corpus. The features are the remaining \(999\) phrase counts
after column standardization.

This setup follows the construction for distributed multinomial
regression in Taddy (2015), which approximates a multinomial model over
the vocabulary by fitting an independent Poisson lasso for each term.
The baseline rate \(\bar y \approx 12\) places the problem near the
beginning of the synthetic cold-start sweep in
Figure~\ref{fig-cold-start-poisson} (\(\mu_0 = 10\)). At this rate, the
first bare Newton step is \(\bar y - 1 \approx 10.83\), while the
intercept-only optimum is \(\log(\bar y) \approx 2.47\). We set
\(\lambda = 0.05\lambda_{\max}\) and use cyclic CD with a
\(1\,000\)-pass budget.

\begin{figure}

\centering{

\pandocbounded{\includegraphics[keepaspectratio]{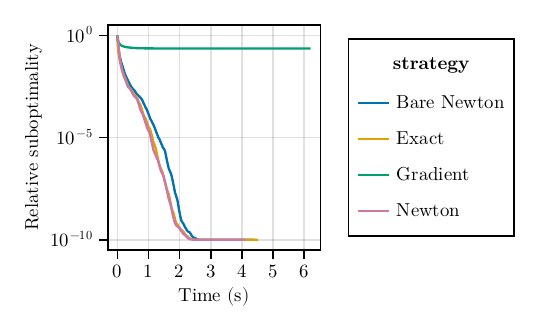}}

}

\caption{\label{fig-real-poisson}Poisson regression on congress109. The
plot shows a relative suboptimality bound versus time for the four
intercept strategies at \(n = 529\), \(p = 999\), \(\bar y = 11.83\),
and \(\lambda = 0.05\lambda_{\max}\).}

\end{figure}%

In Figure~\ref{fig-real-poisson}, the gradient strategy leaves the
intercept unchanged because \(L_0 = \infty\). This failure follows from
the loss family rather than from tuning or a problem-specific imbalance
level. Newton and exact follow the same trajectory because the CD
solver's per-coordinate Armijo loop (Tseng--Yun sufficient decrease,
Algorithm~\ref{algo-cd}) absorbs the cold-start overshoot before the
first intercept update occurs. The synthetic suboptimality curves expose
the unguarded first-pass excursion, but they rejoin the guarded curves
within a few passes at every baseline rate we ran
(Figure~\ref{fig-cold-start-poisson}).

\subsection{Robustness and computational cost}\label{sec-robustness}

The strategy ranking does not depend on the default cyclic ordering. On
w1a, permuting the coordinates changes pass counts by less than a factor
of two. Newton and exact remain fast, and gradient remains slow
(\hyperref[sec-ordering]{Supplement S4.1}). Warm-starting a
regularization path makes Newton and exact still harder to distinguish,
but it does not rescue the gradient strategy on severely imbalanced w1a
(\hyperref[sec-warmstart-path]{Supplement S4.2}).

Instrumenting the intercept solve explains the remaining difference
between Newton and exact. Across \(\mu_0 \in \{0.5, 0.9, 0.99\}\), exact
uses roughly one-and-a-half to twice as many inner work units before
reaching the shared \(10^{-8}\) bound while finishing in essentially the
same number of outer passes. Most of that overhead occurs in the first
cold-start passes. Warm starts largely avoid it
(\hyperref[sec-per-pass-cost]{Supplement S4.3}). The full ordering,
path, and per-pass trajectories appear in
\hyperref[sec-s-robustness]{Supplement S4}.

\subsection{Imbalance and regularization}\label{sec-imbalance-reg}

The preceding single-\(\lambda\) experiments vary one quantity at a
time. Both \(\lambda\) and \(\mu_0\) change \(\hat\eta\) and hence
\(H_{00}\), while the active set also changes the collective coupling in
Equation~\ref{eq-collective-coupling}. We sweep
\(\mu_0 \in \{0.5, 0.7, 0.9, 0.95, 0.99\}\) against
\(\lambda/\lambda_{\max} \in \{0.5, 0.2, 0.1, 0.05, 0.02, 0.01\}\) on
the standardized binary-logistic design of Figure~\ref{fig-mu-extreme}
(three seeds per cell), running each strategy until its local relative
duality gap reaches \(10^{-10}\) or it exhausts the budget of 5000 outer
passes. We report passes to the shared \(10^{-8}\) suboptimality bound
rather than elapsed time to separate the effect of curvature mismatch
from BLAS and memory effects.

\begin{figure}

\centering{

\pandocbounded{\includegraphics[keepaspectratio]{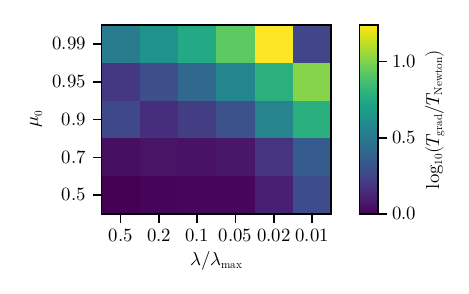}}

}

\caption{\label{fig-mu-reg-gradient}Gradient-to-Newton pass-count ratios
over the \((\mu_0, \lambda/\lambda_{\max})\) grid. The heatmap shows the
log ratio of mean pass counts,
\(\log_{10}(\bar T_{\mathrm{gradient}} / \bar T_{\mathrm{Newton}})\),
with each mean taken over three seeds on the standardized logistic
design; clipped cells mark runs that hit the 5000-pass budget.}

\end{figure}%

\begin{figure}

\centering{

\pandocbounded{\includegraphics[keepaspectratio]{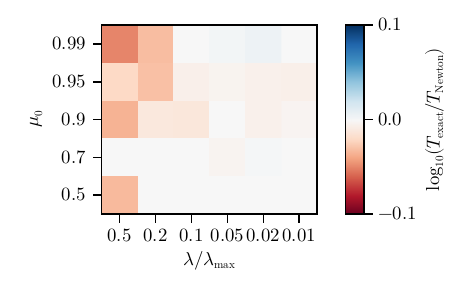}}

}

\caption{\label{fig-mu-reg-exact}Exact-to-Newton pass-count ratios over
the \((\mu_0, \lambda/\lambda_{\max})\) grid. The heatmap shows the log
ratio of mean pass counts,
\(\log_{10}(\bar T_{\mathrm{exact}} / \bar T_{\mathrm{Newton}})\), on
the same grid and with the same three-seed means as
Figure~\ref{fig-mu-reg-gradient}, using a centered diverging colormap
with range \(\pm 0.1\).}

\end{figure}%

Greater imbalance generally makes the gradient strategy slower relative
to Newton (Figure~\ref{fig-mu-reg-gradient}), although the
strongest-penalty column is not strictly monotone. The largest
informative average ratio is roughly \(17\) at \(\mu_0 = 0.99\) and
\(\lambda = 0.02\lambda_{\max}\). The \(\mu_0 = 0.99\),
\(\lambda = 0.01\lambda_{\max}\) corner is the only cell in which any
run exhausts the 5000-pass budget. In one seed, all three strategies hit
the cap without reaching the reporting threshold. The displayed average
ratio in that cell is therefore uninformative rather than small. We
return to it below.

At \(\lambda = 0.05\lambda_{\max}\), the regime most relevant to the
production path experiments, the empirical ratio grows from nearly one
under balance to about eight at \(\mu_0 = 0.99\). The corresponding
fitted curvature ratio \(L_0/H_{00}(\hat\eta)\) rises from about three
to twelve. This ratio gives the strong-coupling limit in
Equation~\ref{eq-frozen-block-rate}. It need not equal the pass-count
ratio at a finite tolerance when other convergence modes contribute.

As \(\lambda\) shrinks, the active set grows, \(H_{00}\) falls to or
below \(\mu_0(1 - \mu_0)\), and the gradient strategy takes more passes
relative to Newton. At \(\lambda = 0.01\lambda_{\max}\), for example,
the ratio rises from about two under balance to about ten at
\(\mu_0 = 0.95\). We omit the \(\mu_0 = 0.99\) corner because all three
strategies reach the cap in one seed, making the average uninformative.

By contrast, Figure~\ref{fig-mu-reg-exact} stays nearly flat at zero:
the ratio of exact to Newton pass counts is approximately one across the
grid, and no cell departs from unity by more than about \(10\%\). This
agrees with Equation~\ref{eq-intercept-drift}. When both methods
converge, the exact strategy needs essentially the same number of outer
passes as Newton. Its additional inner iterations therefore add a modest
overhead that grows with imbalance without reducing the number of outer
passes, as \hyperref[sec-per-pass-cost]{Supplement S4.3} shows.

The two heatmaps support the curvature-mismatch mechanism across
penalties: the empirical ratio generally grows as the intercept update
becomes more conservative (\(H_{00}/L_0 \to 0\)). Both curvature and
coupling change across this grid. The centering sweep of
Figure~\ref{fig-rho-centering} isolates coupling by holding
\(H_{00}/L_0\) fixed. The frozen coordinate-sweep model in
Figure~\ref{fig-rho-frozen} accounts for the transition from similar
rates to the \(L_0/H_{00}\) plateau as the coupled mode becomes slowest.

\subsection{Production solvers}\label{sec-production-solvers}

The production experiments test two aspects of the curvature
explanation. A controlled synthetic family tests how \textbf{glmnet} and
\textbf{biglasso} respond to increasing response imbalance. Their mode
switches change the full IRLS weight scheme. A separate comparison of
two \textbf{skglm} versions changes the intercept update and its
associated stopping residual. Additional diagnostics examine the other
production solvers on the shared synthetic problem, w1a, and
news20-3pct.

We evaluate every fit on the common objective
\(n^{-1}\sum_i \ell_i + \lambda\lVert\beta\rVert_1\). For each problem,
a high-accuracy fit supplies a reference primal \(P_{\mathrm{ref}}\) and
a feasible dual lower bound \(D_{\mathrm{ref}}\). The relative
certificate width
\((P_{\mathrm{ref}} - D_{\mathrm{ref}})/P_{\mathrm{ref}}\) is below
\(10^{-8}\). We use these certificates both to match accuracy in the
imbalance sweep and to bound suboptimality in the remaining diagnostics.

\subsubsection{\texorpdfstring{Response imbalance in \textbf{glmnet} and
\textbf{biglasso}}{Response imbalance in glmnet and biglasso}}\label{response-imbalance-in-glmnet-and-biglasso}

We generate three standardized Gaussian designs with \(n = 2000\),
\(p = 500\), AR(1) correlation \(0.6\), and ten nonzero generating
coefficients of size \(0.5\). Within each seed, we hold the design,
slope coefficients, and uniform random draws fixed. We vary only the
generating intercept, calibrating it so that the mean success
probability is \(\pi \in \{0.5,0.7,0.9,0.95,0.99\}\). The same uniform
draws generate nested binary responses as \(\pi\) increases. Thus, the
sweep changes response imbalance without changing the feature matrix or
slope signal. At the 1\% target minority share, the realized share
ranges from 0.55\% to 1.40\% across seeds.

We compare \texttt{Newton} with \texttt{modified.Newton} in
\textbf{glmnet} and with \texttt{MM} in \textbf{biglasso}. For each
problem, both modes receive the same explicit 50-point geometric path
from \(\lambda_{\max}\) to \(0.05\lambda_{\max}\). This fixes the
relative penalty; its absolute value and the fitted active set can
change with the responses. Before timing, we select each mode's loosest
tolerance on a half-decade grid from \(10^{-3}\) to \(10^{-13}\) whose
final fit satisfies \[
  \frac{P - D_{\mathrm{ref}}}{P_{\mathrm{ref}}} \le 10^{-6}.
\] We then warm each selected fit once and time seven repetitions,
interleaving the modes in a seeded random order. Each timed repetition
must reach the same accuracy target. Timings cover the full
regularization path and exclude input conversion, tolerance calibration,
explicit warm-up, and certificate evaluation. We allow at most \(10^6\)
passes per path in \textbf{glmnet} and \(10^4\) iterations per path
point in \textbf{biglasso}. We compare medians within each problem and
package only when both modes reach the accuracy target. The cached
results also record every configuration that fails to reach it on the
tolerance grid.

\begin{figure}

\centering{

\pandocbounded{\includegraphics[keepaspectratio]{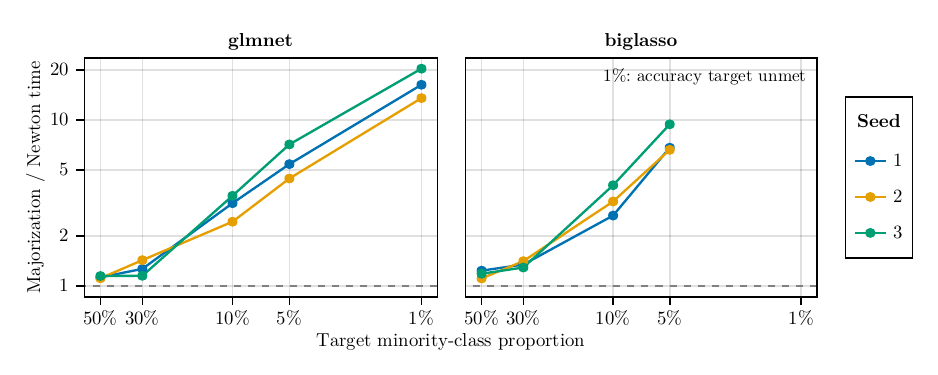}}

}

\caption{\label{fig-production-imbalance-time}Runtime sensitivity to
response imbalance in \textbf{glmnet} and \textbf{biglasso}. Each curve
follows one fixed design and slope signal as the generating intercept
changes. The vertical axis is the ratio of median full-path runtimes for
the majorization mode to the Newton mode, using seven warmed repetitions
per mode. Every fit meets the same certified relative accuracy target of
\(10^{-6}\). Ratios above one favor Newton. The horizontal axis gives
the target minority-class proportion, decreasing from left to right.
Missing comparisons mark problems where a mode did not reach the
accuracy target within the tolerance grid and iteration budget.}

\end{figure}%

The ratio of majorization to Newton runtime rises with imbalance in
every seed. The ratios are close to one under balance. At a 5\% target
minority share, they reach about \(4.4\)--\(7.1\) in \textbf{glmnet} and
\(6.6\)--\(9.4\) in \textbf{biglasso}. At 1\%, \textbf{glmnet}'s ratio
reaches \(13.5\)--\(20.4\). Both \textbf{biglasso} modes miss the
accuracy target on all three 1\% problems within the calibration grid
and iteration budget. We therefore show no runtime ratio for those
cells.

We evaluate the curvature fraction \(H_{00}/L_0\) and collective
coupling \(\kappa\) at the reference solution of each problem
(Figure~\ref{fig-production-imbalance-curvature}). These diagnostics
distinguish the looseness of the global bound from the strength of
intercept--coefficient coupling. The active set used for \(\kappa\)
contains coefficients from the reference solution whose absolute value
exceeds \(10^{-8}\). Timing covers the whole regularization path, so the
runtime ratio also depends on curvature and computational work at
earlier path points.

\begin{figure}

\centering{

\pandocbounded{\includegraphics[keepaspectratio]{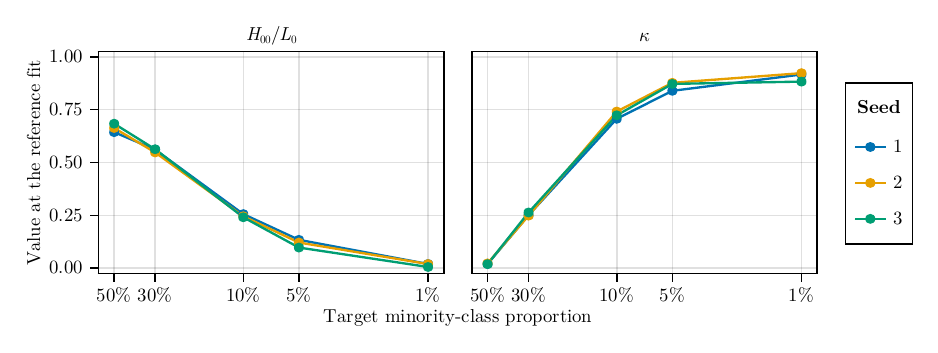}}

}

\caption{\label{fig-production-imbalance-curvature}Curvature and
coupling along the production-solver imbalance sweep. The diagnostics
are evaluated at independently certified reference fits for the same
fifteen problems as Figure~\ref{fig-production-imbalance-time}. The left
panel shows \(H_{00}/L_0\), with \(L_0=1/4\); the right panel shows the
collective active-set coupling \(\kappa\). Each curve follows one seed.}

\end{figure}%

Across seeds, \(H_{00}/L_0\) falls from \(0.64\)--\(0.68\) under balance
to \(0.006\)--\(0.020\) at the 1\% target minority share. The collective
coupling rises from about \(0.02\) to \(0.88\)--\(0.92\). The imbalance
sweep therefore combines a looser global bound with stronger
intercept--coefficient coupling, the regime in which the analysis
predicts the largest difference between curvature schemes.

The mode switches change coefficient curvatures as well as the intercept
correction. This experiment tests the sensitivity of the whole curvature
scheme to imbalance. The following \textbf{skglm} comparison changes
intercept handling within one solver.

\subsubsection{\texorpdfstring{A historical update comparison in
\textbf{skglm}}{A historical update comparison in skglm}}\label{a-historical-update-comparison-in-skglm}

The remaining solver diagnostics use the shared synthetic problem, w1a,
and news20-3pct, with \(\lambda = 0.05\lambda_{\max}\) in each case.
These problems cover a dense Gaussian design, sparse binary features,
and sparse text features. The \textbf{skglm} tolerance sweep fits the
target \(\lambda\) directly, while the \textbf{LIBLINEAR},
\textbf{BlitzL1}, and \textbf{adelie} comparisons use 50-point
regularization paths.\footnote{At the single-\(\lambda\) configuration
  reported in \hyperref[sec-cold-start-diag]{Supplement S5.2},
  \textbf{biglasso} \texttt{alg.logistic\ =\ "Newton"} returns
  inaccurate fits on both extreme-imbalance real datasets, while
  \textbf{adelie} returns an empty state on w1a. These results use loose
  solver tolerances and do not describe either package's default
  configuration.} Each plotted point reports the endpoint of a fit at
one tolerance, evaluated against the shared feasible dual reference
\(D^\star\) for that problem.

\begin{figure}

\centering{

\pandocbounded{\includegraphics[keepaspectratio]{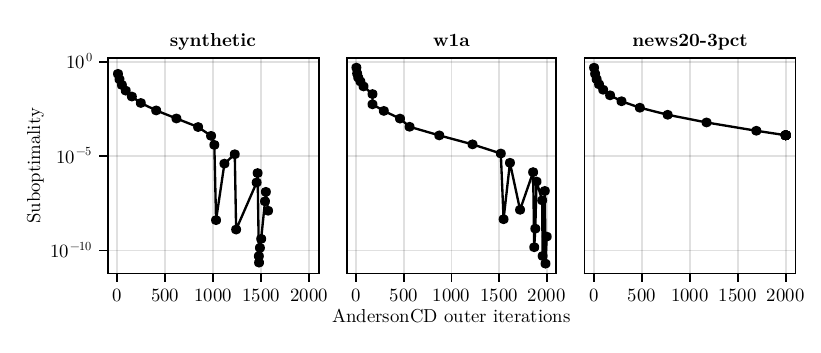}}

}

\caption{\label{fig-real-skglm}\textbf{skglm} 0.5 on three imbalanced
logistic problems. The panels show a suboptimality bound versus
\texttt{AndersonCD} outer iterations on the shared synthetic problem,
w1a, and news20-3pct, using the common feasible dual reference
\(D^\star\) for each problem.}

\end{figure}%

In \textbf{skglm} 0.5, the intercept update takes a conservative
gradient step on the original loss. Writing \(g_0 = \partial_0 F\) for
the averaged-loss gradient, the implementation subtracts \(g_0/4\),
whereas our gradient strategy subtracts \(g_0/L_0 = 4g_0\) with
\(L_0 = 1/4\). The historical step is therefore sixteen times smaller
than the sharp global-curvature step. It corresponds to an effective
curvature bound of \(4\), so its leading-order coupling correction is
\(H_{00}/4\), rather than \(4H_{00}\) (Equation~\ref{eq-grad-residual}).
The tolerance sweep shows slow progress under this additional
conservatism. We submitted a patch to \textbf{skglm} replacing this
gradient update with a Newton step. Our patch was merged upstream as
commit \texttt{8f9cbd77}.\footnote{\url{https://github.com/scikit-learn-contrib/skglm/pull/337},
  merged 2026-05-13. \textbf{skglm} 0.6 is unreleased at the time of
  writing; the controlled comparison below pins the two relevant commits
  directly.}

The \textbf{glmnet} and \textbf{biglasso} switches change the full IRLS
weight scheme, and the cross-package panels also reflect implementation
differences. We compare two \textbf{skglm} commits whose runtime code
differs only in the intercept callback
(Figure~\ref{fig-skglm-controlled}): \texttt{b03644fe}, the last pre-fix
commit and behaviorally identical to the released 0.5, and
\texttt{8f9cbd77}, the merged Newton fix. Running the same
\texttt{AndersonCD} driver on both yields a controlled 30-cell
comparison. The callback also defines the intercept component of
\texttt{AndersonCD}'s stopping criterion: it changes from \(|g_0|/4\) to
\(|g_0|/H_{00}\). Thus, the common numerical tolerance does not impose
the same intercept accuracy in the two versions. This comparison
measures the effect of the historical patch, including its change to
stopping semantics; it does not isolate Newton against the correctly
scaled global-curvature strategy at matched accuracy. The patch uses
bare Newton, so it does not test the Armijo safeguard. The
internal-solver experiments compare the correctly scaled gradient
strategy with safeguarded Newton under a shared objective certificate.
Sources for the historical callback and its use in the stopping
criterion are given in Supplement S1.1.

As in Figure~\ref{fig-mu-reg-gradient}, the ratio generally grows with
\(\mu_0\) and with \(1/\lambda\), although it is not monotone in every
cell. The pre-fix commit needs about \(30\) times as many iterations as
the post-fix commit even in the cell with the smallest difference, and
more than \(400\) times as many in the most imbalanced, weakly
regularized corner.

Even the smallest ratio is large: the historical gradient step is
unusually conservative, while the Newton version reaches its native
tolerance in fewer than 20 outer iterations. These ratios describe the
two implementations and should not be read as speedups over our sharp
global-curvature baseline. Many pre-fix runs hit the 5000-iteration cap,
making the largest displayed ratios lower bounds rather than exact
comparisons.

\begin{figure}

\centering{

\pandocbounded{\includegraphics[keepaspectratio]{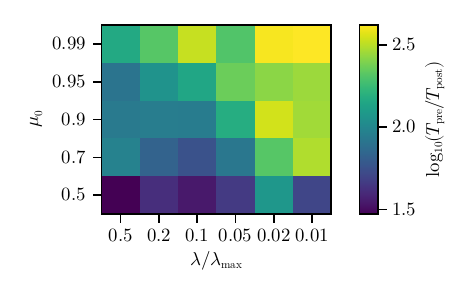}}

}

\caption{\label{fig-skglm-controlled}Within-\textbf{skglm} commit
comparison. The heatmap shows the log ratio of mean iteration counts,
\(\log_{10}(\bar T_{\mathtt{b03644fe}} / \bar T_{\mathtt{8f9cbd77}})\),
for \texttt{AndersonCD} across the same
\((\mu_0,\lambda/\lambda_{\max})\) grid as
Figure~\ref{fig-mu-reg-gradient}, with each mean taken over three seeds;
clipped cells mark runs that hit the 5000-iteration cap. The historical
gradient step is sixteen times smaller than our global-curvature
baseline, and the patch also changes the intercept stopping residual.}

\end{figure}%

\subsubsection{Other production solvers}\label{other-production-solvers}

The prox-Newton solvers provide a complementary test. Resolving the
intercept with local curvature inside each quadratic should avoid the
stall caused by a global majorant. An additional correction on the
original loss may improve the intercept at the subproblem boundary.
\textbf{LIBLINEAR}~(newGLMNET) implements the former, while
\textbf{BlitzL1} implements both. We compare the time they take to
approach the shared certificate on the same three problems in
Figure~\ref{fig-real-proxnewton}, complementing the \textbf{skglm}
iteration-count comparison and the \textbf{glmnet} and \textbf{biglasso}
imbalance sweep.

\begin{figure}

\centering{

\pandocbounded{\includegraphics[keepaspectratio]{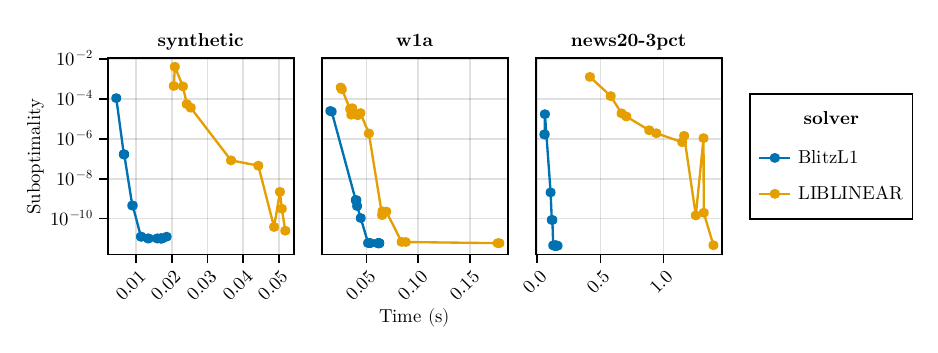}}

}

\caption{\label{fig-real-proxnewton}Prox-Newton solvers on three
imbalanced logistic problems. The panels show a suboptimality bound
versus wall-clock time for \textbf{LIBLINEAR} and \textbf{BlitzL1} on
the shared synthetic problem, w1a, and news20-3pct. \textbf{LIBLINEAR}
runs that hit the internal \texttt{max\_iter} cap without further
progress are omitted. Both solvers use the common feasible dual
reference \(D^\star\) for each problem. Each point is an independent fit
at one tolerance; the single-run timings are indicative.}

\end{figure}%

We test the same prediction in \textbf{adelie}. It builds the same
weighted Gaussian surrogate as \textbf{glmnet} \texttt{Newton}, but it
handles the intercept implicitly: rather than updating \(\beta_0\) as a
coordinate inside the inner CD pass, it centers the features using
weighted means at the start of each surrogate and tracks the residual
mean, reconstructing \(\beta_0 = \bar{y} + \overline{\mathrm{resid}}\)
at the end of the inner solve. By construction, the intercept therefore
reaches the surrogate optimum on every outer pass, with no separate
update step.

That places \textbf{adelie} in the same class as \textbf{glmnet}
\texttt{Newton} and \textbf{biglasso} \texttt{Newton}. In
Figure~\ref{fig-real-adelie}, we sweep its convergence tolerance across
the same three problems used above. The curves are flat for \texttt{tol}
values above roughly \(10^{-4}\): path mode warm-starts each of the 50
lambdas from its predecessor, and those warm starts alone land within
\(10^{-5}\) of the optimum, already meeting the tolerance. Below that
plateau, \textbf{adelie}'s final objectives lie within the shared
reference certificate: the suboptimality bounds are on the order of
\(10^{-11}\) and then become limited by the certificate. Each fit
completes in well under a second, and \textbf{adelie} shows no stall or
failed solve.

\begin{figure}

\centering{

\pandocbounded{\includegraphics[keepaspectratio]{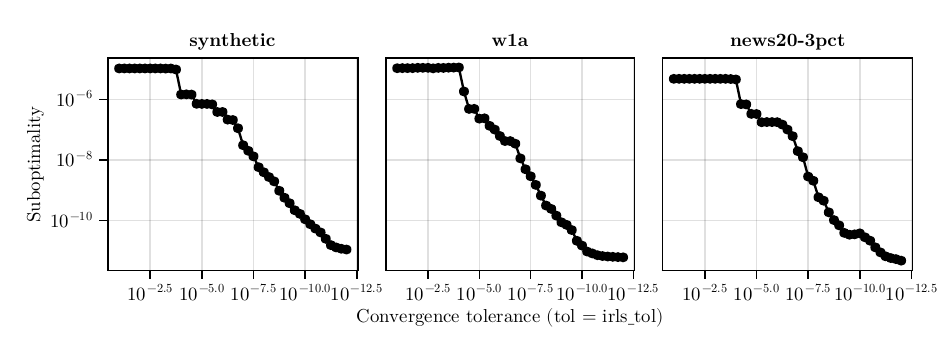}}

}

\caption{\label{fig-real-adelie}\textbf{adelie} on three imbalanced
logistic problems. The panels show a suboptimality bound versus
convergence tolerance on the same three problems used for
\textbf{skglm}, \textbf{LIBLINEAR}, and \textbf{BlitzL1}. The horizontal
axis is the coupled outer and inner tolerance
(\texttt{tol\ =\ irls\_tol}), and the bound uses the common feasible
dual reference \(D^\star\) for each problem.}

\end{figure}%

\subsubsection{Additional solver
diagnostics}\label{additional-solver-diagnostics}

Poisson loss provides a structural check: because \(f''(\eta) = e^\eta\)
is unbounded, the constant-majorant modes have no Poisson analogue. The
production solvers that support Poisson here use local IRLS or
prox-Newton curvature and converge without the gradient-strategy stall
(\hyperref[sec-production-poisson]{Supplement S5.1}).

A separate single-\(\lambda\) diagnostic probes the opposite risk: local
Newton curvature can be unstable at a cold start when solver tolerances
are deliberately loose. The observed \textbf{biglasso} and
\textbf{adelie} failures depend on the dataset and do not describe their
default path-mode calls. We therefore treat them as evidence for
safeguarding the Newton direction, not as a package ranking
(\hyperref[sec-cold-start-diag]{Supplement S5.2}).

\section{Discussion}\label{discussion}

Our analysis yields a practical decision rule: direct coordinate descent
on a smooth GLM loss should update the intercept with one Newton block
step and an Armijo safeguard. Repeatedly optimizing the intercept adds
inner work whose benefit is usually erased by the next coefficient
sweep. Warm-started paths are the main case in which that work becomes
cheap enough to reconsider (\hyperref[sec-warmstart-path]{Supplement
S4.2}). Solvers that freeze a local IRLS or proximal-Newton quadratic
should instead resolve the intercept within that surrogate, where the
quadratic curvature is already available.

The curvature ratio \(H_{00}/L_0\) explains this rule. On the original
nonlinear loss, a global-majorant update removes that fraction of the
local intercept--coefficient coupling to first order
(Equation~\ref{eq-grad-residual}). In the frozen quadratic model, the
fraction is exact. That model then explains how the one-step mismatch
becomes the late-stage \(L_0/H_{00}\) rate plateau under strong
collective coupling (Proposition~\ref{prp-frozen-block-rate},
Figure~\ref{fig-rho-frozen}). This is a local explanation, not a global
iteration-complexity theorem.

The production comparisons test the rule outside our implementation. The
\textbf{glmnet} and \textbf{biglasso} switches hold the package fixed
but change the full IRLS weight scheme, while the before-and-after
\textbf{skglm} comparison changes the intercept callback used for both
updating and stopping (Section~\ref{sec-production-solvers}). The
\textbf{skglm} intervention shows the practical effect of replacing its
unusually conservative historical update with Newton, including the
change to stopping semantics. It does not quantify the advantage over
the correctly scaled global-curvature strategy at matched accuracy. The
other comparisons support the broader prediction about global versus
local curvature, but they do not isolate that update from the rest of
the surrogate or implementation.

The Armijo guard completes the recommendation. It adds little work when
the full Newton step is acceptable and prevents the Poisson cold-start
overshoot in Figure~\ref{fig-cold-start-poisson}. Using it by default
avoids relying on a loss-specific global descent claim.

For multinomial problems, rare nonreference classes and rare reference
classes create low-curvature intercept directions of different forms
(Equation~\ref{eq-multinomial-free-rayleigh}--Equation~\ref{eq-multinomial-reference-rayleigh}).
With long-tailed label distributions, the mechanism that affects the
scalar intercept in binary models therefore routinely affects the
intercept block in multinomial models.

Our analysis has two main limitations. First, our results for the Newton
strategy are \emph{local}. Equation~\ref{eq-newton-error} and
Equation~\ref{eq-newton-coupling} are Taylor identities for a single
Newton step, and Proposition~\ref{prp-frozen-block-rate} freezes the
Hessian and solves the active coefficient block exactly. The frozen
products of permuted coordinate passes reproduce the nonlinear rates,
but we do not prove a general theorem for their Lyapunov exponents or a
global convergence rate for CD with a Newton-updated intercept.

Second, the production-solver imbalance sweep fixes the design family
and \(\lambda/\lambda_{\max} = 0.05\). It includes three seeds at five
prevalence levels, while the additional solver diagnostics use three
benchmark problems. The heatmap in Section~\ref{sec-imbalance-reg}
reproduces the binary direct-CD predictions across a 30-cell grid, but
the \textbf{glmnet} and \textbf{biglasso} comparison does not sweep the
regularization axis. Future work should extend that comparison and check
the multinomial predictions against \textbf{glmnet}'s
symmetric-multinomial path and \textbf{adelie}'s multi-response solver.

\section*{Funding}\label{funding}
\addcontentsline{toc}{section}{Funding}

This research received no external funding.

\section*{Supplementary material}\label{supplementary-material}
\addcontentsline{toc}{section}{Supplementary material}

The supplement collects the technical derivations, extended experiments,
and implementation details that support the main article. It follows the
order of the argument and keeps each figure beside the methods and
interpretation that it supports.

\subsection*{S1. Technical derivations and rate
diagnostics}\label{sec-s-theory}
\addcontentsline{toc}{subsection}{S1. Technical derivations and rate
diagnostics}

\subsubsection*{S1.1 Frozen quadratic rates and
validation}\label{sec-s-rate}
\addcontentsline{toc}{subsubsection}{S1.1 Frozen quadratic rates and
validation}

The Schur analysis describes how an intercept update changes the next
coefficient gradient. To study convergence over repeated passes, we also
need the order of the updates and the errors they leave behind. We
specialize to the \(\ell_1\) penalty used in the experiments, freeze the
Hessian at the optimum, restrict the coefficients to the active set
\(A\), and translate the optimum to zero. Once the active signs are
fixed, the penalty is locally affine and stationarity cancels the linear
terms. For the intercept and active coefficient error vector \(e\), the
quadratic model is \(Q(e) = \tfrac12 e^\top H e\), where \[
  H = \begin{bmatrix}
        H_{00} & H_{0A} \\
        H_{A0} & H_{AA}
      \end{bmatrix}.
\] The iterates move while this Hessian stays fixed. Assume \(H\) is
positive definite. The collective coupling between the intercept and the
active coefficient block is \(\kappa\), restating
Equation~\ref{eq-collective-coupling}: \[
  \kappa = \frac{H_{0A}H_{AA}^{-1}H_{A0}}{H_{00}} \in [0, 1).
\] The inverse coefficient Hessian accounts for joint coefficient
adjustments, so many individually weak intercept--coefficient couplings
can combine. Eliminating the coefficient block gives the Schur
complement \(H_{00} - H_{0A}H_{AA}^{-1}H_{A0} = H_{00}(1 - \kappa)\),
the intercept curvature after the coefficients have adjusted.

\begin{proposition}[]\protect\hypertarget{prp-frozen-block-rate}{}\label{prp-frozen-block-rate}

\emph{(Frozen two-block rate.)} Let \(q = H_{00}/L_0 \in (0,1]\). On the
frozen quadratic, suppose \(H\) is positive definite and each pass first
minimizes over the active coefficient block and then updates the
intercept. The Newton and gradient strategies contract the intercept
error by \[
  r_{\mathrm{N}} = \kappa, \qquad r_{\mathrm{G}} = 1 - q(1 - \kappa),
\] respectively. Exact coefficient-block minimization erases the
incoming coefficient error, so the only potentially nonzero eigenvalue
of each pass is its intercept contraction factor. For fixed
\(0 < \kappa < 1\) and nonzero initial intercept error, let
\(T_{\mathrm{G}}(\varepsilon)\) and \(T_{\mathrm{N}}(\varepsilon)\)
count passes to reduce its magnitude by a factor \(\varepsilon\). Then
\begin{equation}\protect\phantomsection\label{eq-frozen-block-rate}{
  \lim_{\varepsilon \to 0} \frac{T_{\mathrm{G}}(\varepsilon)}{T_{\mathrm{N}}(\varepsilon)}
  = \frac{\log \kappa}{\log \left( 1 - q(1 - \kappa) \right)}
  \longrightarrow \frac{1}{q}
  = \frac{L_0}{H_{00}} \quad \text{as } \kappa \to 1.
}\end{equation} For fixed \(q < 1\), this asymptotic ratio diverges as
\(\kappa \to 0\). At a fixed tolerance, however, Newton needs only one
pass once \(\kappa \leq \varepsilon\), and the gradient pass count stays
finite. These are different orders of limits. At \(\kappa = 0\), Newton
solves the decoupled quadratic in one pass. When \(q = 1\), the two
strategies coincide for every \(\kappa\).

\end{proposition}

To see this, let \(a\) denote the intercept error. Conditional
minimization gives \(\beta_A^+ = -H_{AA}^{-1}H_{A0}a\), after which the
intercept gradient is \(H_{00}(1 - \kappa)a\). Newton therefore leaves
error \(\kappa a\), whereas the gradient strategy leaves
\([1 - q(1 - \kappa)]a\). Taking logarithms gives
Equation~\ref{eq-frozen-block-rate}; expanding both logarithms at
\(\kappa = 1\) gives the limit.

We test the role of coupling by holding \(\mu_0\) fixed and shifting the
feature means. For the horizontal axis of
Figure~\ref{fig-rho-centering}, we summarize the coordinate-wise
couplings by their curvature-weighted mean, \[
  \bar\rho^2 = \frac{\sum_{j=1}^p H_{jj}\rho_{0j}^2}{\sum_{j=1}^p H_{jj}}.
\] This descriptive measure records coupling to individual coefficients;
\(\kappa\) also accounts for their joint adjustment through
\(H_{AA}^{-1}\).

On the synthetic logistic design with \texttt{means\ =\ :random}, each
column carries an \(O(1)\) mean. A centering fraction \(\alpha\)
interpolates between the standardized and uncentered designs,
\begin{equation}\protect\phantomsection\label{eq-centering-sweep}{
  X(\alpha) = \bigl(X - (1 - \alpha)\,\bar{x}^\top\bigr) / s,
  \qquad
  \alpha \in [0, 1],
}\end{equation} with the column scales \(s\) held fixed. Shifting a
column by a constant is absorbed exactly by the unpenalized intercept,
so \(\beta^\star\), \(\hat\eta\), \(H_{00}\), and \(\lambda_{\max}\) are
all invariant along the sweep: \(H_{00}/L_0\) remains about \(0.30\)
when \(\mu_0 = 0.5\) and \(0.17\) when \(\mu_0 = 0.9\), even as
\(\bar\rho^2\) increases from nearly zero to roughly \(0.5\). Thus, the
sweep changes the coupling without changing the curvature ratio.

In Figure~\ref{fig-rho-centering}, \(T_{\mathrm{G}}\) and
\(T_{\mathrm{N}}\) count passes to the shared relative suboptimality
bound \(10^{-6}\). Their empirical ratio rises sharply as coupling
increases from near zero and then settles near the strong-coupling limit
\(L_0/H_{00}\) in Equation~\ref{eq-frozen-block-rate}. To account for
both the transition and the plateau, we retain the coefficient error
modes of the actual coordinate sweep in the calculation below.

\begin{figure}

\centering{

\pandocbounded{\includegraphics[keepaspectratio]{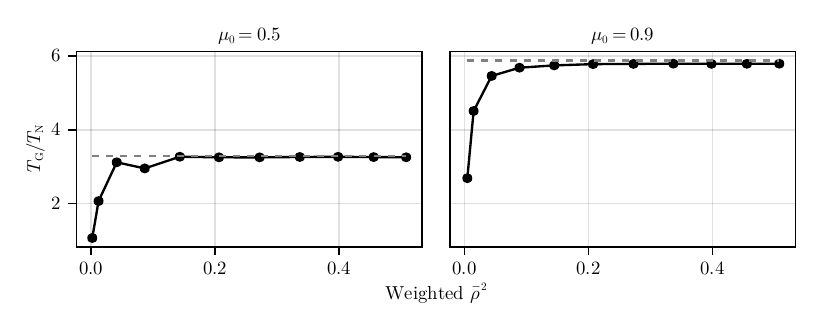}}

}

\caption{\label{fig-rho-centering}Centering sweep at fixed
\(H_{00}/L_0\). The panels show empirical
\(T_{\mathrm{G}}/T_{\mathrm{N}}\) to the relative suboptimality bound
\(10^{-6}\); the dashed line marks the strong-coupling limit
\(L_0/H_{00}\) from Equation~\ref{eq-frozen-block-rate}. The centering
fraction \(\alpha\) in Equation~\ref{eq-centering-sweep} moves the
weighted \(\bar\rho^2\) from standardized to uncentered designs. Here
\(n = 500\), \(p = 1000\), \(s = 10\), \(\lambda = 0.05\lambda_{\max}\),
and the coordinate ordering is permuted.}

\end{figure}%

The centering experiment uses one permuted coordinate sweep per pass,
which retains coefficient error modes absent from the proposition's
exact block solve. Let \(C_\pi\) be the frozen linear map for a
coefficient sweep in order \(\pi\), and let \[
  E_{\mathrm{N}} = I - \frac{1}{H_{00}}e_0e_0^\top H,
  \qquad
  E_{\mathrm{G}} = I - \frac{1}{L_0}e_0e_0^\top H = (1 - q)I + qE_{\mathrm{N}}.
\] One frozen pass is \(P_{\mathrm{N},\pi} = E_{\mathrm{N}}C_\pi\) or
\(P_{\mathrm{G},\pi} = E_{\mathrm{G}}C_\pi\). Because we draw a new
\(\pi\) each pass, the local rate is governed by products of these
random matrices, not by the spectral radius of one cyclic pass.

We compare the top random-product rates with the late-stage slopes of
the nonlinear solver in Figure~\ref{fig-rho-frozen}. We fit each
nonlinear \(\log(\text{relative suboptimality})\) curve between
\(10^{-4}\) and \(10^{-8}\), freeze \(H\) at the converged Newton
solution, and propagate the corresponding random pass matrices for
20,000 passes. The twelve cells include the transition at \(\rho = 0.6\)
and a second \(\rho = 0.2\) sweep that moves \(L_0/H_{00}\) from roughly
\(5\) to \(18\).

\begin{figure}

\centering{

\pandocbounded{\includegraphics[keepaspectratio]{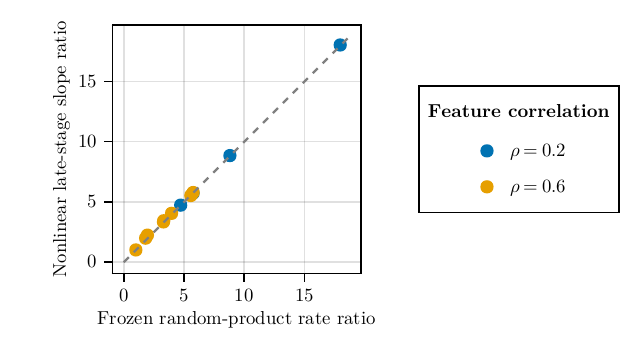}}

}

\caption{\label{fig-rho-frozen}Frozen-pass validation of the centering
mechanism. Each point compares the gradient/Newton ratio of the frozen
random-product log rates with the ratio of fitted late-stage
log-suboptimality slopes from the nonlinear solver, both under permuted
coordinate ordering. The cells cover both the weak-to-strong coupling
transition at \(\rho = 0.6\) and a second \(\rho = 0.2\) imbalance
sweep. The dashed line is equality.}

\end{figure}%

The frozen and nonlinear rates track each other across the transition.
In the \(\rho = 0.6\) sweep, \(\kappa\) rises from about \(0.4\) to
almost one, while the frozen and nonlinear rate ratios both rise from
about one to a little above three. In the \(\rho = 0.2\) sweep, the two
ratios again agree closely and reach roughly \(18\). Newton accepts the
full intercept step on every late-stage pass in these cells, so its
safeguard does not create the difference.

The rates change in two stages. Under weak coupling, an ordinary
coefficient mode controls both strategies, and their rate ratio stays
near one. As collective coupling approaches one, the coupled mode
becomes slowest; relaxing its intercept correction by \(q\) then
produces the \(1/q = L_0/H_{00}\) plateau. We have proved the
exact-block limit in Proposition~\ref{prp-frozen-block-rate} and
verified the full permuted-sweep operator numerically. A general
Lyapunov-exponent theorem for the random matrix products is beyond our
present analysis.

The curvature mismatch also appears in
Section~\ref{sec-production-solvers}. In \textbf{glmnet}'s
\texttt{modified.Newton} mode and \textbf{biglasso}'s \texttt{MM} mode,
the first intercept correction at each linearization uses the global
\(L_0 = 1/4\)~(Section~\ref{sec-irls}). The runtime ratios in
Figure~\ref{fig-production-imbalance-time} measure the practical effect
of changing the weight scheme across imbalance levels. The diagnostics
in Figure~\ref{fig-production-imbalance-curvature} track both the
shrinking curvature fraction and the changing coupling; the mode
switches also change the coefficient updates and the cost per pass.
\textbf{skglm} (version 0.5) applies a more conservative gradient update
directly to the original loss. Its slow convergence from suboptimality
\(\approx 10^{-1}\) to \(\approx 10^{-5}\) in
Figure~\ref{fig-real-skglm} accompanies the intercept's slow climb
toward its fitted value at step size \(1/4\).

An Armijo line search cannot rescue the gradient strategy when it starts
from the global-Lipschitz step \(-\partial_0 F / L_0\). It can only
\emph{shrink} the step, which is the wrong direction when
\(H_{00} \ll L_0\): the step is already too small, not too large. To
recover the iterate-level curvature \(H_{00}\) a line search would need
to \emph{expand} the step beyond \(1/L_0\), settling near \(1/H_{00}\).
But the step of size \(1/H_{00}\) applied to the gradient direction is
exactly the Newton step, so rescuing the gradient strategy turns it into
Newton. A backtracking-Armijo variant of the gradient strategy therefore
stalls identically to the bare-Lipschitz version: the Lipschitz step
satisfies the sufficient-decrease condition on the first attempt, and
the line search never fires (Figure~\ref{fig-mu-extreme}).

At the time of these experiments, \textbf{skglm}~(version 0.5)~(Bertrand
et al. 2022) used constant per-coordinate Lipschitz steps for the
coefficients, but an intercept step of \(1/(4n)\) on the unaveraged
loss. On the averaged loss, its update is \(-g_0/4\), sixteen times
smaller than the \(-4g_0\) step defined by our gradient strategy. It
uses no line search.\footnote{The pinned
  \href{https://github.com/scikit-learn-contrib/skglm/blob/b03644feedd6bac26ac0d5a4f8e0dc9a291a41cc/skglm/datafits/single_task.py\#L372-L373}{\texttt{Logistic.intercept\_update\_step}}
  returns \texttt{mean(-y\ *\ sigmoid(-y\ *\ Xw))\ /\ 4}.
  \href{https://github.com/scikit-learn-contrib/skglm/blob/b03644feedd6bac26ac0d5a4f8e0dc9a291a41cc/skglm/solvers/anderson_cd.py\#L119-L128}{\texttt{AndersonCD}}
  uses the absolute callback value in its stopping criterion and
  subtracts the callback value in the
  \href{https://github.com/scikit-learn-contrib/skglm/blob/b03644feedd6bac26ac0d5a4f8e0dc9a291a41cc/skglm/solvers/anderson_cd.py\#L162-L166}{intercept
  update}.} The other CD-based production packages we surveyed
(\textbf{glmnet}, \textbf{biglasso}, \textbf{LIBLINEAR}, and
\textbf{BlitzL1}) embed the intercept update inside an outer IRLS or
proximal-Newton linearization (Section~\ref{sec-irls}). None of the
production CD solvers we surveyed use a growth-allowing line search on
the intercept alone.

\subsection*{S2. Vector intercepts and multinomial
models}\label{sec-s-multinomial}
\addcontentsline{toc}{subsection}{S2. Vector intercepts and multinomial
models}

\subsubsection*{S2.1 Derivation and
experiments}\label{sec-vector-intercepts}
\addcontentsline{toc}{subsubsection}{S2.1 Derivation and experiments}

We can extend the argument to multinomial logistic regression with \(K\)
classes, where the scalar curvature ratio becomes a matrix operator. The
intercept is a vector \(\beta_0 \in \mathbb{R}^{K-1}\) under the
reference-class parameterization \(\eta_{iK} = 0\). Writing
\(\eta_{ik} = x_i^\top \beta_k + \beta_{0k}\) for
\(k = 1, \dots, K - 1\) and
\(p_{ik} = e^{\eta_{ik}}/\bigl(1 + \sum_{l<K} e^{\eta_{il}}\bigr)\), the
per-observation Hessian block has diagonal \(p_{ik}(1 - p_{ik})\) and
off-diagonal \(-p_{ik} p_{il}\). The intercept Hessian
\(H_{00} = (1/n)\sum_i \bigl(\operatorname{diag}(p_{i,1:K-1}) - p_{i,1:K-1} p_{i,1:K-1}^\top\bigr)\)
is a \((K - 1) \times (K - 1)\) matrix.

The Newton strategy therefore becomes a small block solve
\(\delta = -H_{00}^{-1}\nabla_0 F\) followed by Armijo backtracking on
the loss along that direction. The gradient identity in
Proposition~\ref{prp-profile-equiv}, the within-pass drift
(Equation~\ref{eq-intercept-drift}), and the Newton-coupling bias
(Equation~\ref{eq-newton-coupling}) all carry through with matrix-valued
\(H_{00}\) in place of the scalar.

To derive the matrix analogue of Equation~\ref{eq-grad-residual}, let
\(g_0 = \nabla_0 F\), collect the coefficient gradients in \(g_\beta\),
and let \(H_{\beta 0}\) be the coefficient--intercept cross block. For
an intercept step \(\beta_0 \mapsto \beta_0 - A g_0\), a Taylor
expansion gives \[
  g_\beta^+ = g_\beta - H_{\beta 0} A g_0 + O \big(\lVert A g_0\rVert^2\big).
\] Newton uses \(A = H_{00}^{-1}\), while the gradient strategy uses
\(A = L_0^{-1}I\). Dropping the common second-order remainder gives
\begin{equation}\protect\phantomsection\label{eq-grad-residual-matrix}{
  \begin{aligned}
    g_\beta^{\mathrm{grad}}
     & = g_\beta - H_{\beta 0}H_{00}^{-1}Q g_0,
     & Q                                        & = \frac{H_{00}}{L_0}, \\
    g_\beta^{\mathrm{grad}} - g_\beta^{\mathrm{prof}}
     & = H_{\beta 0}H_{00}^{-1}(I - Q)g_0,
     & g_\beta^{\mathrm{prof}}
     & = g_\beta - H_{\beta 0}H_{00}^{-1}g_0.
  \end{aligned}
}\end{equation} Thus \(Q\), rather than the individual ratios
\(H_{00,kk}/L_0\), is the matrix analogue of the scalar correction
fraction. Because the multinomial majorant satisfies
\(0 \preceq H_{00} \preceq L_0 I\), write
\(H_{00} = U\operatorname{diag}(\lambda_r)U^\top\). The gradient
strategy then retains the fraction \(\lambda_r/L_0 \in [0,1]\) of
Newton's correction in eigenmode \(u_r\). The cross block \(H_{\beta0}\)
subsequently maps that attenuated correction into the coefficient
gradients. A diagonal entry alone does not describe this coupled block
solve.

Rare classes nevertheless force small eigenvalues. Put \(q = K - 1\).
For a nonreference class \(k < K\), the Rayleigh quotient in its
coordinate direction is
\begin{equation}\protect\phantomsection\label{eq-multinomial-free-rayleigh}{
  e_k^\top H_{00} e_k = \frac{1}{n}\sum_i p_{ik}(1 - p_{ik})
                      \leq \bar p_k(1 - \bar p_k).
}\end{equation} Here \(\bar p_k = n^{-1}\sum_i p_{ik}\); at intercept
stationarity, it equals the empirical class prevalence. For the
reference class, the corresponding direction changes all free-class log
odds together. With \(v = \mathbf{1}_q/\sqrt q\),
\begin{equation}\protect\phantomsection\label{eq-multinomial-reference-rayleigh}{
  v^\top H_{00}v = \frac{1}{qn}\sum_i p_{iK}(1 - p_{iK})
                 \leq \frac{\bar p_K(1 - \bar p_K)}{q}.
}\end{equation} In either case, the smallest eigenvalue is no larger
than the displayed Rayleigh quotient and therefore vanishes with the
rare-class prevalence. When the probabilities are constant across
observations, this structure is especially explicit. If
\(D = \operatorname{diag}(\bar p_1, \ldots, \bar p_q)\) and
\(\bar p = (\bar p_1, \ldots, \bar p_q)^\top\), then
\begin{equation}\protect\phantomsection\label{eq-multinomial-intercept-inverse}{
  H_{00} = D - \bar p\bar p^\top,
  \qquad
  H_{00}^{-1} = D^{-1} + \frac{1}{\bar p_K}\mathbf{1}_q\mathbf{1}_q^\top.
}\end{equation} The first term adapts to a rare free class, while the
second adapts the common shift to a rare reference class. The
global-Lipschitz gradient step does not adapt to either. With
\(K = 10\), even a uniform marginal is only \(0.1\), and typical text or
image-classification label distributions skew further. Consequently,
low-curvature intercept modes arise naturally in multinomial models and
are not confined to extreme class imbalance.

We use coordinate descent for the class coefficients and a block update
for the intercept. For each feature \(j\), we cycle through the
\(K - 1\) class coefficients using elementwise \(\ell_1\)
soft-thresholding, refreshing the softmax probabilities after accepted
updates before moving to the next class. Each coordinate step therefore
uses the current gradient and diagonal curvature. We update the
intercept vector once per pass. An alternative is to cycle through one
class at a time while holding the others fixed, as in \textbf{glmnet}'s
symmetric-multinomial path. We return to that scheme after presenting
the block-CD results.

In the synthetic experiments, \(\pi_k\) denotes the baseline probability
used to set \(\beta_{0k} = \log(\pi_k/\pi_K)\) in the data generator.
The feature effects then change the probabilities through the softmax.
Thus \(\pi_k\) is the probability at zero feature contribution, not a
marginal prevalence. The feature means are themselves seeded random
draws. Both these means and the feature amplitude affect the realized
class counts. We reserve \(\bar p_k\) for the average fitted probability
in the Rayleigh bounds above.

\begin{figure}

\centering{

\pandocbounded{\includegraphics[keepaspectratio]{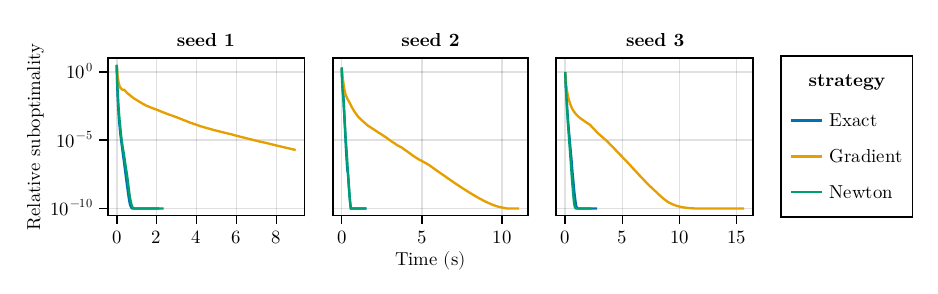}}

}

\caption{\label{fig-multinomial-imbalance}Imbalanced multinomial
logistic regression. The panels show a relative suboptimality bound
versus time for \(K = 5\), baseline probabilities
\((0.7, 0.1, 0.1, 0.05, 0.05)\), \(n = 500\), \(p = 200\),
\(\lambda = 0.05\lambda_{\max}\), cyclic CD, and three random seeds.}

\end{figure}%

We report a relative suboptimality bound
(Figure~\ref{fig-multinomial-imbalance}). The multinomial CD solver
constructs a feasible dual point by blending the centered residual with
the intercept-only feasible point until every implied probability lies
in the simplex, then rescaling onto the \(\ell_\infty\) ball
\(\lVert X^\top \Theta \rVert_\infty \le \lambda\). Within each problem,
we score every run against the largest dual value attained by any
strategy. A stalling run is therefore assessed against a shared lower
bound, and the curves cannot fall below the resolution of that
certificate.

The realized counts in Table~\ref{tbl-multinomial-counts} differ
substantially from the baseline probabilities. At intercept
stationarity, \(\bar p_k = n_k/n\), so the observed counts give the
relevant bounds, reported here to three significant digits. For class 4,
the upper bounds in Equation~\ref{eq-multinomial-free-rayleigh} are
\(0.0493\), \(0.0493\), and \(0.128\) across the three seeds. For the
reference class, Equation~\ref{eq-multinomial-reference-rayleigh} gives
\(0.00345\), \(0.0119\), and \(0.0132\). Dividing the latter by the
Böhning majorant \(L_0 = 1/2\) bounds the smallest eigenvalue of \(Q\)
by \(0.00690\), \(0.0238\), and \(0.0264\), respectively. These realized
prevalences therefore still imply a strongly attenuated gradient
correction in at least one intercept eigenmode near stationarity
(Equation~\ref{eq-grad-residual-matrix}).

\begin{table}

\caption{\label{tbl-multinomial-counts}Realized class counts for
Figure~\ref{fig-multinomial-imbalance}, with \(n=500\) and baseline
probabilities \((0.7,0.1,0.1,0.05,0.05)\). The seeded input
reconstruction in \texttt{experiments/sim-multinomial-prevalence.jl}
reports counts for all sweep cells in
\texttt{results/sim-multinomial-prevalence.csv}. }

\centering{

\begin{adjustbox}{max width=\linewidth}
\begin{tabular}{@{}rrrrrr@{}}

\toprule\noalign{}
Seed & Class 1 & Class 2 & Class 3 & Class 4 & Class 5 (reference) \\
\midrule\noalign{}

1 & 397 & 66 & 4 & 26 & 7 \\
2 & 367 & 7 & 75 & 26 & 25 \\
3 & 142 & 201 & 54 & 75 & 28 \\
\bottomrule\noalign{}
\end{tabular}
\end{adjustbox}

}

\end{table}%

Across the three seeds, Newton and exact bring the shared relative
suboptimality bound below \(10^{-8}\) in \(33\)--\(58\) passes. The
gradient strategy needs \(453\) and \(640\) passes in two seeds and
misses that threshold in the third within the thousand-pass budget.
Thus, the gradient update requires roughly an order of magnitude more
outer work on these problems.

The Newton and exact strategies follow the same visible curve. The
\((K - 1)\)-block Newton solve uses the inverse of the current intercept
Hessian, including the off-diagonal coupling
\(H_{00,kl} = -(1/n)\sum_i p_{ik}p_{il}\). The explicit inverse in
Equation~\ref{eq-multinomial-intercept-inverse} describes the special
case of probabilities constant across observations. In the fitted model,
the full Hessian solve likewise rescales the low-curvature eigenmodes
that the gradient strategy attenuates. The Armijo guard supplies descent
during the cold-start transient; full-step local accuracy applies when
damping is inactive (Equation~\ref{eq-damped-newton-error}).

We extend the test to a two-dimensional grid by shrinking the
reference-class baseline probability \(\pi_K\) to \(0.005\) and
increasing the feature amplitude to four times the baseline value. The
four corners appear in Figure~\ref{fig-multinomial-sweep}; across the
full sweep (Figure~\ref{fig-multinomial-sweep-full}), the Newton and
exact trajectories remain close. Both bring the median bound below
\(10^{-6}\) in all sixteen cells, while the gradient strategy generally
slows as \(\pi_K\) shrinks and amplitude grows. Seven individual
replicates have weak certificates and no reference-class observations.
Their trajectories do not isolate the intercept update, and the cell
medians conceal these failures, as detailed in Supplement S2.2.

\begin{figure}

\centering{

\pandocbounded{\includegraphics[keepaspectratio]{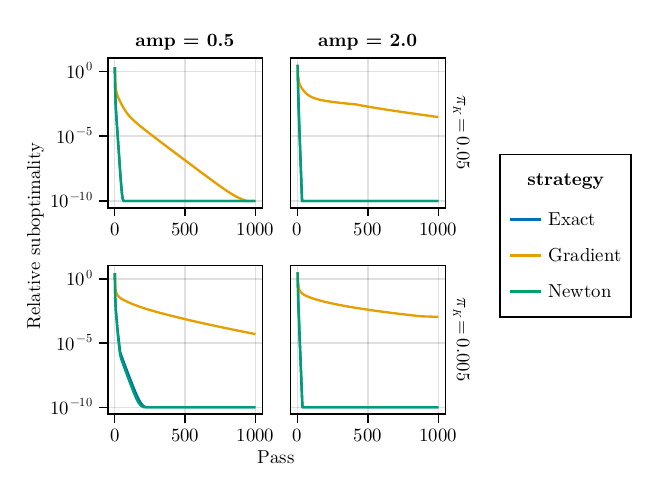}}

}

\caption{\label{fig-multinomial-sweep}Multinomial imbalance-amplitude
sweep. The panels show a relative suboptimality bound versus outer-pass
index for the four corners of the reference-class baseline probability
\(\pi_K \in \{0.05, 0.005\}\) and amplitude \(\in \{0.5, 2.0\}\) grid on
the same \(K = 5\), \(n = 500\), \(p = 200\), \(s = 5\),
\(\lambda = 0.05\lambda_{\max}\), \(\rho = 0.3\) design as
Figure~\ref{fig-multinomial-imbalance}; curves are medians over three
seeds, with early-terminated trajectories carried forward at their
terminal values.}

\end{figure}%

For multinomial models, we therefore recommend a single Armijo-guarded
Newton block step. The line search supplies the damping safeguard needed
at a scalar cold start.

The same pattern holds on real high-dimensional multinomial data. The
Yeoh2002 pediatric ALL gene-expression panel~(Yeoh et al. 2002) is a
\(K = 6\) subtype problem on \(n = 248\) samples and \(p = 12\,625\)
genes with class marginals BCR~\(= 6\%\), MLL~\(= 8\%\), E2A~\(= 11\%\),
T~\(= 17\%\), Hyperdip~\(= 26\%\), and TEL~\(= 32\%\). The two rare
classes BCR and MLL are nonreference classes, so they create
low-curvature directions in the intercept block
(Equation~\ref{eq-multinomial-free-rayleigh}). The gradient strategy
should make little progress along those directions
(Equation~\ref{eq-grad-residual-matrix}).

\begin{figure}

\centering{

\pandocbounded{\includegraphics[keepaspectratio]{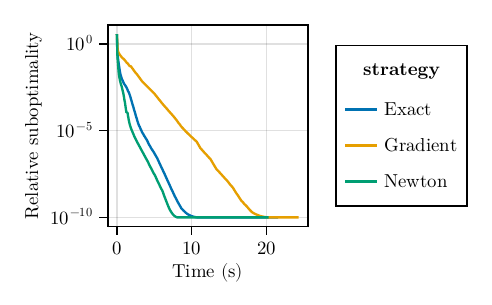}}

}

\caption{\label{fig-real-multinomial}Multinomial logistic regression on
Yeoh2002. The plot shows a relative suboptimality bound versus time at
\(\lambda = 0.05\lambda_{\max}\) under cyclic CD for \(n = 248\),
\(p = 12\,625\), and \(K = 6\). Suboptimality is bounded using the
strongest shared feasible dual point.}

\end{figure}%

On the \(p > n\) Yeoh2002 gene-expression problem, the gradient strategy
lags in the low-curvature modes associated with rare classes, while
Newton and exact need the same number of passes to reach the
\(10^{-4}\), \(10^{-6}\), and \(10^{-8}\) thresholds
(Figure~\ref{fig-real-multinomial}). Their wall-clock trajectories
differ because the intercept updates have different costs. The roughly
fourfold slowdown at a relative suboptimality bound of \(10^{-4}\) is
milder than the order-of-magnitude slowdown of
Figure~\ref{fig-multinomial-imbalance}. The realized class prevalences,
feature effects, and penalty contribution differ between these problems,
so the baseline probabilities alone do not explain the difference.

The same separation appears when the dimensions are reversed. The
StatLog Shuttle dataset~(Michie et al. 1994) is a \(K = 7\) sensor-data
classification with \(n = 58\,000\) observations and \(p = 9\)
features---the opposite shape to Yeoh2002---and class frequencies that
span four orders of magnitude. The largest class contains nearly
\(80\%\) of the observations, while the three rarest have only ten,
thirteen, and fifty samples. Here the reference class is itself among
the rarest, so the Rayleigh quotients fall much further below \(L_0\)
than in Yeoh2002
(Equation~\ref{eq-multinomial-free-rayleigh}--Equation~\ref{eq-multinomial-reference-rayleigh}).

\begin{figure}

\centering{

\pandocbounded{\includegraphics[keepaspectratio]{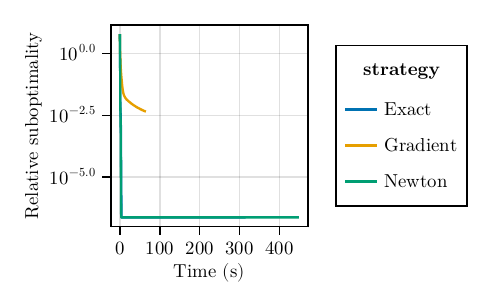}}

}

\caption{\label{fig-real-multinomial-shuttle}Multinomial logistic
regression on StatLog Shuttle. The plot shows a relative suboptimality
bound versus time at \(\lambda = 0.05\lambda_{\max}\) under cyclic CD
for \(n = 58\,000\), \(p = 9\), and \(K = 7\). Suboptimality is bounded
using the strongest shared feasible dual point.}

\end{figure}%

On Shuttle, the Newton and exact strategies bring the relative
suboptimality bound below \(10^{-4}\) in \(20\)--\(21\) passes, while
the gradient strategy remains well above that threshold after the full
thousand-pass budget. The slowdown is far more severe than the roughly
fourfold gap on Yeoh because the rarest Shuttle classes are two orders
of magnitude rarer. This is consistent with the smaller eigenmode
correction factors in Equation~\ref{eq-grad-residual-matrix}; it is not
a claim that a single Hessian diagonal controls the block solve. The
rare-class slowdown therefore appears both when \(p \gg n\) in the
gene-expression data and when \(n \gg p\) in the sensor data
(Figure~\ref{fig-real-multinomial-shuttle}).

The class-by-class scheme of \textbf{glmnet}'s symmetric-multinomial
path leads to a parallel conclusion. Holding all \(\eta_{il}\) for
\(l \neq k\) fixed, the marginal subproblem in class \(k\)'s linear
predictors \(\eta_{ik}\) is, up to a per-observation offset
\(\log\bigl(1 + \sum_{l \neq k} e^{\eta_{il}}\bigr)\), a scalar logistic
regression. The scalar Schur-coupling residual
(Equation~\ref{eq-grad-residual}) therefore applies \emph{per class},
and any rare class with extreme per-class response imbalance triggers
exactly the gradient-strategy failure documented in the binary setting,
now along the rare class's intercept coordinate. Multinomial models
therefore expose two independent failure mechanisms: scalar imbalance
under class-by-class CD and rare-class eigenmode attenuation under block
CD. We recommend the Newton strategy for both schemes.

\subsubsection*{S2.2 Full imbalance--amplitude
sweep}\label{s2.2-full-imbalanceamplitude-sweep}
\addcontentsline{toc}{subsubsection}{S2.2 Full imbalance--amplitude
sweep}

The four-corner view in Figure~\ref{fig-multinomial-sweep} shows the
endpoints of the imbalance--amplitude experiment but hides the
transition between them. Here we report all sixteen cells. We use the
same \(K = 5\) multinomial logistic design with \(n = 500\),
\(p = 200\), \(s = 5\), feature correlation \(\rho = 0.3\), and
\(\lambda = 0.05\lambda_{\max}\). The reference-class baseline
probability \(\pi_K\) ranges from \(0.05\) to \(0.005\), and the feature
amplitude ranges from \(0.5\) to \(2.0\). Each curve is the median
relative suboptimality bound over three seeds. This varies baseline log
odds and signal amplitude; the realized prevalences also change with
amplitude and seed.

\begin{figure}

\centering{

\pandocbounded{\includegraphics[keepaspectratio]{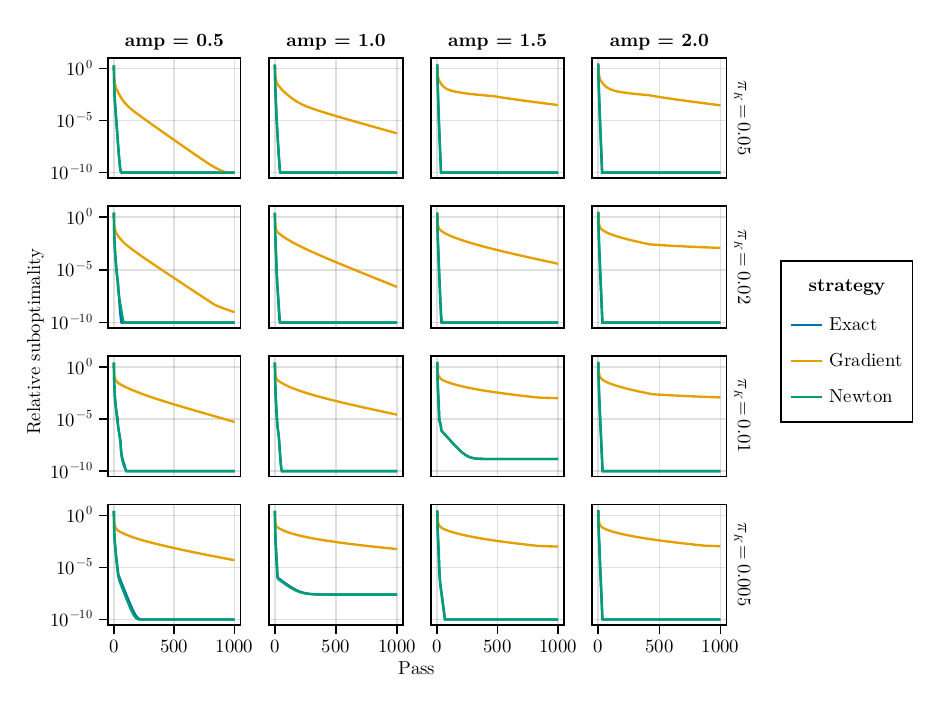}}

}

\caption{\label{fig-multinomial-sweep-full}Full multinomial
imbalance-amplitude sweep. The panels show a relative suboptimality
bound versus outer-pass index over a \(4 \times 4\) grid of
reference-class baseline probability
\(\pi_K \in \{0.05, 0.02, 0.01, 0.005\}\) and feature amplitude
\(\in \{0.5, 1.0, 1.5, 2.0\}\). Curves are medians over three seeds for
\(K = 5\), \(n = 500\), \(p = 200\), \(s = 5\), \(\rho = 0.3\), and
\(\lambda = 0.05\lambda_{\max}\); early-terminated trajectories are
carried forward at their terminal values.}

\end{figure}%

Across the grid, the gradient strategy generally stalls more severely as
\(\pi_K\) decreases and feature amplitude grows; Newton and exact remain
close. Both reach a median bound below \(10^{-6}\) in every cell. At the
replicate level, Newton and exact reach that threshold in \(41\) of the
\(48\) seeded problems, while gradient reaches it in \(14\). The seven
remaining problems all lack reference-class observations and have weak
shared certificates. At \(\pi_K = 0.005\), the reference-class counts
for seeds 1--3 are \((0,7,1)\) at amplitude \(1.5\) and \((0,8,1)\) at
amplitude \(2.0\). Thus, the median curves in these cells conceal an
uncertified replicate.

When the reference class is absent, increasing all free-class intercepts
together strictly decreases the loss at every finite iterate, without
changing the penalty. There is therefore no finite intercept minimizer
for those replicates, and the positive-curvature local analysis at a
finite optimum does not apply. Their weak certificates cannot identify
coefficient convergence or near-separation as the cause.

The class-count cache also flags absent-class replicates elsewhere in
the sweep. An objective certificate can still bound distance to the
infimum in these boundary cases, but it does not imply convergence to
finite intercepts. We retain their trajectories to show this limitation.
The curvature analysis at a finite optimum applies to replicates with
every class observed. The solvers use the same seeded data within each
problem, although the grid does not hold marginal prevalence fixed as
amplitude varies. Where the gradient strategy reaches the \(10^{-6}\)
threshold, it generally needs at least an order of magnitude more passes
than Newton. In most of the remaining certified cells, it misses that
threshold within the \(1\,000\)-pass budget even though Newton reaches
it.

\subsection*{S3. Experimental methods and
reproducibility}\label{sec-s-methodology}
\addcontentsline{toc}{subsection}{S3. Experimental methods and
reproducibility}

We run all internal-solver experiments with the \textbf{Intercepts}
Julia package on Julia 1.11 and the default \textbf{OpenBLAS}. Our
primary measure of work is the number of outer passes. One pass consists
of a sweep through all \(p\) coordinates and an intercept update. Unless
noted, ``passes to tolerance'' means the number of completed passes
needed to bring the shared relative suboptimality bound below
\(10^{-8}\). Pass budgets (\texttt{maxit}) are 1000 for the
single-\(\lambda\) figures and 5000 for the heatmap of
Section~\ref{sec-imbalance-reg}; cells that miss tolerance are shown at
the cap. We standardize the synthetic designs, use cyclic coordinate
ordering unless noted, and set the default AR(1) feature correlation to
\(\rho = 0.6\). We fix the seeds per script. We average three seeds per
cell in the heatmap of Section~\ref{sec-imbalance-reg}, show three seeds
in Figure~\ref{fig-irls-comparison}, and show five in the safeguard
experiments. We use one fixed seed for the remaining single-trajectory
figures.

Every coefficient update is accepted through a Tseng--Yun Armijo
backtrack on the proximal model decrease, with slope constant
\(10^{-4}\), step-halving, and at most 30 trials; in the common case the
first trial is accepted. The \texttt{cdsolver} defaults match the
paper's defaults (cyclic ordering, Newton intercept strategy), and every
experiment script passes its configuration explicitly.

For scalar models, the exact strategy ordinarily returns when the
normalized intercept gradient is at most \(10^{-10}\). If it consumes
all 50 safeguarded Newton steps, a final check accepts a normalized
gradient of at most \(5 \times 10^{-10}\); larger residuals raise an
error. This cap-only margin prevents saturated floating-point predictors
from stalling at the nominal boundary and remains twenty times smaller
than the \(10^{-8}\) reporting threshold. The multinomial implementation
uses a strict \(10^{-8}\) normalized block-gradient threshold and raises
an error whenever it misses that target after 50 steps. Tests exercise
both acceptance and rejection at the scalar cap. Every reported result
for the exact strategy met the corresponding conditional-solve
criterion.

At each diagnostic point, the raw residual lies in the loss-conjugate
domain but need not satisfy dual stationarity for the unpenalized
intercept. Simply subtracting its mean can leave that domain. We
therefore blend the centered residual with the intercept-only feasible
residual, taking the largest blend whose implied binomial probability,
Poisson intensity, or multinomial probability vector remains in its
conjugate domain. Both endpoints satisfy intercept stationarity. We then
rescale the result toward zero to enforce
\(\lVert X^\top \Theta \rVert_\infty \le \lambda\); this second step
preserves the domain because it moves the implied mean toward the
observed response. This feature-feasibility rescaling is the standard
dual-scaling step used by Gap Safe rules~(Fercoq et al. 2015; Ndiaye et
al. 2017); the preceding anchor blend adds the conjugate-domain and
unpenalized-intercept constraints needed here. The solver evaluates the
conjugate only after checking domain membership and stops when its
run-local relative duality gap reaches the stated tolerance.

For reporting, we put all strategies for the same problem instance on
one reference. If \(D_{s,t}\) is the feasible dual value from strategy
\(s\) at diagnostic \(t\), we set
\begin{equation}\protect\phantomsection\label{eq-shared-suboptimality}{
  \underline D = \max_{s,t} D_{s,t}, \qquad \overline P = \min_{s,t} P_{s,t},
  \qquad
  B_{s,t} = \frac{P_{s,t} - \underline D}{|\overline P|}.
}\end{equation} We use a floor of \(10^{-15}\) for the denominator in
code. Weak duality gives
\(P_{s,t} - P^\star \le P_{s,t} - \underline D\), so \(B_{s,t}\) is a
relative suboptimality upper bound with a strategy-independent scale and
reference. The certificate width
\((\overline P - \underline D)/|\overline P|\) is below \(10^{-8}\) for
every scalar trajectory figure in the paper. This is a retrospective
comparison: a displayed curve may cross the reporting threshold before
that run's own duality gap triggers termination. In the two longest
rate-diagnostic sweeps, we compute the tight Newton reference first and
stop a strategy as soon as its primal reaches the corresponding
shared-bound threshold. This changes neither its iterates nor its
reported crossing; it only avoids running beyond the point used in the
figure.

For the real-data experiments, we use pinned datasets from Zenodo~(DOI:
\href{https://doi.org/10.5281/zenodo.20315625}{10.5281/zenodo.20315625})
and standardize them before fitting. We include in the repository the
code needed to retrieve these data and generate the shared standardized
problem used in the additional production-solver diagnostics. Those
cached diagnostics use \textbf{biglasso} 1.6-1 (R), \textbf{adelie}
1.0-10 (R), \textbf{skglm} 0.5 (Python), \textbf{LIBLINEAR} through
\textbf{scikit-learn}'s \texttt{solver="liblinear"} wrapper, and
\textbf{BlitzL1} at commit \texttt{aef8d02}. They record one fit per
tolerance setting. The separate cold-start diagnostic in
\hyperref[sec-cold-start-diag]{Supplement S5.2} uses \textbf{adelie}
1.0-8. The controlled imbalance sweep uses \textbf{glmnet} 5.0 and
\textbf{biglasso} 1.6-1, with seven warmed timing repetitions per mode
and cell. It uses one BLAS/OpenMP thread and the same explicit
regularization path for both modes. The benchmark stores every
calibration result and timing, the reference certificates and curvature
diagnostics, and the package versions in
\texttt{results/production-imbalance/}. Its reference fits use
guarded-Newton CD followed by a joint Newton refinement on the active
coefficients and intercept. We evaluate their primal and feasible dual
objectives in 128-bit arithmetic to avoid loss of certificate accuracy
under saturated predictions. Run \texttt{task\ production-imbalance} to
regenerate these artifacts. We document the other commands, tolerance
grids, and iterate caps in
\href{https://github.com/jolars/intercepts/blob/main/REPRODUCIBILITY.md}{the
reproduction guide}.

We store all experimental outputs in the repository. We therefore render
most figures from cached files, and regenerating a panel requires
rerunning the corresponding experiment. The illustrative
Figure~\ref{fig-parametric} and Figure~\ref{fig-schur-example} are
computed at render time. We pin the Julia, R, and Python toolchains in
the repository. Except for the controlled production imbalance sweep,
they use their default \textbf{OpenBLAS} builds without thread limits.
The earlier experiments ran on an x86\_64 Linux workstation with an AMD
Ryzen 9 7900 (12 cores, 24 threads) and 64 GB of RAM. The production
imbalance sweep and the multinomial experiments ran on an Intel Core
Ultra 7 155U laptop with 32 GB of RAM. The wall-clock measurements
describe these machines and software environments.

\subsection*{S4. Robustness analyses}\label{sec-s-robustness}
\addcontentsline{toc}{subsection}{S4. Robustness analyses}

\subsubsection*{S4.1 Cyclic versus permuted
ordering}\label{sec-ordering}
\addcontentsline{toc}{subsubsection}{S4.1 Cyclic versus permuted
ordering}

The within-pass drift in Equation~\ref{eq-intercept-drift} does not
depend on coordinate ordering, but strict cyclic order deterministically
pairs the intercept solve with the same first coordinate on every pass,
which could compound the drift.

On w1a, we find no such compounding under the per-coordinate Armijo
backtracking used by our coordinate-descent solver, \texttt{cdsolver},
after every coordinate update (\hyperref[sec-s-methodology]{Supplement
S3}). All three strategies converge under both orderings, and the
pass-count differences are less than a factor of two
(Figure~\ref{fig-first-example-ordering}). Newton and exact remain below
100 passes under either ordering, whereas the gradient strategy takes
about 340 under both. The qualitative ranking---Newton and exact fast,
gradient slow---is unchanged, and we use cyclic CD as the default
elsewhere in the paper.

\begin{figure}

\centering{

\pandocbounded{\includegraphics[keepaspectratio]{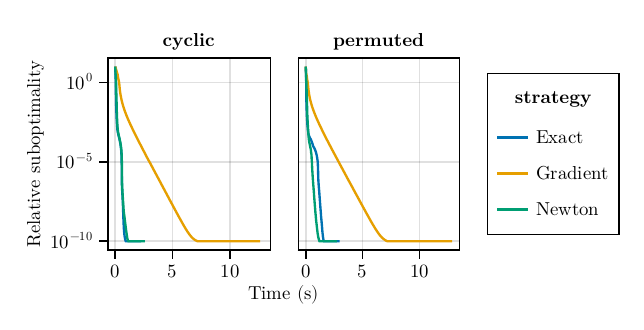}}

}

\caption{\label{fig-first-example-ordering}Ordering comparison on w1a.
The panels show a relative suboptimality bound versus time under cyclic
(left) and permuted (right) coordinate ordering for the same
configuration as Figure~\ref{fig-first-example}.}

\end{figure}%

\subsubsection*{S4.2 Warm-started path}\label{sec-warmstart-path}
\addcontentsline{toc}{subsubsection}{S4.2 Warm-started path}

In the single-\(\lambda\) experiments above, we use the most difficult
regime: a cold start with a large \(|\partial_0 F|\) at the zero
iterate. The drift in Equation~\ref{eq-intercept-drift} and the small
correction in Equation~\ref{eq-grad-residual} have their largest effects
there.

In practice, users rarely solve for a single \(\lambda\) in isolation.
Instead, they sweep a grid \(\lambda_1 > \lambda_2 > \dots > \lambda_K\)
and warm-start each subproblem from the previous solution. As the
argument above suggests, warm-starting begins each subproblem with
\(|\partial_0 F|\) already small. A single Newton step is then
essentially exact at the start, so \(k_0 \in \{1, 2\}\) in the exact
strategy and the \(k_0 - 1\) per-pass overhead evaporates. The
gradient-strategy slowdown should also lessen because the iterate is
close to optimum and the \(H_{00}/L_0\) shrinkage acts on an
already-small residual (Equation~\ref{eq-newton-coupling}).

We test this by solving an \(\ell_1\)-logistic path on two problems---a
simulated standardized design (\(n = 500\), \(p = 1000\), \(s = 10\),
\(\mu_0 = 0.9\), \(\rho = 0.6\)) and the w1a benchmark used in
Figure~\ref{fig-first-example}---over a geometrically spaced grid of
\(K = 20\) values from \(\lambda_{\max}\) down to
\(5 \cdot 10^{-3}\lambda_{\max}\).

Each subproblem runs until its local relative duality gap reaches
\(10^{-8}\) or to a budget of 300 outer passes / 20 wall-seconds,
whichever comes first. We then report the first pass at which the shared
relative suboptimality bound reaches \(10^{-8}\). We compare the three
strategies under both warm start (initialized from the
previous-\(\lambda\) solution) and cold start (re-initialized at zero),
and under both cyclic and permuted feature ordering. The panels here
show cyclic CD; Figure~\ref{fig-warm-start-full} and
Figure~\ref{fig-warm-start-mechanism-full} give the full grid over both
orderings.

\begin{figure}

\centering{

\pandocbounded{\includegraphics[keepaspectratio]{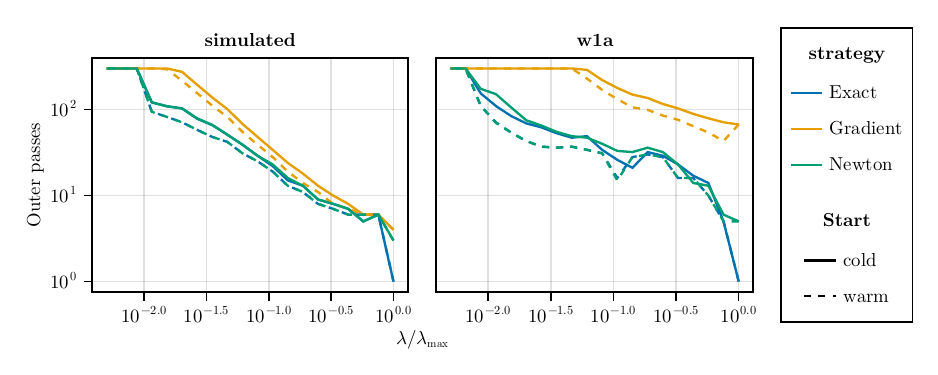}}

}

\caption{\label{fig-warm-start}Warm-started and cold-started
regularization paths. The panels show outer passes per subproblem along
a logistic-regression \(\lambda\) path under cyclic CD; dashed lines use
warm starts, solid lines cold starts, and curves capped at 300 hit the
pass budget. The left panel shows the simulated imbalanced problem
(\(\mu_0 = 0.9\)) and the right panel w1a.}

\end{figure}%

\begin{figure}

\centering{

\pandocbounded{\includegraphics[keepaspectratio]{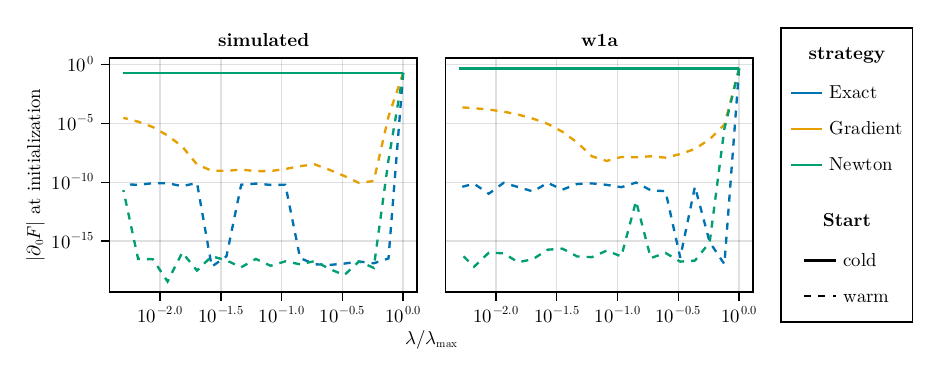}}

}

\caption{\label{fig-warm-start-mechanism}Initial intercept gradients
along the regularization path. The panels show \(|\partial_0 F|\) at the
initial iterate of each subproblem for the same path and cyclic-CD setup
as Figure~\ref{fig-warm-start}; dashed lines use warm starts and solid
lines cold starts.}

\end{figure}%

In Figure~\ref{fig-warm-start-mechanism}, warm-starting drives
\(|\partial_0 F|\) toward zero as \(\lambda\) shrinks and the
previous-\(\lambda\) solution moves closer to the current optimum. Cold
start, by contrast, sets the initial iterate to zero at every
\(\lambda\), so \(|\partial_0 F|\) is the marginal class
imbalance---roughly \(0.20\) on the simulated problem and \(0.47\) on
w1a, identical across strategies.

The pass-count reduction is modest and uneven
(Figure~\ref{fig-warm-start}). Averaged over the simulated path,
warm-starting saves about \(10\%\) of the passes for all three
strategies. On w1a, it again helps the gradient strategy only modestly,
while Newton and exact often improve by roughly \(20\%\)--\(40\%\). The
number of budget-capped subproblems does not fall uniformly, but Newton
and exact remain close under every configuration. This is what
Equation~\ref{eq-intercept-drift} suggests: once \(|\partial_0 F|\) is
already small, the \(k_0 - 1\) per-pass penalty of the exact strategy
largely disappears.

Warm-starting does not rescue the gradient strategy. At the smallest
\(\lambda\) under cyclic ordering, its relative bound improves by
roughly two orders of magnitude on the simulated problem and by more
than one order on w1a. It still remains well above the corresponding
Newton and exact bounds. Thus, a better initial point helps without
compensating for the small \(H_{00}/L_0\) correction.

These results leave the practical recommendation unchanged: use a
safeguarded Newton step in direct CD and a fully resolved intercept on
the frozen quadratic at a prox-Newton/IRLS boundary. The exact
original-loss strategy is also defensible in a warm-started path solve,
where its extra work is small.

\subsubsection*{S4.3 Per-pass cost
decomposition}\label{sec-per-pass-cost}
\addcontentsline{toc}{subsubsection}{S4.3 Per-pass cost decomposition}

The pass-count and wall-clock figures above combine two quantities that
the within-pass drift (Equation~\ref{eq-intercept-drift}) separates: the
\emph{number of outer passes} needed to reach a tolerance and the
\emph{per-pass inner cost} of the intercept update. We test whether the
exact strategy pays extra inner work per outer pass without a
commensurate rate improvement. We count its Newton directions, but floor
the count at one because even a zero-step solve must evaluate the
conditional residual. Local quadratic convergence explains why few
directions remain once an iterate is near the conditional minimizer, but
it does not supply a global bound on the work.

We instrument the solver to record \(k_0\) for every outer pass. We
compare the resulting cost trajectories with the relative-suboptimality
trajectories on the imbalanced simulated design of
Figure~\ref{fig-mu-extreme} (\(n = 500\), \(p = 1000\), \(s = 10\),
cyclic CD, \(\lambda / \lambda_{\max} = 0.05\)), drawn with an
independent seed. The solver runs to a local relative-duality-gap
tolerance of \(10^{-10}\), and we report progress to the shared
\(10^{-8}\) bound (Equation~\ref{eq-shared-suboptimality}), sweeping
\(\mu_0 \in \{0.5,0.9,0.99\}\) to vary the cold-start magnitude of
\(|\partial_0 F|\).

\begin{figure}

\centering{

\pandocbounded{\includegraphics[keepaspectratio]{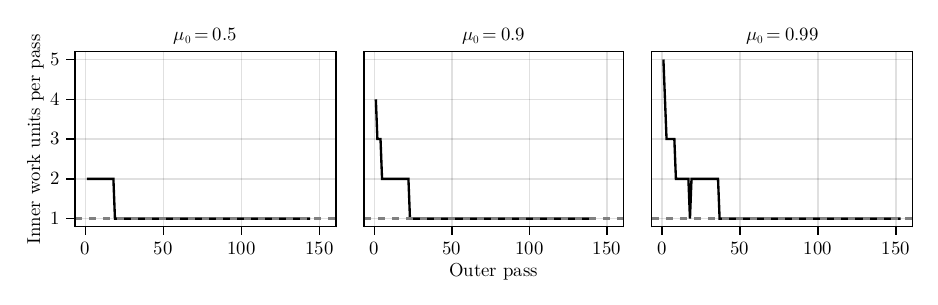}}

}

\caption{\label{fig-per-pass-cost}Per-pass inner work under the exact
strategy. The panels show the Newton-direction count, floored at one for
the mandatory residual check, versus outer-pass index for
\(\mu_0 \in \{0.5, 0.9, 0.99\}\) on an independent draw of the simulated
design of Figure~\ref{fig-mu-extreme}. The dashed line at \(k_0 = 1\)
marks the Newton strategy.}

\end{figure}%

\begin{figure}

\centering{

\pandocbounded{\includegraphics[keepaspectratio]{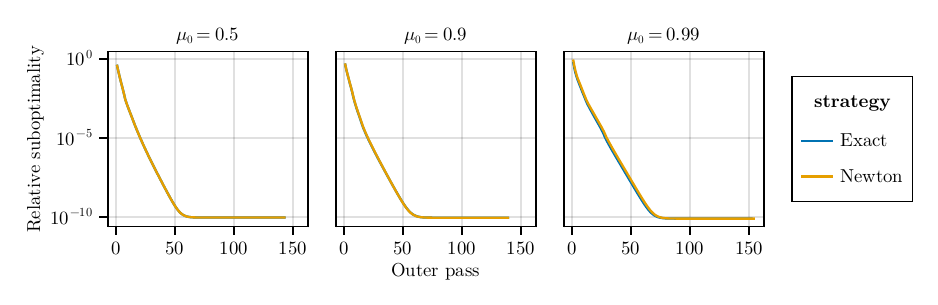}}

}

\caption{\label{fig-per-pass-progress}Per-pass progress on the same
runs. The panels show a relative suboptimality bound after each outer
pass for the Newton and exact strategies on the experiments of
Figure~\ref{fig-per-pass-cost}.}

\end{figure}%

The two figures show the overhead directly. In
Figure~\ref{fig-per-pass-cost}, the inner work peaks between two and
five units on the cold-start pass, with larger peaks under greater
imbalance, then decays to one as the iterate approaches the joint
optimum. Summed through the shared \(10^{-8}\) bound, exact uses roughly
one-and-a-half to twice the work of Newton, again increasing with
imbalance.

The extra work is hard to justify: the Newton and exact strategies trace
almost the same relative-suboptimality curve and reach \(10^{-8}\) in
comparable pass counts that differ by at most one in these runs. The
exact strategy spends more inner work as imbalance grows without
improving outer-pass progress (Figure~\ref{fig-per-pass-progress}).

The penalty is modest in these runs: \(k_0\) remains well below the
inner-loop cap of 50, and most passes need at most one Newton direction.
Every returned update satisfies the nominal \(10^{-10}\)
conditional-gradient tolerance because no update reaches the cap; the
largest observed \(k_0\) is five. A cap hit would invoke the final
numerical check documented in \hyperref[sec-s-methodology]{Supplement
S3}, not automatically produce a cached result. This behavior is
consistent with local quadratic convergence after the iterates approach
their conditional minimizers, but the experiment, rather than a global
Newton bound, establishes the reported cost. The exact strategy incurs
most of its overhead in the first few cold-start passes, when
\(|\partial_0 F|\) is largest. The warm-start regime of
\hyperref[sec-warmstart-path]{Supplement S4.2} largely avoids that
window.

\subsubsection*{S4.4 Full warm-start
grids}\label{s4.4-full-warm-start-grids}
\addcontentsline{toc}{subsubsection}{S4.4 Full warm-start grids}

The cyclic-CD panels in \hyperref[sec-warmstart-path]{Supplement S4.2}
show the ordering used throughout most of the article. Here we restore
the permuted-CD runs to display the complete experiment. We solve a
\(K = 20\) regularization path from \(\lambda_{\max}\) to
\(5 \cdot 10^{-3}\lambda_{\max}\) on two problems: the simulated design
with \(n = 500\), \(p = 1000\), \(s = 10\), \(\mu_0 = 0.9\), and
\(\rho = 0.6\), and the w1a benchmark. Each subproblem starts either
from zero or from the preceding solution and runs until its local
relative duality gap reaches \(10^{-8}\), a budget of 300 outer passes,
or a 20-second time limit. We report passes to the shared \(10^{-8}\)
suboptimality bound.

\begin{figure}

\centering{

\pandocbounded{\includegraphics[keepaspectratio]{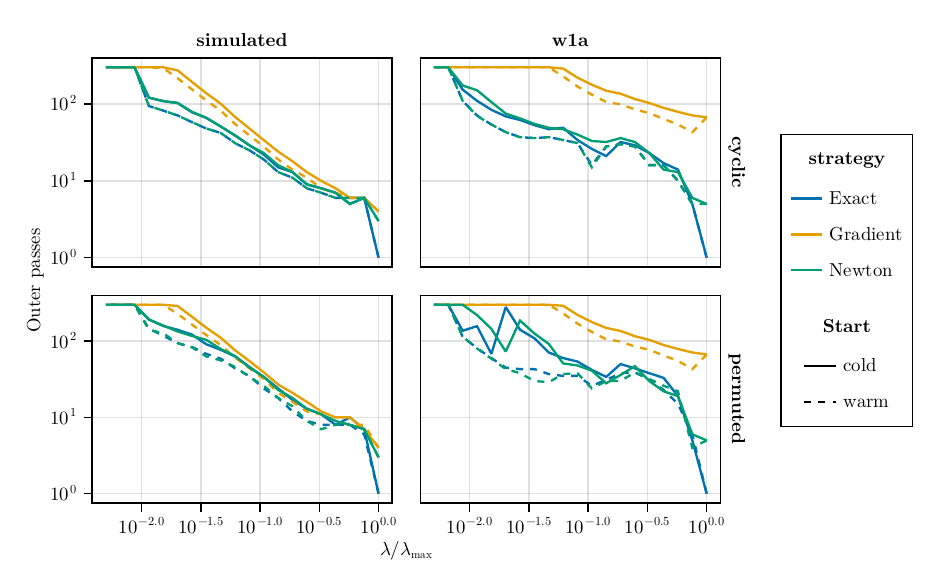}}

}

\caption{\label{fig-warm-start-full}Full warm-started and cold-started
regularization paths. The panels show outer passes per subproblem along
the logistic-regression \(\lambda\) path on the simulated imbalanced
design and w1a for cyclic CD (top) and permuted CD (bottom). Dashed
lines use warm starts, solid lines use cold starts, and curves capped at
300 hit the pass budget.}

\end{figure}%

The full pass-count grid confirms that feature ordering does not change
the comparison among intercept strategies. Warm starts reduce average
pass counts under every configuration, but they remove a budget cap in
only two cases: the cyclic gradient run on the simulated path and the
permuted Newton run on w1a. Newton and exact remain close under both
orderings. The gradient strategy still slows most sharply at small
\(\lambda\), especially on w1a, where its conservative intercept
correction controls the rate.

To connect these pass counts to the proposed mechanism,
Figure~\ref{fig-warm-start-mechanism-full} reports the intercept
gradient before the first pass of every subproblem. This quantity
measures how much intercept error the new subproblem inherits from its
initialization.

\begin{figure}

\centering{

\pandocbounded{\includegraphics[keepaspectratio]{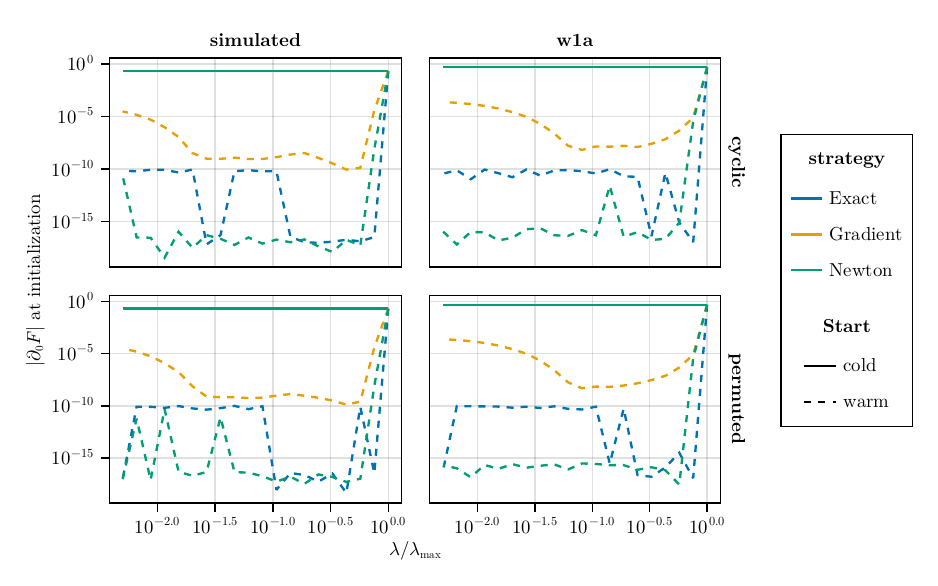}}

}

\caption{\label{fig-warm-start-mechanism-full}Full initial-gradient
regularization paths. The panels show \(|\partial_0 F|\) at the initial
iterate of each subproblem along the \(K = 20\) \(\lambda\) path on the
simulated imbalanced design and w1a for cyclic CD (top) and permuted CD
(bottom). Dashed lines use warm starts, and solid lines use cold
starts.}

\end{figure}%

Under both orderings, warm starting drives \(|\partial_0 F|\) toward
zero as the path proceeds, whereas cold starting resets it to the
marginal class imbalance at every value of \(\lambda\). The repeated
pattern links the Newton--exact agreement in
Figure~\ref{fig-warm-start-full} to the initialization rather than to
cyclic ordering: when the inherited intercept gradient is already small,
resolving the conditional intercept problem beyond one guarded Newton
step adds little.

\subsection*{S5. Production-solver diagnostics}\label{sec-s-production}
\addcontentsline{toc}{subsection}{S5. Production-solver diagnostics}

\subsubsection*{S5.1 Poisson family}\label{sec-production-poisson}
\addcontentsline{toc}{subsubsection}{S5.1 Poisson family}

Poisson loss rules out several production solvers before any trajectory
is recorded. Because \(f''(\eta) = e^\eta\) is unbounded on
\(\mathbb{R}\), no global Lipschitz upper bound exists. Consequently,
neither the \(w \equiv 1/4\) MM majorant of \textbf{glmnet}
\texttt{modified.Newton} and \textbf{biglasso} \texttt{MM} nor the
constant-Lipschitz coordinate step of \textbf{skglm} \texttt{AndersonCD}
has a Poisson analog. The status of each solver on this family is
structural, not empirical:

\begin{itemize}
\tightlist
\item
  \textbf{glmnet} supports \texttt{family\ =\ "poisson"} with a single
  fitting mode (local-IRLS); no \texttt{type.poisson} switch exists.
\item
  \textbf{biglasso} does not implement \texttt{family\ =\ "poisson"} at
  all; its signature is restricted to \texttt{gaussian},
  \texttt{binomial}, \texttt{cox}, and \texttt{mgaussian}.
\item
  \textbf{skglm} \texttt{AndersonCD} requires
  \texttt{datafit.get\_lipschitz()}, which the \texttt{Poisson} datafit
  declines to implement
  (\texttt{AttributeError:\ Poisson\ is\ not\ \ \ compatible\ with\ solver\ AndersonCD}).
  The \texttt{AndersonCD} solver therefore cannot fit a Poisson model.
\item
  \textbf{skglm} \texttt{ProxNewton} pairs with the \texttt{Poisson}
  datafit through local prox-Newton subproblems and represents
  \textbf{skglm} in the panel.
\item
  \textbf{adelie} supports \texttt{glm.poisson(y)} through the same
  direct-Newton / weighted-centered-features parameterization as the
  logistic case.
\end{itemize}

Three solvers based on local Newton curvature therefore remain:
\textbf{glmnet}, \textbf{adelie}, and \textbf{skglm}
\texttt{ProxNewton}. We run each in path mode on congress109 with
\(\lambda = 0.05\lambda_{\max}\), sweeping the solver-specific tolerance
(Figure~\ref{fig-real-poisson-production}).

\begin{figure}

\centering{

\pandocbounded{\includegraphics[keepaspectratio]{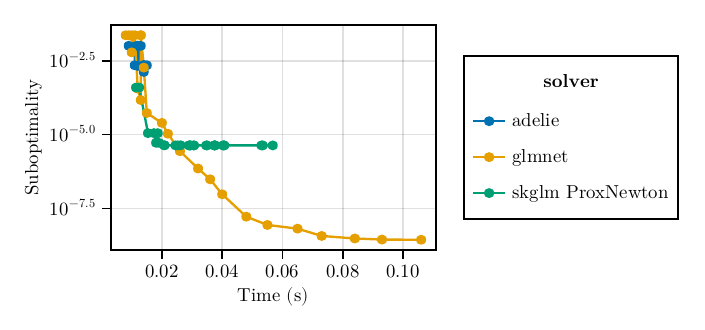}}

}

\caption{\label{fig-real-poisson-production}Production Poisson solvers
on congress109. The plot shows a suboptimality bound versus wall-clock
time for \textbf{glmnet}, \textbf{adelie}, and \textbf{skglm}
\texttt{ProxNewton} in path mode at \(n = 529\), \(p = 999\),
\(\bar y = 11.83\), and \(\lambda = 0.05\lambda_{\max}\). Every solver
uses the common certified feasible dual reference \(D^\star\).}

\end{figure}%

The within-package weight-scheme switches in
Figure~\ref{fig-production-imbalance-time} have no Poisson analog. The
loss itself does not admit the \(w \equiv 1/4\) surrogate used by the
gradient-class solvers. Every working production Poisson solver in our
survey therefore uses local Newton curvature, as the analysis
recommends.

\subsubsection*{S5.2 Single-lambda cold
start}\label{sec-cold-start-diag}
\addcontentsline{toc}{subsubsection}{S5.2 Single-lambda cold start}

The path-mode protocol used above lets inner CD avoid the transient at
\(\beta = 0\). When we instead solve a single \(\lambda\) from scratch,
an unguarded Newton intercept update can become unstable
(Equation~\ref{eq-newton-coupling}). We examine this possibility in
Table~\ref{tbl-cold-start-diag} at \(\lambda = 0.05\lambda_{\max}\) with
loose inner and outer solver tolerances of \(10^{-1}\), sweeping the
iterate cap from \(10^1\) to \(10^5\); the path-mode reference primal is
\(0.0634\) on w1a and \(0.0497\) on news20-3pct. These settings form a
targeted diagnostic rather than a test of either package's default
configuration, whose convergence tolerances are stricter.

\begin{longtable}[]{@{}
  >{\raggedright\arraybackslash}p{(\linewidth - 8\tabcolsep) * \real{0.1068}}
  >{\raggedright\arraybackslash}p{(\linewidth - 8\tabcolsep) * \real{0.3301}}
  >{\raggedright\arraybackslash}p{(\linewidth - 8\tabcolsep) * \real{0.2816}}
  >{\raggedleft\arraybackslash}p{(\linewidth - 8\tabcolsep) * \real{0.0680}}
  >{\raggedright\arraybackslash}p{(\linewidth - 8\tabcolsep) * \real{0.2136}}@{}}
\caption{Single-\(\lambda\) cold-start diagnostic on w1a and
news20-3pct. Each row reports the intercept estimate, primal objective,
and iteration count at \(\lambda = 0.05\lambda_{\max}\) with inner and
outer solver tolerances of \(10^{-1}\); the iterate cap is swept from
\(10^1\) to \(10^5\), and the path-mode reference primal is \(0.0634\)
on w1a and \(0.0497\) on news20-3pct. The \textbf{adelie} entries use
version 1.0-8. }\label{tbl-cold-start-diag}\tabularnewline
\toprule\noalign{}
\begin{minipage}[b]{\linewidth}\raggedright
dataset
\end{minipage} & \begin{minipage}[b]{\linewidth}\raggedright
solver
\end{minipage} & \begin{minipage}[b]{\linewidth}\raggedright
\(\beta_0\)
\end{minipage} & \begin{minipage}[b]{\linewidth}\raggedleft
primal
\end{minipage} & \begin{minipage}[b]{\linewidth}\raggedright
iterations
\end{minipage} \\
\midrule\noalign{}
\endfirsthead
\toprule\noalign{}
\begin{minipage}[b]{\linewidth}\raggedright
dataset
\end{minipage} & \begin{minipage}[b]{\linewidth}\raggedright
solver
\end{minipage} & \begin{minipage}[b]{\linewidth}\raggedright
\(\beta_0\)
\end{minipage} & \begin{minipage}[b]{\linewidth}\raggedleft
primal
\end{minipage} & \begin{minipage}[b]{\linewidth}\raggedright
iterations
\end{minipage} \\
\midrule\noalign{}
\endhead
\bottomrule\noalign{}
\endlastfoot
w1a & \textbf{biglasso} \texttt{alg.logistic\ =\ "Newton"} & \(-723\)
(range \([-1722,-291]\)) & \(16.19\) & \(10^5\) cap, oscillates \\
w1a & \textbf{biglasso} \texttt{alg.logistic\ =\ "MM"} & \(-3.51\) &
\(0.088\) & \(2\) \\
w1a & \textbf{adelie} & --- & --- & empty state returned \\
news20-3pct & \textbf{biglasso} \texttt{alg.logistic\ =\ "Newton"} &
\(-3.48\) & \(0.310\) & \(2\) (premature stop) \\
news20-3pct & \textbf{biglasso} \texttt{alg.logistic\ =\ "MM"} &
\(-3.48\) & \(0.089\) & \(2\) \\
news20-3pct & \textbf{adelie} & \(-5.22\) & \(0.052\) & converges \\
\end{longtable}

In Table~\ref{tbl-cold-start-diag}, the global majorant \(L_0 = 1/4\)
acts as an \(X\)-free safeguard on the intercept: two scalars (\(n\) and
\(\bar y\)) determine the step, and that is enough to absorb the
cold-start transient at trivial cost on the unpenalized coordinate. The
analogous safeguard for a coefficient would require column-norm
information.

Under this specific configuration, \textbf{biglasso}'s default Newton
mode returns inaccurate fits on both datasets: it reaches the iteration
cap while oscillating on w1a and stops prematurely on news20-3pct.
Adelie, whose only mode uses local Newton curvature, returns an empty
state on w1a but converges on news20-3pct. The contrast between the
datasets, despite their matched \(\bar y\), is consistent with the
failure depending on more than the cold-start gradient magnitude. A
plausible explanation is that, on news20-3pct's wider sparse design, the
\(\ell_1\) penalty pulls many entries to zero before the linear
predictor reaches extreme values, damping the IRLS weight collapse. The
diagnostic does not establish behavior at other values of \(\lambda\),
under other stopping tolerances, or for the packages' otherwise default
calls.

\section*{References}\label{references}
\addcontentsline{toc}{section}{References}

\protect\phantomsection\label{refs}


\begin{CSLReferences}{1}{1}
\bibitem[\citeproctext]{ref-bertrand2022}
Bertrand, Quentin, Quentin Klopfenstein, Pierre-Antoine Bannier,
Gauthier Gidel, and Mathurin Massias. 2022. {``Beyond {L1}: Faster and
Better Sparse Models with {skglm}.''} \emph{Advances in {Neural
Information Processing Systems}} (New Orleans, USA) 35
(November--December): 38950--65.
\url{https://doi.org/10.52202/068431-2823}.

\bibitem[\citeproctext]{ref-bohning1992}
Böhning, Dankmar. 1992. {``Multinomial Logistic Regression Algorithm.''}
\emph{Annals of the Institute of Statistical Mathematics} 44 (1):
197--200. \url{https://doi.org/10.1007/BF00048682}.

\bibitem[\citeproctext]{ref-breheny2011}
Breheny, Patrick, and Jian Huang. 2011. {``Coordinate Descent Algorithms
for Nonconvex Penalized Regression, with Applications to Biological
Feature Selection.''} \emph{The Annals of Applied Statistics} 5 (1):
232--53. \url{https://doi.org/10.1214/10-AOAS388}.

\bibitem[\citeproctext]{ref-fan2008}
Fan, Rong-En, Kai-Wei Chang, Cho-Jui Hsieh, Xiang-Rui Wang, and Chih-Jen
Lin. 2008. {``{LIBLINEAR}: A Library for Large Linear Classification.''}
\emph{Journal of Machine Learning Research} 9 (61): 1871--74.
\url{http://jmlr.org/papers/v9/fan08a.html}.

\bibitem[\citeproctext]{ref-fercoq2015}
Fercoq, Olivier, Alexandre Gramfort, and Joseph Salmon. 2015. {``Mind
the Duality Gap: Safer Rules for the Lasso.''} In \emph{Proceedings of
the 32nd International Conference on Machine Learning}, edited by
Francis Bach and David Blei, vol. 37. Proceedings of Machine Learning
Research. PMLR. \url{https://proceedings.mlr.press/v37/fercoq15.html}.

\bibitem[\citeproctext]{ref-friedman2010}
Friedman, Jerome, Trevor Hastie, and Robert Tibshirani. 2010.
{``Regularization Paths for Generalized Linear Models via Coordinate
Descent.''} \emph{Journal of Statistical Software} 33 (1): 1--22.
\url{https://doi.org/10.18637/jss.v033.i01}.

\bibitem[\citeproctext]{ref-hunter2004}
Hunter, David R., and Kenneth Lange. 2004. {``A Tutorial on {MM}
Algorithms.''} \emph{The American Statistician} 58 (1): 30--37.
\url{https://doi.org/10.1198/0003130042836}.

\bibitem[\citeproctext]{ref-johnson2015}
Johnson, Tyler B, and Carlos Guestrin. 2015. {``Blitz: A Principled
Meta-Algorithm for Scaling Sparse Optimization.''} \emph{Proceedings of
the 32nd {International} {Conference} on {Machine Learning}} (Lille,
France) 37: 1171--79.
\url{https://proceedings.mlr.press/v37/johnson15.html}.

\bibitem[\citeproctext]{ref-krishnapuram2005}
Krishnapuram, Balaji, Lawrence Carin, Mário A. T. Figueiredo, and
Alexander J. Hartemink. 2005. {``Sparse Multinomial Logistic Regression:
Fast Algorithms and Generalization Bounds.''} \emph{{IEEE} Transactions
on Pattern Analysis and Machine Intelligence} 27 (6): 957--68.
\url{https://doi.org/10.1109/TPAMI.2005.127}.

\bibitem[\citeproctext]{ref-lee2014}
Lee, Jason D., Yuekai Sun, and Michael A. Saunders. 2014. {``Proximal
{Newton-Type} Methods for Minimizing Composite Functions.''}
\emph{{SIAM} Journal on Optimization} 24 (3): 1420--43.
\url{https://doi.org/10.1137/130921428}.

\bibitem[\citeproctext]{ref-lovell1963}
Lovell, Michael C. 1963. {``Seasonal Adjustment of Economic Time Series
and Multiple Regression Analysis.''} \emph{Journal of the American
Statistical Association} 58 (304): 993--1010.
\url{https://doi.org/10.1080/01621459.1963.10480682}.

\bibitem[\citeproctext]{ref-massias2020}
Massias, Mathurin, Samuel Vaiter, Alexandre Gramfort, and Joseph Salmon.
2020. {``Dual Extrapolation for Sparse Generalized Linear Models.''}
\emph{Journal of Machine Learning Research} 21 (234): 1--33.
\url{http://jmlr.org/papers/v21/19-587.html}.

\bibitem[\citeproctext]{ref-mccullagh1989}
McCullagh, Peter, and John A. Nelder. 1989. \emph{Generalized Linear
Models}. 2nd ed. Chapman; Hall.

\bibitem[\citeproctext]{ref-michie1994}
Michie, D., D. J. Spiegelhalter, and C. C. Taylor, eds. 1994.
\emph{Machine Learning, Neural and Statistical Classification}. Ellis
Horwood.

\bibitem[\citeproctext]{ref-ndiaye2017}
Ndiaye, Eugene, Olivier Fercoq, Alexandre Gramfort, and Joseph Salmon.
2017. {``Gap Safe Screening Rules for Sparsity Enforcing Penalties.''}
\emph{Journal of Machine Learning Research} 18 (128): 1--33.
\url{https://jmlr.org/papers/v18/16-577.html}.

\bibitem[\citeproctext]{ref-platt1998}
Platt, John C. 1998. {``Fast Training of Support Vector Machines Using
Sequential Minimal Optimization.''} In \emph{Advances in {Kernel
Methods}: {Support Vector Learning}}, 1st ed., edited by Bernhard
Schölkopf, Christopher J. C. Burges, and Alexander J. Smola. MIT Press.
\url{https://doi.org/10.7551/mitpress/1130.003.0016}.

\bibitem[\citeproctext]{ref-taddy2015distributed}
Taddy, Matt. 2015. {``Distributed Multinomial Regression.''} \emph{The
Annals of Applied Statistics} 9 (3): 1394--414.
\url{https://doi.org/10.1214/15-AOAS831}.

\bibitem[\citeproctext]{ref-tay2023}
Tay, J. Kenneth, Balasubramanian Narasimhan, and Trevor Hastie. 2023.
{``Elastic Net Regularization Paths for All Generalized Linear
Models.''} \emph{Journal of Statistical Software} 106 (1): 1--31.
\url{https://doi.org/10.18637/jss.v106.i01}.

\bibitem[\citeproctext]{ref-yang2024}
Yang, James, and Trevor Hastie. 2024. \emph{A Fast and Scalable
Pathwise-Solver for Group Lasso and Elastic Net Penalized Regression via
Block-Coordinate Descent}. \url{https://arxiv.org/abs/2405.08631}.

\bibitem[\citeproctext]{ref-yeoh2002}
Yeoh, Eng-Juh, Mary E. Ross, Sheila A. Shurtleff, et al. 2002.
{``Classification, Subtype Discovery, and Prediction of Outcome in
Pediatric Acute Lymphoblastic Leukemia by Gene Expression Profiling.''}
\emph{Cancer Cell} 1 (2): 133--43.
\url{https://doi.org/10.1016/S1535-6108(02)00032-6}.

\bibitem[\citeproctext]{ref-yuan2012}
Yuan, Guo-Xun, Chia-Hua Ho, and Chih-Jen Lin. 2012. {``An Improved
{GLMNET} for {L1-Regularized} Logistic Regression.''} \emph{Journal of
Machine Learning Research} 13 (64): 1999--2030.
\url{http://jmlr.org/papers/v13/yuan12a.html}.

\bibitem[\citeproctext]{ref-zeng2017}
Zeng, Yaohui, and Patrick Breheny. 2021. {``The Biglasso Package: A
Memory- and Computation-Efficient Solver for Lasso Model Fitting with
Big Data in {R}.''} \emph{The R Journal} 12 (2): 6--19.
\url{https://doi.org/10.32614/RJ-2021-001}.

\end{CSLReferences}
\end{document}